\documentclass[12pt]{article}
\usepackage{amsmath}
\usepackage{graphicx}
\usepackage{enumerate}
\usepackage{natbib}
\usepackage{url} 

\usepackage{amsmath}
\usepackage{graphicx,psfrag,epsf}
\usepackage{enumerate}
\usepackage{natbib}
\usepackage{url} 
\usepackage{amsthm}
\usepackage{booktabs}
\usepackage{framed}  
\usepackage{caption}
\usepackage{pgfplots}
\usepgfplotslibrary{dateplot}
\usetikzlibrary{snakes}
\usepackage{float}
\usepackage{amsfonts}
\usepackage{multirow, booktabs}
\usepackage{cleveref}
\usepackage{subcaption}
\usepackage[ruled,vlined]{algorithm2e}
\usepackage[T1]{fontenc}
\usepackage[utf8]{inputenc}
\usepackage{authblk}
\usepackage[multiple]{footmisc}
\usepackage{blindtext,titlefoot}
\usepackage{sectsty}
\usepackage{xcolor}
\usepackage{tikz}
\usepackage{amsmath}
\usepackage{graphicx}
\usepackage{comment}
\usepackage{amsfonts}
\usepackage{bbm}
\usepackage{amssymb}
\usepackage{setspace}
\usepackage{natbib}
\usepackage[figuresright]{rotating}
\usepackage{adjustbox}
\usepackage{pifont}
\usetikzlibrary{positioning, shapes.geometric}

\usepackage{tabularx}
\usepackage{ragged2e} 
\newcolumntype{L}{>{\RaggedRight}X} 
\usepackage{lipsum} 

\usepackage[ruled,vlined]{algorithm2e}
\usepackage{amsmath, amssymb}
\usepackage{amsfonts, multirow, epsfig, subfig}
\usepackage{graphicx, pdflscape, verbatim, enumerate, colortbl, setspace}
\usepackage{setspace, color,bm}
\usepackage[normalem]{ulem}
\usepackage{cite}
\usepackage{multirow}
\usepackage{booktabs,array}
\usepackage{url}
\usepackage{bbm}

\newtheorem{proposition}{Proposition}
\theoremstyle{definition}
\newtheorem{condition}{Condition}

\newtheorem{lemma}{Lemma}

\newtheorem{theorem}{Theorem}

\theoremstyle{definition}
\newtheorem{remark}{Remark}

\newtheorem{assumption}{Assumption}

\makeatletter
\newcommand*{\rom}[1]{\expandafter\@slowromancap\romannumeral #1@}
\makeatother

\begin{document}

\def\spacingset#1{\renewcommand{\baselinestretch}%
{#1}\small\normalsize} \spacingset{1}

\sectionfont{\bfseries\large\sffamily}%
%
\newcommand*\emptycirc[1][1ex]{\tikz\draw (0,0) circle (#1);} 
\newcommand*\halfcirc[1][1ex]{%
  \begin{tikzpicture}
  \draw[fill] (0,0)-- (90:#1) arc (90:270:#1) -- cycle ;
  \draw (0,0) circle (#1);
  \end{tikzpicture}}
\newcommand*\fullcirc[1][1ex]{\tikz\fill (0,0) circle (#1);} 

\subsectionfont{\bfseries\sffamily\normalsize}%
%


\def\spacingset#1{\renewcommand{\baselinestretch}%
{#1}\small\normalsize} \spacingset{1}

\begin{center}
    \Large \bf Randomization-Based Inference for Average Treatment Effects in Inexactly Matched Observational Studies
\end{center}


\begin{center}
  \large $\text{Jianan Zhu}^{1}$, $\text{Jeffrey Zhang}^{2}$, $\text{Zijian Guo}^{3}$, and $\text{Siyu Heng}^{*, 1}$
\end{center}

\begin{center}
   \large \textit{$^{1}$Department of Biostatistics, New York University}
\end{center}

\begin{center}
   \large \textit{$^{2}$Data Science Institute, University of Chicago}
\end{center}

\begin{center}
   \large \textit{$^{3}$Center for Data Science, Zhejiang University}
\end{center}

\let\thefootnote\relax\footnotetext{$^{*}$Corresponding Author: Siyu Heng, Department of Biostatistics, School of Global Public Health, New York University, New York, NY 10003, U.S.A. (email: siyuheng@nyu.edu). }

\bigskip

\begin{abstract}

Matching is a widely used causal inference design that aims to approximate a randomized experiment using observational data by forming matched sets of treated and control units based on similarities in their covariates. Ideally, treated units are exactly matched with controls on these covariates, enabling randomization-based inference for treatment effects as in a randomized experiment, under the assumption of no unobserved covariates. However, inexact matching often occurs, leading to residual covariate imbalance after matching. Previous matched studies have typically overlooked this issue and relied on conventional randomization-based inference, assuming that some covariate balance criteria are met. Recent research, however, has shown that this approach can introduce significant bias and proposed methods to correct for bias arising from inexact matching in randomization-based inference. These methods, however, are primarily focused on the constant treatment effect and its extensions (i.e., Fisher’s sharp null) and do not apply to average treatment effects (i.e., Neyman’s weak null). To address this gap, we introduce a new method--inverse post-matching probability weighting--for conducting randomization-based inference for average treatment effects under inexact matching. Our theoretical and simulation results indicate that, compared to conventional randomization-based inference methods, our approach significantly reduces bias and improves coverage rates in the presence of inexact matching.
\end{abstract}

\noindent%

\spacingset{1.73} 

\section{Introduction}
\label{intro}
Matching is one of the most commonly used causal inference frameworks in observational studies. It seeks to mimic a randomized experiment with observational (non-experimental) data by matching each treated unit with control units based on proximity in covariate values. Ideally, treated and control units are exactly matched on covariates so that the treatments are as-if randomly assigned within each matched set, and randomization-based inference can therefore be conducted (assuming no unobserved covariates) as in a randomized experiment (\citealp{rosenbaum2002observational, rosenbaum2020design}). However, matching is typically inexact in practice, especially when continuous or multiple covariates exist. Previous matched studies have routinely ignored inexact matching, relying on the conventional randomization-based inference as long as the matched dataset satisfies some prespecified covariate balance criteria (e.g., the absolute standardized mean difference for each covariate is less than 0.2; \citealp{rosenbaum2020design}) or passes some balance tests (e.g., \citealp{gagnon2019classification, branson2021randomization}). However, recent studies suggested that this routine practice can introduce severe bias to randomization-based inference (\citealp{guo2023statistical, pimentel2024covariate}). To correct for bias in randomization-based inference due to inexact matching, there are two existing approaches: subclassification based on the extent of post-matching covariate imbalance (\citealp{rosenbaum1988permutation}) and covariate-adaptive randomization-based inference (\citealp{pimentel2024covariate}). However, these two approaches focus primarily on the constant treatment effect model and its extensions (i.e., Fisher's sharp null). They do not apply to the average treatment effect (i.e., Neyman's weak null), which allows unlimited effect heterogeneity and does not rely on any treatment effect models.

To fill this important gap, we propose a new approach -- inverse post-matching probability weighting (IPPW) -- to conduct randomization-based inference for the sample average treatment effect in potentially inexactly matched observational studies (of which exact matching is a special case). The core idea of IPPW is to incorporate the \textit{post-matching} covariate imbalance information to re-weight both the \textit{post-matching} difference-in-means estimator and the corresponding Neyman-type variance estimator. To our knowledge, this is the first randomization-based inference method for handling inexact matching beyond the constant treatment effect model and its extensions (i.e., Fisher's sharp null); see Table~\ref{tab: contribution} for a summary. Both the theoretical and simulation results suggest that the proposed IPPW method is promising to reduce estimation bias and improve the coverage rate of confidence intervals for the sample average treatment effect under inexact matching. We have also developed an open-source \texttt{R} package \texttt{RIIM} (\textbf{R}andomization-Based \textbf{I}nference under \textbf{I}nexact \textbf{M}atching) for implementation of our methods.

\vspace{0.3cm}

\begin{table}[htbp]
\setlength{\tabcolsep}{4pt}
\centering
\captionsetup{font=small}
\caption{Applicability of some existing methods and our proposed method for randomization-based inference in matched observational studies. The column ``General Matching Design'' indicates whether the method can be applied beyond pair matching (e.g., matching with multiple controls). }
\small
\scalebox{0.87}{
\begin{tabular}{c|c c c c}
Randomization-Based Inference & Constant Effect & Average Effect & General Matching Design  & Inexact Matching \\
\hline
Rosenbaum (1987; \textit{Biometrika}) &  \checkmark & \ding{55}  & \checkmark &  \ding{55}  \\
Rosenbaum (1988; \textit{JRSSC}) &  \checkmark & \ding{55}  &  \ding{55} & \checkmark  \\
Fogarty (2018; \textit{JRSSB})   & \checkmark & \checkmark & \checkmark & \ding{55} \\
Pimentel \& Huang (2024; \textit{JRSSB})  & \checkmark & \ding{55} & \checkmark & \checkmark \\
This Work  & \checkmark & \checkmark & \checkmark & \checkmark \\

\end{tabular}
}
\label{tab: contribution}

\end{table}
\vspace{-0.6cm}
\begin{remark}
Existing work shows that, after matching, the matched dataset generally cannot be treated as if it arose from a randomized experiment with uniform treatment assignment within matched sets, even under large sample asymptotics. For example, \citet{guo2023statistical} show that the matching discrepancies can remain statistically meaningful in large pair-matched observational studies, which may invalidate standard post-matching randomization tests. \citet{savje2022inconsistency} further show that matching without replacement is in general asymptotically biased for the average treatment effect on the treated when treated units, if no bias correction after matching was performed. 
\end{remark}

\begin{remark}
    There are some existing methods for correcting for the bias associated with inexact matching under the \textit{super-population} inference framework (e.g., \citealp{abadie2011bias, guo2023statistical}). Our work, focusing on randomization-based (finite-population) causal inference under inexact matching, is intrinsically different from these existing methods in terms of the target causal estimands, the sources of randomness, and the statistical methodologies; see Remark S2 in Appendix E for details.
\end{remark}

\vspace{-0.5cm}
\section{Review: Randomization-Based Inference Under Exact Matching}\label{sec: review}

Consider a general matching design with $I$ matched sets and $N$ total units. In matched set $i \in \{1,\dots,I\}$, there are $n_{i}$ units  (so $N=\sum_{i=1}^{I}n_{i}$), among which $m_{i}$ units received the treatment, $n_{i}-m_{i}$ units received the control, and $\min \{m_{i}, n_{i}-m_{i}\}=1$. This general setting covers many widely used matching designs (\citealp{rosenbaum2002observational, rosenbaum2020design}). For example, if $n_{i}=2$ and $m_{i}=1$ for all $i$, the study design is pair matching. When $m_{i}=1$  and $n_{i}-m_{i} \geq 2$ for all $i$, the study design is matching with multiple controls. If $\min \{m_{i}, n_{i}-m_{i}\}=1$ for all $i$ is the only constraint, the matching design is full matching (\citealp{hansen2004full}). Then, for unit $j$ in matched set $i$, let $Z_{ij}$ denote its treatment indicator (i.e., $Z_{ij}=1$ if receiving the treatment and $Z_{ij}=0$ if receiving the control), $\mathbf{x}_{ij}=(x_{ij1}, \dots, x_{ijK})$ its $K$-dimensional (observed) covariates, and $Y_{ij}$ its observed outcome. Following the potential outcomes framework (\citealp{neyman1923application, rubin1974estimating}), we have $Y_{ij}=Z_{ij}Y_{ij}(1)+(1-Z_{ij})Y_{ij}(0)$, where $Y_{ij}(1)$ and $Y_{ij}(0)$ denote the potential outcome under treatment and that under control, respectively. Let $\mathbf{Z}=(Z_{11}, \dots, Z_{In_{I}})$ denote the treatment indicators vector, $\mathbf{Y}=(Y_{11}, \dots, Y_{In_{I}})$ the observed outcomes vector, and $\mathbf{X}=(\mathbf{x}_{11}, \dots, \mathbf{x}_{In_{I}})$ all the covariates information in the dataset, including the intercept term. Let $\mathcal{Z}=\{\mathbf{Z}\in\{0,1\}^{N}:\sum_{j=1}^{n_{i}}Z_{ij}=m_{i}\}$ denote the collection of all possible treatment assignments after matching. 
\begin{assumption}\label{assumption: ignorability}(No Unobserved Covariates): $(Y_{ij}(1),Y_{ij}(0))\perp\!\!\!\perp Z_{ij}\mid \mathbf{x}_{ij}$.
\end{assumption}
\begin{assumption}\label{assumption: positivity}(Positivity Assumption): $\text{pr}(Z_{ij}=1 \mid\mathbf{x}_{ij})\in [\delta, 1-\delta]$ for some fixed $\delta\in (0, 0.5]$.
\end{assumption}
Under Assumptions~\ref{assumption: ignorability} and \ref{assumption: positivity}, as well as exact matching (i.e., $\mathbf{x}_{ij}=\mathbf{x}_{ij^{\prime}}$ for all $i, j, j^{\prime}$), the treatments are as-if randomly assigned within each matched set (\citealp{rosenbaum2002observational, rosenbaum2020design, ding2024first}): 
\begin{equation}
\label{eqn: randomization assumption}
    \text{pr}(Z_{ij}=1 \mid \mathcal{Z}, \mathbf{X})=m_{i}/n_{i}, \text{ for $i=1,\dots, I, j=1,\dots, n_{i}$}.
\end{equation}
In randomization-based inference, all the potential outcomes are fixed, and the only probability distribution that enters into inference is the randomization assumption (\ref{eqn: randomization assumption}), which holds by design in a randomized experiment or an exactly matched observational study (\citealp{rosenbaum2002observational, rosenbaum2020design, fogarty2018mitigating, basse2024randomization}). Then, researchers can apply randomization-based inference to various causal estimands. For example, for inferring the sample average treatment effect $\lambda=N^{-1}\sum_{i=1}^{I}\sum_{j=1}^{n_{i}}\{Y_{ij}(1)-Y_{ij}(0)\}$, researchers can adopt the commonly used difference-in-means estimator $\widehat{\lambda}$ (\citealp{neyman1923application, li2017general, fogarty2018mitigating}), where $\widehat{\lambda}=\sum_{i=1}^{I}(n_{i}/N)\widehat{\lambda}_{i}$, with $\widehat{\lambda}_{i}=\sum_{j=1}^{n_{i}}\{Z_{ij}Y_{ij}/m_{i}-(1-Z_{ij})Y_{ij}/(n_{i}-m_{i})\}$. Under exact matching (in which (\ref{eqn: randomization assumption}) holds), the post-matching difference-in-means estimator $\widehat{\lambda}$ is an unbiased estimator for $\lambda$, and an asymptotically valid variance estimator for $\widehat{\lambda}$ can be derived using the methods in \citet{fogarty2018mitigating} to facilitate randomization-based inference. 

\vspace{-0.5cm}
 
\section{Our Proposed Approach: Randomization-Based Inference via Inverse Post-Matching Probability Weighting}\label{sec: IPPW method}
\subsection{The IPPW estimator with oracle post-matching treatment assignment probabilities}\label{sec: oracle}

In practice, we rarely expect covariates to be exactly matched between the treated and control units, i.e., $\mathbf{x}_{ij}= \mathbf{x}_{ij^{\prime}}$ may not hold in practice. Actually, previous work \citep{savje2022inconsistency, guo2023statistical} has shown that such post-matching covariate discrepancies can substantially bias downstream randomization-based inference, even when the sample size goes to infinity. To address this issue, we propose a randomization-based inference method for inferring the sample average treatment effect $\lambda$ in inexactly matched observational studies. Our approach consists of two components. First, we propose a new randomization-based estimator, called the inverse post-matching probability weighting (IPPW) estimator, to correct for bias due to inexact matching by re-weighting the \textit{post-matching} difference-in-means estimator according to discrepancies of \textit{post-matching} treatment assignment probabilities. Second, we derive a new variance estimator for the proposed IPPW estimator, which is the first randomization-based and model-free variance estimator that is (asymptotically) valid under inexact matching. In Section~\ref{sec: oracle}, we derive the IPPW estimator and its variance estimator under oracle propensity scores, which serves as an intermediate step toward deriving the IPPW estimator and its variance estimator with estimated propensity scores in Section~\ref{sec: M-estimation}.

Specifically, we let $e_{ij}$ denote the propensity score of unit $j$ in matched set $i$, and $\mathbf{e}=(e_{11},\dots,e_{In_{I}})$ the propensity score vector. Then, we can express each post-matching treatment assignment probability $p_{ij}=\text{pr}(Z_{ij}=1\mid\mathcal{Z},\mathbf{X})$ in terms of $(e_{i1}, \dots, e_{in_{i}})$. Specifically, for a matched set $i$ with one treated and one or multiple controls (i.e., $m_{i}=1$), following the arguments in \citet{pimentel2024covariate}, we have 
$p_{ij} =\text{pr}(Z_{ij}=1  \mid \mathcal{Z}, \mathbf{X})=\frac{\text{odds}\{e_{ij}\}}{\sum_{j'=1}^{n_i}\text{odds}\{e_{ij'}\}}.$
If the matched set $i$ has one control and multiple treated units (i.e., $n_{i}-m_{i}=1$ and $m_{i}>1$), we have $p_{ij}=\text{pr}(Z_{ij}=1  \mid \mathcal{Z}, \mathbf{X})=1-\frac{\text{odds}\{1-e_{ij}\}}{\sum_{j'=1}^{n_i}\text{odds}\{1-e_{ij'}\}}.$
Then, the oracle form of the proposed IPPW estimator (under the oracle $p_{ij}$) is defined as
\begin{equation}\label{eqn: IPPW oracle}
    \widehat{\lambda}_{*} = \sum_{i=1}^{I}\frac{n_{i}}{N}\widehat{\lambda}_{*, i},  \text{ where }  \widehat{\lambda}_{*, i}=\frac{1}{n_{i}}\sum_{j=1}^{n_{i}}\Big(\frac{Z_{ij}Y_{ij}}{p_{ij}}-\frac{(1-Z_{ij})Y_{ij}}{1-p_{ij}}\Big).
\end{equation}
As shown in (\ref{eqn: IPPW oracle}), the idea of IPPW is to re-weight the \textit{post-matching} difference-in-means estimator \textit{within each matched set} based on discrepancies in $p_{ij}$ due to inexact matching. Under exact matching, each $p_{ij}=m_{i}/n_{i}$, so the IPPW estimator $\widehat{\lambda}_{*}$ reduces to the post-matching difference-in-means estimator $\widehat{\lambda}$.

\begin{proposition}\label{prop: unbiasdness} Under Assumptions~\ref{assumption: ignorability} and \ref{assumption: positivity}, we have $E(\widehat{\lambda}_{*}\mid \mathcal{Z})=\lambda$.
\end{proposition}
Proposition~\ref{prop: unbiasdness} shows that the (oracle) IPPW estimator $\widehat{\lambda}_{*}$ is unbiased for estimating $\lambda$, even under inexact matching. The proofs of all the theoretical results are in Appendix B in the supplemental material. We next derive an asymptotically valid variance estimator for $\widehat{\lambda}_{*}$, by adapting the approach of \citet{kang2016full} and \citet{fogarty2018mitigating} to the inexact matching case. Specifically, note that $\widehat{\lambda}_{*}=I^{-1}\sum_{i=1}^{I}V_{i}$, where $V_i=w_{i} \widehat{\lambda}_{*,i}$ with $w_i=In_i/N$ representing the ratio between the matched set size $n_{i}$ and the average size $N/I$. Then, we propose the following variance estimator for $\widehat{\lambda}_{*}$: 
\begin{equation}\label{eqn: oracle variance estimator}
    S^2_{*}=\frac{1}{I(I-1)}\sum_{i=1}^I\big\{V_i-\widehat{\lambda}_*\big\}^2, \text{ where $\widehat{\lambda}_{*}=\frac{1}{I}\sum_{i=1}^{I}V_{i}$.}
\end{equation}
Under some mild regularity conditions (i.e., Conditions 1--2 stated below), Theorem~\ref{thm: CI of the SATE true} shows that the confidence interval based on $\widehat{\lambda}_{*}$ and $S^2_{*}$, as well as a finite-population central limit theorem, is asymptotically valid without any modeling assumptions on treatment effects, even under inexact matching. In Appendix B.2, we extend the variance estimator $S^{2}_{*}$ to accommodate covariate adjustment to further improve efficiency, as well as discuss other options of variance estimators for $\widehat{\lambda}_{*}$. 

\begin{condition}\label{condition: bounded main}(Bounded Matched Sets and Bounded Outcomes): There exists some constant $C_{1}<\infty$ such that $n_i \leq C_1$ for all $i$. Also, there exists a constant $M<\infty$ such that $|Y_{ij}|\leq M$ for all $i, j$.
\end{condition}

\begin{condition}\label{condition: convergence main} (Convergence of Finite-Population Means): For each matched set $i$, we define $\mu_{i}=E(\widehat{\lambda}_{*,i}\mid\mathcal{Z})$ and $\nu_{i}^2=\text{var} (\widehat{\lambda}_{*,i}\mid\mathcal{Z})$, we have: (i) $I^{-1}\sum_{i=1}^{I}w_{i}\mu_{i}$, $I^{-1}\sum_{i=1}^{I}w_{i}^2\mu_{i}$, and $I^{-1}\sum_{i=1}^{I}w_{i}^2\mu_{i}^2$ converge to some finite values; (ii) $I^{-1}\sum_{i=1}^{I}w_{i}^2\nu_{i}^2$ converges to some finite positive value.
\end{condition}

\begin{theorem}\label{thm: CI of the SATE true}
 Consider the confidence interval $CI^{\lambda}_{*}=[\widehat{\lambda}_{*} - \Phi^{-1}(1-\alpha/2) \times S_{*}, \widehat{\lambda}_{*} + \Phi^{-1}(1-\alpha/2) \times S_{*}]$, where $\alpha\in (0,0.5)$ is some prespecified level, $\Phi$ is the distribution function of $N(0,1)$, and $S_{*}=\sqrt{S_{*}^{2}}$. Assuming independence of treatment assignments across matched sets, Assumptions~\ref{assumption: ignorability} and \ref{assumption: positivity}, and Conditions~\ref{condition: bounded main}--\ref{condition: convergence main}, we have $\lim_{I\rightarrow \infty} \text{pr}\big(\lambda\in CI^{\lambda}_{*} \mid \mathcal{Z}\big) \geq 1-\alpha$. 
\end{theorem}
The two key ingredients for establishing Theorem~\ref{thm: CI of the SATE true} are the asymptotic normality of $\widehat{\lambda}_{*}$ (see Appendix B.3) and the asymptotic conservativeness of $S_{*}^{2}$ (see Appendix B.2). The confidence interval in Theorem~\ref{thm: CI of the SATE true} is sharp when treatment effects are constant across all units in the sample and all strata have the same sample size. In general, however, it is not sharp under treatment effect heterogeneity. To our knowledge, Theorem~\ref{thm: CI of the SATE true} gives the first (asymptotically) valid confidence interval for the sample average treatment effect $\lambda$ under inexact matching, which is universally applicable for general matching designs.

\vspace{-0.3cm}
\subsection{The IPPW estimator with estimated post-matching treatment assignment probabilities}\label{sec: M-estimation}

In practice, the true post-matching probabilities $p_{ij}$ involved in $\widehat{\lambda}_{*}$ and $S_{*}^{2}(Q)$ are unknown. A general and sensible strategy for handling this is the commonly adopted ``plug-in'' strategy (\citealp{rosenbaum1987model, ding2024first, pimentel2024covariate}): we replace each oracle $p_{ij}$ with its estimate $\widehat{p}_{ij}$, obtained by replacing the oracle propensity scores $e_{ij}$ in $p_{ij}$ with the estimated propensity scores $\widehat{e}_{ij}$. Then, we let $\widehat{\lambda}_{\diamond}$ denote the IPPW estimator obtained by replacing each $p_{ij}$ in $\widehat{\lambda}_{*}$ with $\widehat{p}_{ij}$: 
\begin{equation*}\label{eqn: IPPW with estimated PS}
    \widehat{\lambda}_{\diamond} = \sum_{i=1}^{I}(n_{i}/N)\widehat{\lambda}_{\diamond, i},  \text{ where }  \widehat{\lambda}_{\diamond, i}=\frac{1}{n_{i}}\sum_{j=1}^{n_{i}}\Big(\frac{Z_{ij}Y_{ij}}{\widehat{p}_{ij}}-\frac{(1-Z_{ij})Y_{ij}}{1-\widehat{p}_{ij}}\Big).
\end{equation*}

\begin{proposition}\label{prop: consistency}
    Under Assumptions~\ref{assumption: ignorability} and \ref{assumption: positivity}, as well as some regularity conditions specified in Appendix B.4, we have $\widehat{\lambda}_{\diamond}\xrightarrow{a.s.} \lambda$ as $I\rightarrow \infty$ (i.e., the $\widehat{\lambda}_{\diamond}$ is strongly consistent).
\end{proposition}

The convergence rate of $\widehat{\lambda}_{\diamond}$ depends on the convergence rate of the propensity score model. More importantly, Proposition~\ref{prop: consistency} establishes that $\widehat{\lambda}_{\diamond}$ is consistent for estimating $\lambda$, even though the true post-matching treatment assignment probabilities are unknown and must be estimated from the data. To construct a variance estimator of $\widehat{\lambda}_{\diamond}$ and enable inference, we consider two strategies. The first strategy adopts the commonly used ``plug-in'' approach \citep{pimentel2024covariate}: we replace each oracle $p_{ij}$ with its estimate $\widehat{p}_{ij}$ in the variance estimator formula (\ref{eqn: oracle variance estimator}) derived under oracle $p_{ij}$, and denote the resulting variance estimator by $S_{\diamond}^{2}$. Then, the proposed randomization-based confidence interval for $\lambda$ can be expressed as $CI^{\lambda}_{\diamond}=[\widehat{\lambda}_{\diamond} - \Phi^{-1}(1-\alpha/2) \times S_{\diamond}, \widehat{\lambda}_{\diamond} + \Phi^{-1}(1-\alpha/2) \times S_{\diamond}]$, where $S_{\diamond}=\sqrt{S^{2}_{\diamond}}$. In Appendix B.5, we provide a detailed characterization of the finite-sample bias between the plug-in confidence interval $CI^{\lambda}_{\diamond}$ and the oracle confidence interval $CI^{\lambda}_{*}$. In the simulation studies presented in Appendix C, this plug-in strategy performs well and can substantially improve coverage rates compared with the conventional inference methods in matched studies, in both parametric and nonparametric propensity score settings.

The second strategy for obtaining the variance estimator of $\widehat{\lambda}_{\diamond}$ is to extend the finite-population M-estimation theory \citep{xu2021mestimator} to matched data, which takes the propensity score estimation and treatment effect estimation as a joint process to account for uncertainties from both. Unlike M-estimators under the infinite-population paradigm, our approach considers an infinite sequence of ever-large finite populations, in each of which the randomness arises solely from the treatment assignment vector $\mathbf{Z}$. Under full matching, subjects within each finite population are grouped into mutually independent matched sets. The propensity score estimation utilizes all the units within the finite population. As such, we adopt new notations. Specifically, for a given finite population, let $\mathbf{O}_i=\{\mathbf{Y}_i,\mathbf{Z}_i,\mathbf{X}_i\}$ denote the observed data corresponding to matched set $i=1,\cdots,I$, where $\mathbf{Y}_i=(Y_{i1},\dots, Y_{in_{i}})$, $\mathbf{Z}_i=(Z_{i1}, \dots, Z_{in_{i}})$, and $\mathbf{X}_i=(\mathbf{x}_{i1}, \dots, \mathbf{x}_{in_{i}}) \in \mathbbm{R}^{p\times{n_i}}$. Let $\theta_0$ denote the p-dimensional true parameter vector of the propensity score model, and let $\nu_0=\frac{1}{N}\sum_{i=1}^{n_{i}}Y_{ij}(1)$ and $\nu_0'=\frac{1}{N}\sum_{i=1}^{n_{i}}Y_{ij}(0)$ represent the finite-population means of the potential outcomes under treatment and control, respectively. For matched set $i$, denote $\boldsymbol\psi^{full}(\mathbf{O}_i,\theta,\nu,\nu')$ the corresponding estimating equations (see Appendix B.6 for specific forms in specific examples). For the matched data, under the finite-population M-estimation framework, the estimates $(\widehat{\theta},\widehat{\nu},\widehat{\nu}')$ for $(\theta_0,\nu_0,\nu'_0)$ are obtained by solving the following finite-population estimating equations in terms of $(\theta,\nu,\nu') \in \mathbbm{R}^{p+2}$:
\begin{equation}\label{eqn: M-estimator}
    \frac{1}{I}\sum_{i=1}^{I}\boldsymbol\psi^{full}(\mathbf{O}_i,\theta,\nu,\nu')=\mathbf{0}_{(p+2)\times1}.
\end{equation}
We now aim to estimate the asymptotic variance of ($\widehat{\theta},\widehat{\nu},\widehat{\nu}'$). To this end, we apply the sandwich variance construction for finite-population M-estimators proposed in \citet{xu2021mestimator} to the matched dataset. Specifically, we construct the following finite-population asymptotic variance matrix, which is analogous to the asymptotic variance-covariance matrix in the traditional infinite population version of M-estimation:
\begin{equation*}
V(\theta_0,\nu_0.\nu'_0)=A(\theta_0,\nu_0.\nu'_0)^{-1}B(\theta_0,\nu_0.\nu'_0)[A(\theta_0,\nu_0.\nu'_0)^{-1}]^T \in \mathbbm{R}^{(p+2)\times(p+2)},
\end{equation*}
where $A(\cdot)$ and $B(\cdot)$ are defined as follows:
\begin{align*}
    &A(\theta_0,\nu_0.\nu'_0)=\lim_{I \rightarrow \infty} \frac{1}{I} \sum_{i=1}^{I}E\Big[-\nabla_{\theta,\nu,\nu',}\boldsymbol\psi^{full}(\mathbf{O}_i,\theta_0,\nu_0,\nu'_0)\Big]\in \mathbbm{R}^{(p+2)\times(p+2)}, \\    
    &B(\theta_0,\nu_0.\nu'_0)=\lim_{I \rightarrow \infty} \frac{1}{I} \sum_{i=1}^{I}E\Big[\boldsymbol\psi^{full}(\mathbf{O}_i,\theta_0,\nu_0,\nu'_0)\boldsymbol\psi_i^{full}(\mathbf{O}_i,\theta_0,\nu_0,\nu'_0)^T\Big]\in \mathbbm{R}^{(p+2)\times(p+2)}.
\end{align*}
For the detailed expression of $V(\theta_0,\nu_0.\nu'_0)$, please see Appendix B.6. While the variance matrix $V(\theta_0,\nu_0.\nu'_0)$ depends on the oracle knowledge of the true parameters and the moments of the derivatives of the $\boldsymbol{\psi}$ functions, it can be consistently estimated by substituting the estimated (sample-based) parameters and sample-based moments. We denote this sample-based estimator of $V(\theta_0,\nu_0.\nu'_0)$ by $\widehat{V}(\widehat{\theta},\widehat{\nu},\widehat{\nu}')\in \mathbbm{R}^{(p+2)\times(p+2)}$. Recall that the sample average treatment effect $\lambda=\nu_{0}-\nu_{0}^{\prime}$ and the plug-in estimator $\widehat{\lambda}_{\diamond}=\widehat{\nu}-\widehat{\nu}^{\prime}$. Therefore, we construct the variance estimator $S_{\mathcal{M}}^2$ for $\widehat{\lambda}_{\diamond}$ as:
\begin{equation*}
    S_{\mathcal{M}}^2=\frac{1}{I}\widehat{V}(\widehat{\theta},\widehat{\nu},\widehat{\nu}')_{(p+1),(p+1)}+\frac{1}{I}\widehat{V}(\widehat{\theta},\widehat{\nu},\widehat{\nu}')_{(p+2),(p+2)}-\frac{2}{I}\widehat{V}(\widehat{\theta},\widehat{\nu},\widehat{\nu}')_{(p+1),(p+2)}.
\end{equation*}
The detailed expressions of $\widehat{V}(\widehat{\theta},\widehat{\nu},\widehat{\nu}')_{(p+1),(p+1)}$, $\widehat{V}(\widehat{\theta},\widehat{\nu},\widehat{\nu}')_{(p+2),(p+2)}$, and $\widehat{V}(\widehat{\theta},\widehat{\nu},\widehat{\nu}')_{(p+1),(p+2)}$ can be found in Appendix B.6. By adapting the finite-population M-estimation theory \citep{xu2021mestimator} to matched observational studies, we have the following asymptotic normality result for the plug-in estimator $\widehat{\lambda}_{\diamond}$:
\begin{theorem}[Asymptotic Normality of the Plug-in IPPW Estimator]\label{thm: asymptotic normality of plug-in}
Denote $\boldsymbol{\gamma}_{0}=(\theta_0,\nu_0,\nu'_0)$ and $\widehat{\boldsymbol{\gamma}}=(\widehat{\theta},\widehat{\nu},\widehat{\nu}')$. Under Assumptions~\ref{assumption: ignorability} and \ref{assumption: positivity}, independence of treatment assignments across matched sets, and the regularity conditions specified in Appendix B.6 (i.e., the common regularity conditions for finite-population M-estimation), we have $\sqrt{I}(\widehat{\boldsymbol{\gamma}}-\boldsymbol{\gamma}_{0})^T\xrightarrow{d}\mathcal{N}(\mathbf{0}, V_{fp}) \ \text{as} \ I\rightarrow\infty$, where $V_{fp}$ denotes the true asymptotic variance matrix of $\widehat{\boldsymbol{\gamma}}$.
\end{theorem}
Also, we have the following validity guarantee for the variance estimator for the plug-in estimator $\widehat{\lambda}_{\diamond}$:
\begin{theorem}[Valid Variance Estimator for the Plug-in IPPW Estimator]\label{thm: valid variance for the plug-in estimator} Under the setup in Theorem~\ref{thm: asymptotic normality of plug-in}, as $I\rightarrow \infty$, we have (i) $\widehat{V}(\widehat{\theta},\widehat{\nu},\widehat{\nu}')\xrightarrow{p}V(\theta_0,\nu_0.\nu'_0)$, and (ii) $V(\theta_0,\nu_0.\nu'_0)\succcurlyeq V_{fp}$, where the matrix inequality ``$\succcurlyeq$'' denotes that $V(\theta_0,\nu_0.\nu'_0)-V_{fp}$ is a positive semidefinite matrix. Combing (i) and (ii), we have $\lim_{I\rightarrow \infty} S_{\mathcal{M}}^2/\text{var}(\widehat{\lambda}_{\diamond})\xrightarrow{p} b$ for some constant $b>1$.
\end{theorem}
Combining Theorems~\ref{thm: asymptotic normality of plug-in}
and \ref{thm: valid variance for the plug-in estimator}, we have the following theoretical guarantee for randomization-based inference based on the plug-in IPPW estimator $\widehat{\lambda}_{\diamond}$ (i.e., the IPPW etimator with estimated propensity scores).
\begin{theorem}[Valid Confidence Interval based on the Plug-in IPPW Estimator]\label{thm: CI of the SATE estimated}
    Consider the confidence interval $CI^{\lambda}_{\mathcal{M}}=[\widehat{\lambda}_{\diamond} - \Phi^{-1}(1-\alpha/2) \times S_{\mathcal{M}}, \widehat{\lambda}_{\diamond} + \Phi^{-1}(1-\alpha/2) \times S_{\mathcal{M}}]$, where $\widehat{\lambda}_{\diamond}$ is the IPPW estimator with estimated propensity scores, $S_{\mathcal{M}}=\sqrt{S^{2}_{\mathcal{M}}}$, and $\alpha\in(0,0.5)$. Under the setup in Theorems~\ref{thm: asymptotic normality of plug-in} and \ref{thm: valid variance for the plug-in estimator}, we have $\lim_{I \rightarrow \infty}\text{pr}(\lambda \in CI^{\lambda}_{\mathcal{M}}|\mathcal{Z})\geq1-\alpha$ for any level $\alpha\in(0,0.5)$.
\end{theorem}

To our knowledge, Theorem 2 provides the first rigorously valid confidence interval for the sample average treatment effect $\lambda$ in inexactly matched observational studies, without relying on known propensity scores. A simulation study assessing the performance of the proposed IPPW estimator is presented in Appendix C.1. The results show that the proposed IPPW estimator (with both oracle and estimated propensity scores) substantially reduces estimation bias and improves coverage rates, compared with the existing estimators for $\lambda$ (e.g., the conventional matching estimator in randomization-based inference). A data application is also provided in Appendix D.

\vspace{-0.5cm}

\section{Incorporating Outcome Modeling into the IPPW Estimator}\label{sec: AIPPW}

In this section, we show that our framework has the flexibility of incorporating outcome models into the IPPW estimator (in addition to weighting), which will be referred to as the augmented IPPW (AIPPW) estimator. We provide both theoretical insights (e.g., asymptotic biases) and numerical insights (e.g., finite-sample biases) to illustrate when incorporating outcome models may be preferable to the original IPPW estimator. Throughout this section, the outcome models may be estimated either from the study sample itself or from external (out-of-sample) data, yielding an estimated potential outcome function for the treated (denoted as $\widehat{g}_1(\mathbf{x}_{ij})$) and that for the control (denoted as $\widehat{g}_0(\cdot)$). These models may be specified using any parametric or non-parametric approach. By extending the form of the augmented inverse probability weighting (AIPW) estimator in the super-population, unmatched settings (\citealp{robins1994estimation}) to randomization-based (finite-population) inference for matched observed studies, we propose the following form of the AIPPW estimator under the oracle post-matching treatment assignment probabilities $p_{ij}$ (henceforth termed the \textit{oracle AIPPW estimator}):
\begin{equation*}\label{eqn: AIPPW oracle}
    \widehat{\lambda}_{\dagger}=\sum_{i=1}^{I}\frac{n_i}{N}\widehat{\lambda}_{\dagger,i}, \text{ where } \widehat{\lambda}_{\dagger, i}=\frac{1}{n_i}\sum_{j=1}^{n_i}\bigg\{\frac{Z_{ij}}{p_{ij}}
    \Big(Y_{ij}-\widehat{g}_1(\mathbf{x}_{ij})\Big)- 
    \frac{(1-Z_{ij})}{1-p_{ij}}
    \Big(Y_{ij}-\widehat{g}_0(\mathbf{x}_{ij})\Big)+\widehat{g}_1(\mathbf{x}_{ij})-\widehat{g}_0(\mathbf{x}_{ij})\bigg\}.
\end{equation*}
As seen from the expression above, the AIPPW estimator is analogous to the conventional augmented inverse probability weighting (AIPW) estimator in that it combines outcome modeling with weighting. The key distinction is that the weighting component in AIPPW is tailored to the matched design and accommodates inexact matching. Specifically, AIPPW estimators weight observations by the post-matching treatment assignment probabilities $p_{ij}$ rather than the pre-matching propensity scores. Moreover, it is straightforward to show that the oracle AIPPW estimator is unbiased for estimating $\lambda$, even when $\widehat{g}_1(\mathbf{x}_{ij})$ and $\widehat{g}_{0}(\mathbf{x}_{ij})$ are misspecified (see Proposition S3 and its proof in Appendix B.8). That is, when $p_{ij}$ is known, such as in matched or finely stratified randomized experiments (\citealp{imai2008variance, fogarty2018regression, fogarty2018mitigating}), both the oracle IPPW estimator $\widehat{\lambda}_{*}$ and the oracle AIPPW estimator $\widehat{\lambda}_{\dagger}$ are unbiased estimators of $\lambda$. In this setting, variance, or equivalently efficiency, becomes the primary consideration when choosing between $\widehat{\lambda}_{*}$ and $\widehat{\lambda}_{\dagger}$. As suggested by the literature on matched or stratified randomized experiments (e.g., \citealp{fogarty2018regression, liu2020regression}), using suitable outcome models $\widehat{g}_1(\mathbf{x}_{ij})$ and $\widehat{g}_{0}(\mathbf{x}_{ij})$ can reduce the variance of the estimator and thereby improve the efficiency of randomization-based inference.  

However, unlike matched randomized experiments, in matched observational studies, $p_{ij}$ are typically unknown and must be estimated from the observed data. Following principles similar to those in Section 3.2 in the main text, we propose the following AIPPW estimator based on the estimated post-matching treatment assignment probabilities $\widehat{p}_{ij}$, which we refer to as the \textit{plug-in AIPPW estimator}:
\begin{equation*}\label{eqn: AIPPW estimate}
\widehat{\lambda}_{\ddagger}=\sum_{i=1}^{I}\frac{n_i}{N}\widehat{\lambda}_{\ddagger,i}, \text{ where } \widehat{\lambda}_{\ddagger, i}=\frac{1}{n_i}\sum_{j=1}^{n_i}\bigg\{\frac{Z_{ij}}{\widehat{p}_{ij}}
    \Big(Y_{ij}-\widehat{g}_1(\mathbf{x}_{ij})\Big)- 
    \frac{(1-Z_{ij})}{1-\widehat{p}_{ij}}
    \Big(Y_{ij}-\widehat{g}_0(\mathbf{x}_{ij})\Big)+\widehat{g}_1(\mathbf{x}_{ij})-\widehat{g}_0(\mathbf{x}_{ij})\bigg\}.
\end{equation*}

For the plug-in AIPPW estimator $\widehat{\lambda}_{\ddagger}$, we have the following double robustness property:

\begin{proposition}\label{prop: consistency_AIPPW}
    Under Assumptions~\ref{assumption: ignorability} and \ref{assumption: positivity}, independence of post-matching treatment assignments across matched sets, and some mild regularity conditions specified in Appendix B.4, we have $\widehat{\lambda}_{\ddagger}\xrightarrow{a.s.} \lambda$ as $I\rightarrow \infty$ (i.e., the $\widehat{\lambda}_{\ddagger}$ is strongly consistent) provided that either one of the following two consistency conditions hold as $I\rightarrow \infty$: (i) $\widehat{p}_{ij}\xrightarrow{a.s.}p_{ij}$; or (ii) $\widehat{g}_1(\mathbf{x}_{ij})\xrightarrow{a.s.} Y_{ij}(1)$ and $\widehat{g}_0(\mathbf{x}_{ij})\xrightarrow{a.s.} Y_{ij}(0)$.
\end{proposition}
In other words, Proposition~\ref{prop: consistency_AIPPW} shows that $\widehat{\lambda}_{\ddagger}$ remains consistent if either the propensity score model or the outcome model is correctly specified. However, this double robustness does not necessarily imply that the AIPPW estimator should always be preferred to the IPPW estimator, even as $I\rightarrow \infty$. For example, if neither the propensity score model nor the outcome model is correctly specified, incorporating misspecified outcome models may increase both the asymptotic bias and the finite-sample estimation biases. Specifically, suppose that the propensity score model is misspecified, so that $\widehat{p}_{ij} \xrightarrow{a.s.} \widetilde{p}_{ij}$ for some $\widetilde{p}_{ij}\neq p_{ij}$, and that the outcome models are also misspecified, so that $\widehat{g}_{1}(\mathbf{x}_{ij})\xrightarrow{a.s.} \widetilde{g}_1(\mathbf{x}_{ij})$ for some $\widetilde{g}_1(\mathbf{x}_{ij})\neq Y_{ij}(1)$ and $\widehat{g}_{0}(\mathbf{x}_{ij})\xrightarrow{a.s.} \widetilde{g}_0(\mathbf{x}_{ij})$ for some $\widetilde{g}_0(\mathbf{x}_{ij})\neq Y_{ij}(0)$. Then, we can derive the following expressions for the asymptotic estimation biases, relative to the sample average treatment effect $\lambda$, of the IPPW estimator $\widehat{\lambda}_{\diamond}$ and the AIPPW estimator $\widehat{\lambda}_{\ddagger}$:
  \begin{align*} 
             \lim_{I \to \infty}\widehat{\lambda}_{\diamond}-\lambda &\overset{\mathrm{a.s.}}{=}\lim_{I \to \infty}\widehat{\lambda}_{\diamond}-\lim_{I \to \infty}\widehat{\lambda}_*\\
            &=\lim_{I \to \infty}\Bigg\{\frac{1}{N}\sum_{i=1}^{I}\sum_{j=1}^{n_i}\bigg[Z_{ij}Y_{ij}\bigg(\frac{1}{\widetilde{p}_{ij}}-\frac{1}{p_{ij}}\bigg)-(1-Z_{ij})Y_{ij}\bigg(\frac{1}{1-\widetilde{p}_{ij}}-\frac{1}{1-p_{ij}}\bigg)\bigg]\Bigg\},
    \end{align*}
    and
      \begin{align*}      
           \lim_{I \to \infty}\widehat{\lambda}_{\ddagger}-\lambda&\overset{\mathrm{a.s.}}{=}\lim_{I \to \infty}\Bigg\{\frac{1}{N}\sum_{i=1}^{I}\sum_{j=1}^{n_i}\bigg[Z_{ij}Y_{ij}\bigg(\frac{1}{\widetilde{p}_{ij}}-\frac{1}{p_{ij}}\bigg)-(1-Z_{ij})Y_{ij}\bigg(\frac{1}{1-\widetilde{p}_{ij}}-\frac{1}{1-p_{ij}}\bigg)\\
            &\quad \quad \quad \quad \quad \quad \quad \quad \quad \quad \quad \quad +(\widetilde{p}_{ij}-Z_{ij})\bigg(\frac{\widetilde{g}_1(\mathbf{x}_{ij})}{\widetilde{p}_{ij}}+\frac{\widetilde{g}_0(\mathbf{x}_{ij})}{1-\widetilde{p}_{ij}}\bigg)\bigg]\Bigg\}.
    \end{align*}
That is, the asymptotic bias of $\widehat{\lambda}_{\diamond}$ and that of $\widehat{\lambda}_{\ddagger}$ differ by $\{\lim_{I \to \infty}\widehat{\lambda}_{\ddagger}-\lambda\}-\{\lim_{I \to \infty}\widehat{\lambda}_{\diamond}-\lambda\} = \lim_{I\rightarrow \infty}\big\{\frac{1}{N}\sum_{i=1}^{I}\sum_{j=1}^{n_i}(\widetilde{p}_{ij}-Z_{ij})\big(\frac{\widetilde{g}_1(\mathbf{x}_{ij})}{\widetilde{p}_{ij}}+\frac{\widetilde{g}_0(\mathbf{x}_{ij})}{1-\widetilde{p}_{ij}}\big)\big\}$. In particular, if both $\{\lim_{I \to \infty}\widehat{\lambda}_{\ddagger}-\lambda\}-\{\lim_{I \to \infty}\widehat{\lambda}_{\diamond}-\lambda\}$ and $\lim_{I \to \infty}\widehat{\lambda}_{\diamond}-\lambda$ are positive (or negative), we have $|\lim_{I \to \infty}\widehat{\lambda}_{\ddagger}-\lambda|>|\lim_{I \to \infty}\widehat{\lambda}_{\diamond}-\lambda|$, i.e., the asymptotic absolute bias of the AIPPW estimator $\widehat{\lambda}_{\ddagger}$ is larger than that of the IPPW estimator $\widehat{\lambda}_{\diamond}$. In addition, the simulation results on the finite-sample biases of the AIPPW and IPPW estimators (in Appendix C.3) are consistent with the above theoretical comparisons of asymptotic biases. Specifically, the simulation studies in Appendix C.3 demonstrate that outcome-model augmentation can either increase or decrease the magnitude of finite-sample bias of the proposed IPPW estimation framework, depending on how well the outcome models are specified. 



\vspace{-0.5cm}

\section*{Acknowledgement}

The authors thank Rebecca Betensky, Colin Fogarty, Hyunseung Kang, Samuel Pimentel, and Bo Zhang for the helpful discussions and comments. The work of Siyu Heng was supported in part by NIH Grant R21DA060433 and an NYU Research Catalyst Prize. The work of Zijian Guo was partly supported by NIH R01LM013614, NIH R01AG086379, and NSF DMS 2413107, when he was a faculty member at the Rutgers University in the United States.

\vspace{-0.5cm}

\section*{Supplementary Material}
The online supplementary material includes all the technical proofs, additional theoretical and simulation results, additional discussions and remarks, and a data application using our proposed methods.

\renewcommand{\theproposition}{S\arabic{proposition}}
\renewcommand{\thelemma}{S\arabic{lemma}}
\renewcommand{\thecondition}{S\arabic{condition}}
\renewcommand{\theremark}{S\arabic{remark}}
\renewcommand{\thetable}{S\arabic{table}}
\renewcommand{\thetheorem}{S\arabic{theorem}}

\bibliographystyle{apalike}
\bibliography{references}

\def\spacingset#1{\renewcommand{\baselinestretch}%
{#1}\small\normalsize} \spacingset{1}

\begin{center}
    \Large \bf Supplementary Material for ``Randomization-Based Inference for Average Treatment Effects in Inexactly Matched Observational Studies''
\end{center}

\spacingset{1.73}
\section*{Appendix A: Extension of the Proposed Method to Instrumental Variable Studies}\label{sec: Appendix IV}

In the main text, we focus on handling overt bias in randomization-based inference due to inexact matching on observed covariates. In many settings, hidden bias due to unobserved covariates may also exist and cannot be directly adjusted for by the method proposed in Section 3 in the main text. To address this limitation, we show how to combine the IPPW method with the instrumental variable (IV) method to simultaneously address the concerns for observed and unobserved covariates. Following the classic framework of matched IV studies (\citealp{baiocchi2010building, rosenbaum2020design}), we still consider the notations used in Sections 2 and 3 in the main text, with the only differences being that the $Z_{ij}$ is the indicator of a binary observational instrumental variable, and $\mathbf{x}_{ij}$ represent the observed IV-outcome covariates. Under the IV ignorability assumption (i.e., the IVs are independent of the potential outcomes conditional on $\mathbf{x}_{ij}$ being adjusted for), the IV randomization assumption holds conditional on exact matching on the observed IV-outcome covariates (i.e., $\mathbf{x}_{ij}=\mathbf{x}_{ij^{\prime}}$) \citep{baiocchi2010building, kang2016full}. We let $D_{ij}=Z_{ij}D_{ij}(1)+(1-Z_{ij})D_{ij}(0)$ denote the observed value of the actual treatment of interest, where $D_{ij}(1)$ and $D_{ij}(0)$ denote the potential treatment value under $Z_{ij}=1$ and that under $Z_{ij}=0$, respectively. Also, in this section, the notations $Y_{ij}(1)$ and $Y_{ij}(0)$ denote the potential outcome under IV value $Z_{ij}=1$ and that under $Z_{ij}=0$, respectively. 


In matched IV studies, a commonly considered estimand is the \textit{effect ratio} (\citealp{baiocchi2010building, rosenbaum2020design}), defined as the ratio between the average IV effect on the outcome and that on the treatment: effect ratio $\theta=\{\sum_{i=1}^{I}\sum_{j=1}^{n_{i}}Y_{ij}(1)-Y_{ij}(0)\}/\{\sum_{i=1}^{I}\sum_{j=1}^{n_{i}}D_{ij}
\\(1)-D_{ij}(0)\}$, in which we assume that the IV relevance assumption holds (i.e., $\sum_{i=1}^{I}\sum_{j=1}^{n_{i}}D_{ij}\\(1)-D_{ij}(0)>0$). The effect ratio $\theta$ reduces to the classic sample complier average treatment effect when the treatment $D_{ij}$ is binary and when the exclusion restriction assumption (i.e., the IV $Z_{ij}$ can only affect $Y_{ij}$ through its effect on $D_{ij}$) and IV monotonicity assumptions hold (i.e., $D_{ij}(1)\geq D_{ij}(0)$ for all $i, j$). In the previous matched IV studies, researchers have routinely ignored inexact matching and relied on the conventional Wald estimator $\widehat{\theta}$ to conduct randomization-based inference for $\theta$ (\citealp{kang2016full, rosenbaum2020design}). As discussed in previous sections, ignoring inexact matching in randomization-based inference may cause severe bias. To mitigate such bias, we propose a new estimator called the \textit{bias-corrected Wald estimator} and the corresponding bias-corrected variance estimator to conduct randomization-based inference in inexactly matched IV studies. Specifically, consider the null hypothesis $H_{\theta_{0}}:\theta=\theta_{0}$, where $\theta_{0}$ is some prespecified value. We propose the following bias-corrected test statistic for testing $H_{\theta_{0}}$: 
\begin{equation*}
A_{*}(\theta_{0})=\frac{1}{I}\sum_{i=1}^{I}A_{*,i}(\theta_{0}), \text{ where } A_{*,i}(\theta_{0})=\sum_{j=1}^{n_{i}}\frac{Z_{ij}}{p_{ij}}(Y_{ij}-\theta_{0}D_{ij})-\sum_{j=1}^{n_{i}}\frac{1-Z_{ij}}{1-p_{ij}}(Y_{ij}-\theta_{0}D_{ij}).
\end{equation*}
Solving the estimating equation $A_{*}(\theta)=0$ under the oracle $p_{ij}\in (0,1)$ gives the oracle form of the bias-corrected Wald estimator for the effect ratio:
\begin{equation*}
    \widehat{\theta}_{*}=\frac{\sum_{i=1}^{I}\sum_{j=1}^{n_{i}}\frac{1}{p_{ij}(1-p_{ij})}Y_{ij}(Z_{ij}-p_{ij})}{\sum_{i=1}^{I}\sum_{j=1}^{n_{i}}\frac{1}{p_{ij}(1-p_{ij})}D_{ij}(Z_{ij}-p_{ij})}.
\end{equation*}
\setcounter{proposition}{0}
\begin{proposition}\label{prop: consistency_IV}
   Under the IV ignorability and relevance assumptions, as well as regularity conditions~\ref{condition: convergence of estimand IPPW} and~\ref{condition: bounded IPPW} specified in Appendix B.11, we have $\widehat{\theta}_{*}\xrightarrow{a.s.} \theta$ as $I\rightarrow \infty$ (i.e., the $\widehat{\theta}_{*}$ is strongly consistent).
\end{proposition}
To facilitate bias-corrected randomization-based inference using $A_{*}(\theta_{0})$, we also propose the following variance estimator for $A_{*}(\theta_{0})$: $V_{*}^2(\theta_{0})=\{I(I-1)\}^{-1}\sum_{i=1}^{I}\{A_{*,i}(\theta_{0})-A_{*}(\theta_{0})\}^2$. Then, consider the confidence set $CS^{\theta}_{*}=\{\theta_{0}: |A_{*}(\theta_{0})/\sqrt{V^{2}_{*}(\theta_{0})}|\leq\Phi^{-1}(1-\alpha/2)\}$, where $\alpha\in (0,1/2)$.
\setcounter{theorem}{0}
\begin{theorem}\label{the: IV confidence set}
Assuming independence across matched sets, the IV ignorability assumption, the IV relevance assumption, and regularity conditions ~\ref{condition: convergence of estimand IPPW}--\ref{condition: slow} specified in Appendices B.11-13, we have $\liminf_{I\rightarrow \infty}\text{pr}\big(\theta\in CS^{\theta}_{*}\mid \mathcal{Z} \big)\geq 1-\alpha$ for any level $\alpha\in (0,1/2)$.
\end{theorem}

To our knowledge, the $CS^{\theta}_{*}$ is the first valid confidence set for the effect ratio $\theta$ (which includes the sample complier average treatment effect as a special case) that has a theoretical guarantee of coverage rate under inexact matching (as stated in Theorem~\ref{the: IV confidence set}). The core idea of the proof of Theorem~\ref{the: IV confidence set} is to extend the arguments in \citet{baiocchi2010building} and \citet{kang2016full} from the perfect randomization case (assuming exact matching) to the biased randomization case (allowing for inexact matching). 

Note that the oracle form of the bias-corrected Wald estimator $\widehat{\theta}_{*}$ and the corresponding confidence set $CS^{\theta}_{*}$ involve the true post-matching IV assignment probabilities $p_{ij}$. In practical applications, we can adopt the commonly used ``plug-in" strategy (\citealp{rosenbaum1987model, ding2024first, pimentel2024covariate}) to replace each $p_{ij}$ with the estimate $\widehat{p}_{ij}$. Our simulation studies in Section C.4 show that the bias-corrected Wald estimator (based on either $\widehat{p}_{ij}$ or $p_{ij}$) outperforms the conventional Wald estimator in terms of estimation bias and coverage rate under the considered settings of inexact matching.

\section*{Appendix B: Technical Proofs and Additional Theoretical Results}\label{sec: Appendix proofs}

\subsection*{B.1: Proof of Proposition 1}

For each matched set $i$, we have 
\begin{align*}
    E(\widehat{\lambda}_{*, i}\mid \mathcal{Z})&=E\Big\{\frac{1}{n_{i}}\sum_{j=1}^{n_{i}}\Big(\frac{Z_{ij}Y_{ij}}{p_{ij}}-\frac{(1-Z_{ij})Y_{ij}}{1-p_{ij}}\Big)\mid \mathcal{Z}\Big\}\\
    &=\frac{1}{n_{i}}\Big\{\sum_{j=1}^{n_{i}}\frac{Y_{ij}(1)E(Z_{ij}\mid\mathcal{Z}) }{p_{ij} }-\sum_{j=1}^{n_{i}}\frac{Y_{ij}(0)(1-E(Z_{ij}\mid \mathcal{Z})) }{1-p_{ij}}\Big\}\\
    &=\frac{1}{n_{i}}\sum_{j=1}^{n_{i}}\{Y_{ij}(1)-Y_{ij}(0)\}.
\end{align*}

Therefore, we have 
\begin{equation*}       
    E(\widehat{\lambda}_{*}\mid\mathcal{Z})=E\Big(\sum_{i=1}^{I}\frac{n_{i}}{N} \widehat{\lambda}_{*, i}\mid \mathcal{Z}\Big)=\sum_{i=1}^{I}\frac{n_{i}}{N}E(\widehat{\lambda}_{*, i}\mid \mathcal{Z})=\frac{1}{N}\sum_{i=1}^{I}\sum_{j=1}^{n_{i}}\{Y_{ij}(1)-Y_{ij}(0)\}.
\end{equation*}

\subsection*{B.2: Statement and Proof of Proposition S2}
There is an asymptotically valid general class of variance estimator for $\widehat{\lambda}_{*}$. The key idea is to extend the arguments in \citet{fogarty2018mitigating} from the perfect randomization setting to the biased randomization (inexact matching) setting. Specifically, let $Q$ be any $I \times L$ matrix with $I > L$ ($I$ is the number of matched sets). For example, a canonical choice for $Q$ is the $I\times 1$ matrix (vector) with all the entries being one (i.e., a unit vector). Another common choice for $Q$ is an $I\times2$ matrix with the first column being all ones and the second being matched set sample weights $w_{i}=In_{i}/N$, $i=1,\dots, I$. In addition, the matrix $Q$ could contain covariate information aggregated at the matched set level. For example, when $K<I-1$, we can set $Q=(\mathbf{1}_{I\times 1}, \overline{\mathbf{x}}_{1}, \dots, \overline{\mathbf{x}}_{K})$, where $\mathbf{1}_{I\times 1}=(1,\dots, 1)^{T}$ is an $I$-dimensional unit vector and each $\overline{\mathbf{x}}_{k}=(n_{1}^{-1}\sum_{j=1}^{n_{1}}x_{1jk}, \dots, n_{I}^{-1}\sum_{j=1}^{n_{I}}x_{Ijk})^{T}$ is a vector recording the mean value of the $k$-th covariate within each matched set. Let $H_Q=Q(Q^TQ)^{-1}Q^T$ be the corresponding hat matrix of $Q$, and $h_{Qii}$ the $i$-th diagonal element of $H_Q$. Next, define $y_i=\widehat{\lambda}_{*,i}/\sqrt{1-h_{Qii}}$ and $\mathbf{y}=(y_1,\ldots,y_{I})$. Let $\mathcal{I}$ be an $I\times I$ identity matrix and $W$ be an $I \times I$ diagonal matrix with the $i$-th diagonal entry being $w_i=In_i/N$. Then, we propose the following variance estimator $S^2_{*}(Q)$ for $\widehat{\lambda}_{*}$: $S^2_{*}(Q)=I^{-2}\mathbf{y}W(\mathcal{I}-H_Q)W\mathbf{y}^{T}$. When $Q={I\times 1}$ (the unit vector), the general variance estimator $S^{2}_{*}(Q)$ reduces to $S^{2}_{*}$, i.e., it coincides with the variance estimator proposed in Section~3.1 of the main text.
\begin{proposition}\label{prop: variance}
    Assuming independence of treatment assignments across matched sets, along with Assumptions 1 and 2 in the main text. Then, for any prespecified $Q$, we have $E(S^2_{*}(Q)\mid \mathcal{Z})\geq \text{var}(\widehat{\lambda}_{*}\mid \mathcal{Z})$. 
\end{proposition}

\begin{proof}
Recall that for any random vector $\mathbf{y}$, we have $E(\mathbf{y}A\mathbf{y}^{T})=\mathbf{\beta} A\mathbf{\beta}^{T}+\text{tr}(A\Sigma)$ holds if $A$ is a symmetric matrix, $\mathbf{\beta}$ is the expectation vector of $\mathbf{y}$, and $\Sigma$ is the covariance matrix of $\mathbf{y}$. Since $W(\mathcal{I}-H_{Q})W$ is symmetric, we have:
\begin{align*}       
    I^2E\{S^2_{*}(Q)\mid\mathcal{Z}\}&=\mathbf{\beta}W(\mathcal{I}-H_{Q})W\mathbf{\beta}^{T}+\text{tr}(W(\mathcal{I}-H_{Q})W\Sigma) \\ 
    &=\mathbf{\beta}W(\mathcal{I}-H_{Q})W\mathbf{\beta}^{T}+\sum^{I}_{i=1}w_{i}^2
    \text{var}(y_{i}\mid\mathcal{Z})\times(1-h_{Qii}) \\ 
    &=\mathbf{\beta}W(\mathcal{I}-H_{Q})W\mathbf{\beta}^{T}+\sum^{I}_{i=1}w_{i}^2\frac{\text{var}(\widehat{\lambda}_{*, i}\mid\mathcal{Z})}{(1-h_{Qii})}\times(1-h_{Qii}) \\ 
    &=\mathbf{\beta}W(\mathcal{I}-H_{Q})W\mathbf{\beta}^{T}+\sum^{I}_{i=1}w_{i}^2\text{var}(\widehat{\lambda}_{*, i}\mid\mathcal{Z}).
\end{align*}
Under independence across matched sets, we have $\text{var}(\widehat{\lambda}_{*}\mid\mathcal{Z})=\sum_{i=1}^{I}\frac{n_{i}^2}{N^2}\text{var}(\widehat{\lambda}_{*, i}\mid\mathcal{Z})$. Therefore, we have $I^2E\{S^2_{*}(Q)\mid\mathcal{Z}\}=\mathbf{\beta}W(\mathcal{I}-H_{Q})W\mathbf{\beta}^{T}+I^2\text{var}(\widehat{\lambda}_{*}\mid\mathcal{Z})$, which implies that $E\{S^2_{*}(Q)\mid \mathcal{Z}\}-\text{var}(\widehat{\lambda}_{*}\mid \mathcal{Z})=I^{-2}\mathbf{\beta}W(\mathcal{I}-H_{Q})W\mathbf{\beta}^{T}\geq0$ (because the projection matrix $\mathcal{I}-H_{Q}$ is positive semi-definite). Moreover, when there is no treatment effect, $\beta=\boldsymbol{0}$ and $E\{S^2_{*}(Q)\mid \mathcal{Z}\}=\text{var}(\widehat{\lambda}_{*}\mid \mathcal{Z})$. In this case, $S^2_*(Q)$ is an unbiased estimator for $\text{var}(\widehat{\lambda}_{*}\mid \mathcal{Z})$.
\end{proof}

\vspace{-0.7cm}

\subsection*{B.3: Proof of Theorem 1}
To prove Theorem 1, we consider the following regularity conditions, which can be implied by Conditions 1--2 in the main text.
\setcounter{condition}{0}
\begin{condition}\label{condition: no extreme pairs}(No Extreme Matched Sets):
For each matched set $i$. we define $\widehat{\lambda}_{*,i}^{+}=
\text{max}_{\mathbf{Z}_{i}\in \mathcal{Z}_{i} }\widehat{\lambda}_{*, i}$, $\widehat{\lambda}_{*,i}^{-}=
\text{min}_{\mathbf{Z}_{i}\in \mathcal{Z}_{i} }\widehat{\lambda}_{*, i}$, and $M_{i}=\widehat{\lambda}_{*,i}^{+}-\widehat{\lambda}_{*,i}^{-}$, where $\mathbf{Z}_{i}=(Z_{i1}, \dots, Z_{in_{i}})$ and $\mathcal{Z}_{i}=\{\mathbf{Z}_{i}\in \{0,1\}^{n_{i}}: \sum_{j=1}^{n_{i}}Z_{ij}=m_{i}\}$ denotes the collection of all possible $\mathbf{Z}_{i}$. Then, we let $l_{i}=\min_{\mathbf{z}_{i} \in \mathcal{Z}_{i}} \text{pr}(\mathbf{Z}_{i}=\mathbf{z}_{i}\mid \mathcal{Z}_{i})$ and $\widetilde{M}_{i}=w_{i}M_{i}$. As $I \to \infty$, we have $\text{max}_{1\leq i \leq I}\widetilde{M}_{i}^2/\{\sum^{I}_{i=1}(l_{i})^{3}\widetilde{M}_{i}^2\} \to 0$. 
\end{condition}

\begin{condition}\label{condition: bounded size of matched sets}(Bounded Matched Sets and Bounded Entries of the Design Matrix $Q$): There exists some constant $C_{1}<\infty$ such that $n_i \leq C_1$ for all $i=1,\dots, I$, and $|q_{il}|\leq C_{1}$ for all $i=1,\dots, I, l=1,\dots, L$, where $q_{il}$ is the entry at the $i$-th row and $l$-th column of the design matrix $Q$.
\end{condition}

\begin{condition}\label{condition: bounded fourth moments}(Bounded Fourth Moments): There exists some constant $C_2<\infty$ such that for all $I$, we have $I^{-1}\sum_{i=1}^{I}M_{i}^{4} \leq C_2$, $I^{-1}\sum_{i=1}^{I}(\widehat{\lambda}_{*,i}^{+})^{4} \leq C_2$, and $I^{-1}\sum_{i=1}^{I}(\widehat{\lambda}_{*,i}^{-})^{4} \leq C_2$.
\end{condition}

\begin{condition}\label{condition: convergence} (Convergence of Finite-Population Means): For each matched set $i$, we define $\mu_{i}=E(\widehat{\lambda}_{*,i}\mid\mathcal{Z})$ and $\nu_{i}^2=\text{var}(\widehat{\lambda}_{*,i}\mid\mathcal{Z})$. As $I\rightarrow\infty$, we have: (i) $I^{-1}\sum_{i=1}^{I}w_{i}\mu_{i}$, $I^{-1}\sum_{i=1}^{I}w_{i}^2\mu_{i}$, and $I^{-1}\sum_{i=1}^{I}w_{i}^2\mu_{i}^2$ converge to some finite values; (ii) $I^{-1}\sum_{i=1}^{I}w_{i}^2\nu_{i}^2$ converges to some finite positive value; (iii) For $l=1,\dots,L$, the $I^{-1}\sum_{i=1}^{I}w_{i}\mu_{i}q_{il}$ converge to some finite values; (iv) $I^{-1}Q^{T}Q$ converges to some finite, invertible $L \times L$ matrix $\widetilde{Q}$.
\end{condition}

Condition~\ref{condition: no extreme pairs} states that, as the sample size goes to infinity, the contribution from any single matched set would not be comparable (proportional) to the combined contributions from all the matched sets. For example, this condition will naturally hold if both the size of each matched set and test statistics contributed by the matched sets are bounded, and all propensity scores $e_{ij}\in [\rho, 1-\rho]$ for some $\rho>0$. Conditions~\ref{condition: bounded size of matched sets}--\ref{condition: convergence} are also commonly considered in randomization-based inference for matched or stratified causal studies (\citealp{rosenbaum2002observational, fogarty2018mitigating}). Conditions~\ref{condition: no extreme pairs}--\ref{condition: convergence} are generalizations of some common regularity conditions from the exact matching to the potentially inexact matching case \citep{rosenbaum2002observational,fogarty2018mitigating}. Also, Conditions~\ref{condition: no extreme pairs}--\ref{condition: convergence} are weaker than Conditions 1--2 in the main text.

Next, recall the following exact form of the Lindeberg-Feller central limit theorem.
\begin{lemma}\label{lemma: lindeberg-feller}(Lindeberg-Feller Central Limit Theorem): Suppose we have a triangular array of random variables $X_{n,m}, 1\leq m\leq n$ with $E(X_{n,m})=0$ and $\sum^{n}_{m=1}E(X_{n,m}^2) \rightarrow \sigma^2 >0$ as $n \rightarrow \infty$. If the sequence of $X_{n,m}$ satisfies $\lim_{n \to \infty}\sum^{n}_{m=1}E[X_{n,m}^2 \mathbbm{1}\{|X_{n,m}|>c\}]=0$ for all $c>0$, we have $\sum^{n}_{m=1}X_{n,m}\xrightarrow{d} N(0,\sigma^2)$ as $n \rightarrow \infty$.
\end{lemma}

Under Condition~\ref{condition: no extreme pairs}, we can prove the asymptotic normality of $\widehat{\lambda}_{*}$, of which the core idea of the proof is to adjust the proof strategy in \citet{Su2024} to the inexact matching case.

\begin{lemma}\label{lem: asymptotic normality}
Suppose Condition~\ref{condition: no extreme pairs} holds, and the treatment assignments are independent across matched sets. Under Assumption 1 in the main text, we have 
\begin{equation*}
    \frac{\widehat{\lambda}_{*}-\lambda}{\sqrt{\text{var}(\widehat{\lambda}_{*}\mid\mathcal{Z})}} \xrightarrow{d} N(0,1) \text{ as $I \rightarrow \infty$.}
\end{equation*}
\end{lemma}
\begin{proof}
Recall that $\widehat{\lambda}_{*,i}$ is the contribution to the test statistic $\widehat{\lambda}_{*}$ from the matched set i. Also, recall that $\mu_{i}=E(\widehat{\lambda}_{*,i}\mid\mathcal{Z})$ denote the expectation of $\widehat{\lambda}_{*,i}$ and $\nu_{i}^2=\text{var}(\widehat{\lambda}_{*,i}\mid\mathcal{Z})$ denote the variance of $\widehat{\lambda}_{*,i}$. Let $\mu=\sum_{i=1}^{I}\frac{n_{i}}{N}\mu_{i}=E(\widehat{\lambda}_{*}\mid\mathcal{Z})$ and $\sigma^2=\sum_{i=1}^{I}(\frac{n_{i}}{N})^2\nu_{i}^2=\text{var}(\widehat{\lambda}_{*}\mid\mathcal{Z})$. Next, we define $Y_{i}=\frac{n_{i}(\widehat{\lambda}_{*,i}-\mu_{i})}{N\sigma}$. Since $\widehat{\lambda}_{*,i}^{-}\leq \mu_{i} \leq \widehat{\lambda}_{*,i}^{+}$, we have $|\widehat{\lambda}_{*,i}-\mu_{i}|\leq \widehat{\lambda}_{*,i}^{+}-\widehat{\lambda}_{*,i}^{-}=M_{i}$ for any $\mathbf{Z}_{i}=(Z_{i1}, \dots, Z_{in_{i}})\in \mathcal{Z}_{i}$, where $\mathcal{Z}_{i}=\{\mathbf{Z}_{i}\in \{0,1\}^{n_{i}}: \sum_{j=1}^{n_{i}}Z_{ij}=m_{i}\}$. Therefore, we have

\vspace{-1.2cm}
\begin{align*}
    \nu_{i}^2=\sum_{\mathbf{z}_{i} \in \mathcal{Z}_{i}}\big\{\text{pr}(\mathbf{Z}_{i}=\mathbf{z}_{i}\mid \mathcal{Z})\times (\widehat{\lambda}_{*,i}-\mu_{i})^2\big\} \geq p_{i}^{-}(\widehat{\lambda}_{*,i}^{-}-\mu_{i})^2+p_{i}^{+}(\widehat{\lambda}_{*,i}^{+}-\mu_{i})^2,
\end{align*}
\vspace{-1cm}
\\ where $p_{i}^{-}=\text{pr}(\widehat{\lambda}_{*,i}=\widehat{\lambda}_{*,i}^{-}\mid\mathcal{Z})$ and $p_{i}^{+}=\text{pr}(\widehat{\lambda}_{*,i}=\widehat{\lambda}_{*,i}^{+}\mid\mathcal{Z})$. Moreover, since 

\vspace{-1.2cm}
\begin{align*}
    p_{i}^{+}\widehat{\lambda}_{*,i}^{+}+(1-p_{i}^{+})\widehat{\lambda}_{*,i}^{-}\leq \mu_{i} \leq p_{i}^{-}\widehat{\lambda}_{*,i}^{-}+(1-p_{i}^{-})\widehat{\lambda}_{*,i}^{+},
\end{align*}
\vspace{-1.6cm}
\\ we have

\vspace{-0.8cm}
\begin{equation*}
    \mu_{i}-\widehat{\lambda}_{*,i}^{-} \geq p_{i}^{+}(\widehat{\lambda}_{*,i}^{+}-\widehat{\lambda}_{*,i}^{-})=p_{i}^{+}M_{i},\quad \widehat{\lambda}_{*,i}^{+}-\mu_{i} \geq p_{i}^{-}(\widehat{\lambda}_{*,i}^{+}-\widehat{\lambda}_{*,i}^{-})=p_{i}^{-}M_{i}.
\end{equation*}
\vspace{-1.2cm}
\\ Then, we can get

\vspace{-1.3cm}
\begin{align*}
    \nu_{i}^2 &\geq p_{i}^{-}(\widehat{\lambda}_{*,i}^{-}-\mu_{i})^2+p_{i}^{+}(\widehat{\lambda}_{*,i}^{+}-\mu_{i})^2 \geq p_{i}^{-}(p_{i}^{+})^2M_{i}^2+p_{i}^{+}(p_{i}^{-})^2M_{i}^2 =p_{i}^{-}p_{i}^{+}(p_{i}^{-}+p_{i}^{+})M_{i}^2.
\end{align*}
\vspace{-1.2cm}
\\ Since $l_{i}=\min_{\mathbf{z}_{i} \in \mathcal{Z}_{i}} \text{pr}(\mathbf{Z}_{i}=\mathbf{z}_{i}\mid \mathcal{Z}_{i})$ is the lower bound of all possible treatment assignment probabilities in each matched set $i$, we have $p_{i}^{+}, p_{i}^{-} \geq l_{i}$ and $\nu_{i}^2 \geq 2(l_{i})^3M_{i}^2$. Therefore, we have

\vspace{-1.3cm}
\begin{align*}
    Y_{i}^2=\frac{n_{i}^2(\widehat{\lambda}_{*,i}-\mu_{i})^2}{N^2\sigma^2}=\frac{(\frac{n_{i}}{N})^2(\widehat{\lambda}_{*,i}-\mu_{i})^2}{\sum_{i=1}^{I}(\frac{n_{i}}{N})^2\nu_{i}^2} \leq  \frac{\widetilde{M}_{i}^2}{\sum_{i=1}^{I}2(l_{i})^3\widetilde{M}_{i}^2}.
\end{align*}
\vspace{-1cm}
\\ Therefore, under Condition~\ref{condition: no extreme pairs}, we can conclude that $\text{max}_{1 \leq i \leq I}Y_{i}^2 \rightarrow 0$ as $I \rightarrow \infty$. This implies that, for any $c>0$, we have $\lim_{I \to \infty}\sum_{i=1}^{I}E[Y_{i}^2 \mathbbm{1} \{|Y_{i}^2|>c^{2}\}\mid\mathcal{Z}]=\lim_{I \to \infty}\sum_{i=1}^{I}\\ E[Y_{i}^2 \mathbbm{1} \{|Y_{i}|>c\}\mid\mathcal{Z}] \rightarrow 0$, i.e., the Lindeberg-Feller condition holds. Since we have $E(Y_{i}\mid\mathcal{Z})=0$ for all $i=1,\dots,I$ and for all $I$, and $\sum_{i=1}^{I}E(Y_{i}^2\mid\mathcal{Z})=1$ for all $I$, the desired result is obtained by invoking Lemma~\ref{lemma: lindeberg-feller} and Proposition 1. 
\end{proof}

Then, we can prove some important lemmas for proving Theorem 1. The core idea of these proofs is to extend the arguments in \citet{fogarty2018mitigating} from the perfect randomization case to the biased randomization (inexact matching) case.

\begin{lemma}\label{lem: converge in p}
    Under Conditions \ref{condition: bounded size of matched sets}--\ref{condition: convergence}, Assumption 1 in the main text, and independence of treatment assignments across matched sets, as $I \rightarrow \infty$, we have

    \vspace{-1cm}
    \begin{align*}
        I^{-1}\sum_{i=1}^{I}w_{i}\widehat{\lambda}_{*,i} &\xrightarrow{p} \lim\limits_{I \rightarrow \infty}I^{-1}\sum_{i=1}^{I}w_{i}\mu_{i}, \\ I^{-1}\sum_{i=1}^{I}w_{i}^2\widehat{\lambda}_{*,i} &\xrightarrow{p} \lim\limits_{I \rightarrow \infty}I^{-1}\sum_{i=1}^{I}w_{i}^2\mu_{i}, \\ I^{-1}\sum_{i=1}^{I}w_{i}\widehat{\lambda}_{*,i}q_{il} &\xrightarrow{p} \lim\limits_{I \rightarrow \infty}I^{-1}\sum_{i=1}^{I}w_{i}\mu_{i}q_{il}, \\    I^{-1}\sum_{i=1}^{I}w_{i}^2\widehat{\lambda}_{*,i}^2 &\xrightarrow{p} \lim\limits_{I \rightarrow \infty}I^{-1}\sum_{i=1}^{I}w_{i}^2(\nu_{i}^2+\mu_{i}^2).
    \end{align*}
\end{lemma}

\begin{proof}
    Note that

\vspace{-1.3cm}
\begin{align*}
    E\Big(I^{-1}\sum_{i=1}^{I}w_{i}\widehat{\lambda}_{*,i}\mid\mathcal{Z}\Big)&=I^{-1}\sum_{i=1}^{I}w_{i}\mu_{i}, \\ E\Big(I^{-1}\sum_{i=1}^{I}w_{i}^2\widehat{\lambda}_{*,i}\mid\mathcal{Z}\Big)&=I^{-1}\sum_{i=1}^{I}w_{i}^2\mu_{i}, \\ E\Big(I^{-1}\sum_{i=1}^{I}w_{i}\widehat{\lambda}_{*,i}q_{il}\mid\mathcal{Z}\Big)&=I^{-1}\sum_{i=1}^{I}w_{i}\mu_{i}q_{il}, \\  E\Big(I^{-1}\sum_{i=1}^{I}w_{i}^2\widehat{\lambda}_{*,i}^2\mid\mathcal{Z}\Big)&=I^{-1}\sum_{i=1}^{I}w_{i}^2(\nu_{i}^2+\mu_{i}^2).
\end{align*}
 Next, we can prove that all the variances of the left-hand side terms in Lemma~\ref{lem: converge in p} converge to zero. For the first term, we have

\vspace{-1cm}
\begin{align*}
    \text{var}\Big( I^{-1}\sum_{i=1}^{I}w_{i}\widehat{\lambda}_{*,i}\mid\mathcal{Z}\Big) &=I^{-2}\sum_{i=1}^{I}w_{i}^2\nu_{i}^2 \\ 
    &\leq I^{-2}\sum_{i=1}^{I}w_{i}^2M_{i}^2 \\  
    &\leq I^{-2}\Big(\sum_{i=1}^{I}w_{i}^4\Big)^{1/2}\Big(\sum_{i=1}^{I}M_{i}^4\Big)^{1/2} \\
    &=I^{-1}\Big(I^{-1}\sum_{i=1}^{I}w_{i}^4\Big)^{1/2}\Big(I^{-1}\sum_{i=1}^{I}M_{i}^4\Big)^{1/2} \\
    &\leq C_{1}^2C_{2}^{1/2}/I \rightarrow 0 \ \text{as} \ I \rightarrow \infty.
\end{align*}
\vspace{-1.5cm}
\\ For the second term, we have

\vspace{-1.3cm}
\begin{align*}
    \text{var}\Big(I^{-1}\sum_{i=1}^{I}w_{i}^2\widehat{\lambda}_{*,i}\mid \mathcal{Z}\Big)&=I^{-2}\sum_{i=1}^{I}w_{i}^4\nu_{i}^2 \\ 
    &\leq I^{-2}\sum_{i=1}^{I}w_{i}^4M_{i}^2 \\  
    &\leq I^{-2}\Big(\sum_{i=1}^{I}w_{i}^8\Big)^{1/2}\Big(\sum_{i=1}^{I}M_{i}^4\Big)^{1/2} \\
    &=I^{-1}\Big(I^{-1}\sum_{i=1}^{I}w_{i}^8\Big)^{1/2}\Big(I^{-1}\sum_{i=1}^{I}M_{i}^4\Big)^{1/2} \\
    &\leq C_{1}^4C_{2}^{1/2}/I \rightarrow 0 \ \text{as} \ I \rightarrow \infty.
\end{align*}

\vspace{-0.5cm} 
\noindent For the third term, we have

\vspace{-1.3cm}
\begin{align*}
    \text{var}\Big(I^{-1}\sum_{i=1}^{I}w_{i}\widehat{\lambda}_{*,i}q_{il}\mid \mathcal{Z}\Big)&=I^{-2}\sum_{i=1}^{I}w_{i}^2\nu_{i}^2q_{il}^2 \\ 
    &\leq I^{-2}\sum_{i=1}^{I}w_{i}^2M_{i}^2q_{il}^2 \\  
    &\leq I^{-2}\Big(\sum_{i=1}^{I}w_{i}^4M_{i}^4\Big)^{1/2}\Big(\sum_{i=1}^{I}q_{il}^4\Big)^{1/2} \\
    &=I^{-1}\Big(I^{-1}\sum_{i=1}^{I}w_{i}^4M_{i}^4\Big)^{1/2}\Big(I^{-1}\sum_{i=1}^{I}q_{il}^4\Big)^{1/2} \\
    &\leq C_{1}^{4}C_{2}^{1/2}/I \rightarrow 0 \ \text{as} \ I \rightarrow \infty.
\end{align*}
For the fourth term, we have 

\vspace{-1.3cm}
\begin{align*}
    \text{var}\Big(I^{-1}\sum_{i=1}^{I}w_{i}^2\widehat{\lambda}_{*,i}^2\mid \mathcal{Z}\Big)&=I^{-2}\sum_{i=1}^{I}w_{i}^{4}\text{var}(\widehat{\lambda}_{*,i}^2\mid\mathcal{Z}) \\ 
    &\leq I^{-2}\sum_{i=1}^{I}w_{i}^4E(\widehat{\lambda}_{*,i}^4\mid\mathcal{Z}) \\
    &\leq I^{-2}\sum_{i=1}^{I}w_{i}^4\{(\widehat{\lambda}_{*,i}^{-})^4+(\widehat{\lambda}_{*,i}^{+})^4\} \\ 
    &\leq2C_{1}^4C_{2}/I \rightarrow 0 \ \text{as} \ I \rightarrow \infty.
\end{align*}
\vspace{-1.3cm}
\\ Invoking Chebyshev's inequality, the desired convergence results stated in Lemma~\ref{lem: converge in p} can be proved based on the above results. Let us take the first term as an example. Define the random variable $L_I=I^{-1}\sum_{i=1}^{I}w_{i}\widehat{\lambda}_{*,i}-E\Big(I^{-1}\sum_{i=1}^{I}w_{i}\widehat{\lambda}_{*,i}\mid\mathcal{Z}\Big)$. Then, for all $\epsilon>0$, we have 

\vspace{-1.5cm}
\begin{align*}
    \text{pr}\big\{|L_I-E(L_I\mid \mathcal{Z})|\geq \epsilon\mid \mathcal{Z}\big\}\leq \frac{\text{var}(L_I\mid \mathcal{Z})}{\epsilon^2}.
\end{align*}
\vspace{-1.3cm}

\noindent Since $E(L_I\mid \mathcal{Z})=0$, and $\text{var}(L_I\mid \mathcal{Z})=\text{var}\Big( I^{-1}\sum_{i=1}^{I}w_{i}\widehat{\lambda}_{*,i}\mid\mathcal{Z}\Big) \rightarrow 0$ as $I \rightarrow \infty$, we can conclude that $I^{-1}\sum_{i=1}^{I}w_{i}\widehat{\lambda}_{*,i}-E\Big(I^{-1}\sum_{i=1}^{I}w_{i}\widehat{\lambda}_{*,i}\mid\mathcal{Z}\Big)=I^{-1}\sum_{i=1}^{I}w_{i}\widehat{\lambda}_{*,i}-I^{-1}\sum_{i=1}^{I}w_{i}\mu_{i}\xrightarrow{p}0$ as $I \rightarrow \infty$. Therefore, by Condition \ref{condition: convergence}, we have $I^{-1}\sum_{i=1}^{I}w_{i}\widehat{\lambda}_{*,i} \xrightarrow{p} \lim_{I \rightarrow \infty}I^{-1}\sum_{i=1}^{I}w_{i}\mu_{i}$. Similar arguments can be applied to other convergence results in Lemma~\ref{lem: converge in p}.
\end{proof}
\vspace{-0.5cm}
\begin{lemma}\label{lem: h_Q goes to zero}
   Under Conditions \ref{condition: bounded size of matched sets} and \ref{condition: convergence}, there exists some $I^{\prime}<\infty$ and $C^{\prime}<\infty$ such that for all $I\geq I^{\prime}$ and all $i=1,\dots,I$, we have $|h_{Qii}|\leq C^{\prime}/I$.
\end{lemma}
\begin{proof}
    Let $Q_{i\cdot}$ denote the $i$-th row of matrix $Q$. Note that

    \vspace{-0.8cm}
    $$h_{Qii}=Q_{i\cdot}(Q^{T}Q)^{-1}Q_{i\cdot}^{T}=I^{-1}Q_{i\cdot}(I^{-1}Q^{T}Q)^{-1}Q_{i\cdot}^{T}.$$ 
    \vspace{-1.2cm}
    
    \noindent Hence, by Condition \ref{condition: convergence}, for any $i=1,\dots, I$, we have

    \vspace{-1cm}
    $$\lim\limits_{I \rightarrow \infty}h_{Qii}=\lim\limits_{I \rightarrow \infty}I^{-1}Q_{i\cdot}(I^{-1}Q^{T}Q)^{-1}Q_{i\cdot}^{T}=\lim\limits_{I \rightarrow \infty}I^{-1}Q_{i\cdot}(\widetilde{Q})^{-1}Q_{i\cdot}^{T}.$$
\vspace{-0.9cm}
    
    \noindent Since all the entries of $Q$ are uniformly bounded by some constant (according to Condition 2), the above equation can immediately imply the desired result. 
\end{proof}

\begin{lemma}\label{lem: valid variance}
    Under Conditions \ref{condition: bounded size of matched sets} and \ref{condition: convergence}, Assumption 1 in the main text, and independence of treatment assignments across matched sets, as $I \rightarrow \infty$, we have
    \begin{align*}
        \frac{\text{var}(\widehat{\lambda}_{*}\mid\mathcal{Z})}{S_{*}^2(Q)}\xrightarrow{p}1-\frac{\lim\limits_{I \rightarrow \infty}I^{-1}\boldsymbol{\mu} W(\mathcal{I}-H_{Q})W\boldsymbol{\mu}^{T}}{\lim\limits_{I \rightarrow \infty}I^{-1}\boldsymbol{\mu} W(\mathcal{I}-H_{Q})W\boldsymbol{\mu}^{T}+\lim\limits_{I\rightarrow\infty}I^{-1}\sum_{i=1}^{I}w_i^2\nu_{i}^2} \in (0,1],
    \end{align*}
    where $\boldsymbol{\mu}=(\mu_{1},\dots, \mu_{I})$.
\end{lemma}

\begin{proof}
We decompose $IS_{*}^2(Q)$ into two components: the first component is $I^{-1}\mathbf{y}WW\mathbf{y}^{T}$ and the second one is $-I^{-1}\mathbf{y}WQ(Q^{T}Q)^{-1}Q^{T}W\mathbf{y}^{T}$. By Lemma~\ref{lem: h_Q goes to zero}, there exist some $I^{\prime}<\infty$ and $C^{\prime}<\infty$ such that for all $I\geq I^{\prime}$ and all $i=1,\dots,I$, we have $|h_{Qii}|\leq C^{\prime}/I$ and $|1-h_{Qii}|\geq 1/2$. Therefore, for any $I\geq I^{\prime}$, we have

\vspace{-1.2cm}
\begin{align*}
       \Big(I^{-1}\mathbf{y}WW\mathbf{y}^{T}-I^{-1}\sum_{i=1}^{I}w_{i}^2\widehat{\lambda}_{*,i}^2\Big)^{2}&= \Big(I^{-1}\sum_{i=1}^{I}w_{i}^2\widehat{\lambda}_{*,i}^2\frac{h_{Qii}}{1-h_{Qii}}\Big)^{2}\\
    &\leq I^{-1}\Big(I^{-1}\sum_{i=1}^{I}w_{i}^4\widehat{\lambda}_{*,i}^4 \Big)\Big(\sum_{i=1}^{I} \frac{h^{2}_{Qii}}{(1-h_{Qii})^{2}}\Big)\\
    & \leq I^{-1} C_{1}^{4} \Big(I^{-1}\sum_{i=1}^{I}(\widehat{\lambda}^{+}_{*,i})^4+I^{-1}\sum_{i=1}^{I}(\widehat{\lambda}^{-}_{*,i})^4 \Big)\Big(4\sum_{i=1}^{I} (C^{\prime}/I)^{2}\Big)\\
    &\leq \frac{8C_{1}^{4}C_{2}C^{\prime2}}{I^{2}}\rightarrow 0 \text{ as $I\rightarrow\infty$.}
\end{align*}
This implies that, for the first component of $IS_{*}^{2}(Q)$, we have (by Lemma~\ref{lem: converge in p})

\vspace{-1.2cm}
\setcounter{equation}{0}
\begin{align}\label{eqn: the frist component}
I^{-1}\mathbf{y}WW\mathbf{y}^{T}&=I^{-1}\sum_{i=1}^{I}w_{i}^2\widehat{\lambda}_{*,i}^2/(1-h_{Qii})\nonumber\\
&\xrightarrow{p} \lim\limits_{I \rightarrow \infty}I^{-1}\sum_{i=1}^{I}w_{i}^2(\nu_{i}^2+\mu_{i}^2)=\lim\limits_{I \rightarrow \infty}\Big\{I^{-1}\boldsymbol{\mu} WW\boldsymbol{\mu}^{T}+I\text{var}(\widehat{\lambda}_{*}\mid\mathcal{Z})\Big\}.
\end{align}
\vspace{-1.2cm}

Also, recall that the second component of $IS_{*}^{2}(Q)$ is

\vspace{-0.7cm}
\begin{equation*}
    -I^{-1}\mathbf{y}WQ(Q^{T}Q)^{-1}Q^{T}W\mathbf{y}^{T}= -I^{-1}\mathbf{y}WQ(I^{-1}Q^{T}Q)^{-1}I^{-1}Q^{T}W\mathbf{y}^{T}.
\end{equation*}
Note that the $l$-th row of $I^{-1}Q^{T}W\mathbf{y}^{T}$ equals $I^{-1}\sum_{i=1}^{I}w_{i}(\widehat{\lambda}_{*,i}/\sqrt{1-h_{Qii}})q_{il}$, and the $l^{\prime}$-th column of $I^{-1}\mathbf{y}WQ$ equals $I^{-1}\sum_{i=1}^{I}w_{i}(\widehat{\lambda}_{*,i}/\sqrt{1-h_{Qii}})q_{il^{\prime}}$. Meanwhile, for $I\geq I^{\prime}$, we have 

\vspace{-1cm}
\begin{align*}
      &\quad \ \Big|\Big(I^{-1}\sum_{i=1}^{I}\frac{w_{i}\widehat{\lambda}_{*,i}q_{il}}{\sqrt{1-h_{Qii}}}\Big)\Big(I^{-1}\sum_{i=1}^{I}\frac{w_{i}\widehat{\lambda}_{*,i}q_{il^{\prime}}}{\sqrt{1-h_{Qii}}}\Big)-\Big(I^{-1}\sum_{i=1}^{I}w_{i}\widehat{\lambda}_{*,i}q_{il}\Big)\Big(I^{-1}\sum_{i=1}^{I}w_{i}\widehat{\lambda}_{*,i}q_{il^{\prime}}\Big)\Big|\\
      &\leq I^{-2}\sum_{i=1}^{I}\sum_{i^{\prime}=1}^{I}\Big|\frac{w_{i}\widehat{\lambda}_{*,i}q_{il}}{\sqrt{1-h_{Qii}}}\frac{w_{i^{\prime}}\widehat{\lambda}_{*,i^{\prime}}q_{i^{\prime}l^{\prime}}}{\sqrt{1-h_{Qi^{\prime}i^{\prime}}}}-\big(w_{i}\widehat{\lambda}_{*,i}q_{il}\big)\big(w_{i^{\prime}}\widehat{\lambda}_{*,i^{\prime}}q_{i^{\prime}l^{\prime}}\big)\Big| \\
       &\leq  I^{-1}\sqrt{I^{-2}\sum_{i=1}^{I}\sum_{i^{\prime}=1}^{I}\big(w_{i}\widehat{\lambda}_{*,i}q_{il}\big)^{2}\big(w_{i^{\prime}}\widehat{\lambda}_{*,i^{\prime}}q_{i^{\prime}l^{\prime}}\big)^{2}}\sqrt{\sum_{i=1}^{I}\sum_{i^{\prime}=1}^{I}\Big(\frac{1}{\sqrt{1-h_{Qii}}}\frac{1}{\sqrt{1-h_{Qi^{\prime}i^{\prime}}}}-1\Big)^2} \\
       &\leq I^{-1}C_{1}^{4}\Big(I^{-1}\sum_{i=1}^{I}\widehat{\lambda}_{*,i}^{2}\Big)4C^{\prime 2} \\
    &\leq \frac{4C_{1}^{4}C_{2}^{1/2}C^{\prime2}}{I}\rightarrow 0 \ \text{ as $I\rightarrow\infty$.}
\end{align*}
Let $\widetilde{q}_{st}$ denote the entry at the $s$-th row and $t$-th column of the matrix $\widetilde{Q}^{-1}$, where $s=1,\dots, L$ and $t=1,\dots, L$. Also, we let $\widehat{\boldsymbol{\lambda}}_{*}=(\widehat{\lambda}_{*,1}, \dots, \widehat{\lambda}_{*,I})$. We have

\vspace{-1.2cm}
\begin{align*}
&\quad \big|I^{-1}\mathbf{y}WQ\widetilde{Q}^{-1}I^{-1}Q^{T}W\mathbf{y}^{T}-I^{-1}\widehat{\boldsymbol{\lambda}}_{*}^{T}WQ\widetilde{Q}^{-1}I^{-1}Q^{T}W\widehat{\boldsymbol{\lambda}}_{*}^{T}\big|\\
&=\Big|\sum_{l=1}^{L}\sum_{l^{\prime}=1}^{L}\widetilde{q}_{l^{\prime}l}\Big\{\Big(I^{-1}\sum_{i=1}^{I}\frac{w_{i}\widehat{\lambda}_{*,i}q_{il}}{\sqrt{1-h_{Qii}}}\Big)\Big(I^{-1}\sum_{i=1}^{I}\frac{w_{i}\widehat{\lambda}_{*,i}q_{il^{\prime}}}{\sqrt{1-h_{Qii}}}\Big)-\Big(I^{-1}\sum_{i=1}^{I}w_{i}\widehat{\lambda}_{*,i}q_{il}\Big)\Big(I^{-1}\sum_{i=1}^{I}w_{i}\widehat{\lambda}_{*,i}q_{il^{\prime}}\Big)\Big\}\Big|\\
&\leq \sum_{l=1}^{L}\sum_{l^{\prime}=1}^{L}|\widetilde{q}_{l^{\prime}l}|\Big|\Big(I^{-1}\sum_{i=1}^{I}\frac{w_{i}\widehat{\lambda}_{*,i}q_{il}}{\sqrt{1-h_{Qii}}}\Big)\Big(I^{-1}\sum_{i=1}^{I}\frac{w_{i}\widehat{\lambda}_{*,i}q_{il^{\prime}}}{\sqrt{1-h_{Qii}}}\Big)-\Big(I^{-1}\sum_{i=1}^{I}w_{i}\widehat{\lambda}_{*,i}q_{il}\Big)\Big(I^{-1}\sum_{i=1}^{I}w_{i}\widehat{\lambda}_{*,i}q_{il^{\prime}}\Big)\Big|\\
&\leq \Big(\sum_{l=1}^{L}\sum_{l^{\prime}=1}^{L}|\widetilde{q}_{l^{\prime}l}|\Big)\frac{4C_{1}^{4}C_{2}^{1/2}C^{\prime2}}{I}\rightarrow 0 \ \text{ as $I\rightarrow\infty$.}
\end{align*}
Also, by Lemma \ref{lem: converge in p}, we have

\vspace{-1cm}
\begin{equation*}
   - I^{-1}\widehat{\boldsymbol{\lambda}}_{*}^{T}WQ\widetilde{Q}^{-1}I^{-1}Q^{T}W\widehat{\boldsymbol{\lambda}}_{*}^{T} \xrightarrow{p} - \lim_{I \rightarrow \infty} I^{-1}\boldsymbol{\mu} WQ\widetilde{Q}^{-1}I^{-1}Q^{T}W\boldsymbol{\mu}^{T} \ \text{as $I\rightarrow \infty$. }
\end{equation*}
Therefore, we have

\vspace{-1cm}
\begin{align*}
    -I^{-1}\mathbf{y}WQ\widetilde{Q}^{-1}I^{-1}Q^{T}W\mathbf{y}^{T}\xrightarrow{p} \lim_{I \rightarrow \infty} -I^{-1}\boldsymbol{\mu} WQ\widetilde{Q}^{-1}I^{-1}Q^{T}W\boldsymbol{\mu}^{T}\ \text{as $I\rightarrow \infty$.}
\end{align*}
Since $\lim_{I\rightarrow \infty} I^{-1}Q^{T}Q=\widetilde{Q}$ (Condition \ref{condition: convergence}), for the second component of $IS_{*}^{2}(Q)$, we have

\vspace{-1.2cm}
\begin{align}\label{eqn: second component}
    -I^{-1}\mathbf{y}WQ(Q^{T}Q)^{-1}Q^{T}W\mathbf{y}^{T}&=-I^{-1}\mathbf{y}WQ(I^{-1}Q^{T}Q)^{-1}I^{-1}Q^{T}W\mathbf{y}^{T} \nonumber\\
    &\xrightarrow{p} \lim_{I \rightarrow \infty} -I^{-1}\boldsymbol{\mu} WQ(Q^{T}Q)^{-1}Q^{T}W\boldsymbol{\mu}^{T}\nonumber\\
    &= \lim_{I \rightarrow \infty} -I^{-1}\boldsymbol{\mu} WQ(I^{-1}Q^{T}Q)^{-1}I^{-1}Q^{T}W\boldsymbol{\mu}^{T} \ \text{as $I\rightarrow \infty$.}
\end{align}
\vspace{-1.2cm}

Note that the projection matrix $\mathcal{I}-H_{Q}$ is positive semi-definite, combining (\ref{eqn: the frist component}) and (\ref{eqn: second component}), we have

\vspace{-1.5cm}
\begin{align*}
    I\{S_{*}^2(Q)-\text{var}(\widehat{\lambda}_{*}\mid\mathcal{Z})\} \xrightarrow{p} \lim\limits_{I \rightarrow \infty}I^{-1}\boldsymbol{\mu} W(\mathcal{I}-H_{Q})W\boldsymbol{\mu}^{T} \geq 0.
\end{align*}
Therefore, as $I \rightarrow \infty$, we have
\begin{align*}
    \frac{\text{var}(\widehat{\lambda}_{*}\mid\mathcal{Z})}{S_{*}^2(Q)}&=1-\frac{S_{*}^2(Q)-\text{var}(\widehat{\lambda}_{*}\mid\mathcal{Z})}{S_{*}^2(Q)} \\ 
    &=1-\frac{I\{S_{*}^2(Q)-\text{var}(\widehat{\lambda}_{*}\mid\mathcal{Z})\}}{IS_{*}^2(Q)} \\ 
    &\xrightarrow{p}1-\frac{\lim\limits_{I \rightarrow \infty}I^{-1}\boldsymbol{\mu} W(\mathcal{I}-H_{Q})W\boldsymbol{\mu}^{T}}{\lim\limits_{I \rightarrow \infty}I^{-1}\boldsymbol{\mu} W(\mathcal{I}-H_{Q})W\boldsymbol{\mu}^{T}+\lim\limits_{I\rightarrow\infty}I^{-1}\sum_{i=1}^{I}w_i^2\nu_{i}^2} \in (0,1],
\end{align*}
in which $\lim_{I \rightarrow \infty}I^{-1}\boldsymbol{\mu} W(\mathcal{I}-H_{Q})W\boldsymbol{\mu}^{T}\geq 0$ and $\lim_{I\rightarrow\infty}I^{-1}\sum_{i=1}^{I}w_{i}^{2}\nu_{i}^2>0$ (by Condition \ref{condition: convergence}). 
\end{proof}

Finally, we are ready to prove Theorem 1.
\begin{proof}
We have
\vspace{-0.4cm}
\begin{align*}
    &\quad \ \lim\limits_{I \rightarrow \infty}\text{pr}\Big(\widehat{\lambda}_{*}-\Phi^{-1}(1-\alpha/2)\sqrt{S_{*}^2(Q)} \leq \lambda \leq \widehat{\lambda}_{*}+\Phi^{-1}(1-\alpha/2)\sqrt{S_{*}^2(Q)}\mid\mathcal{Z} \Big) \\  &=\lim\limits_{I \rightarrow \infty}\text{pr}\bigg(-\Phi^{-1}(1-\alpha/2) \leq \frac{\widehat{\lambda}_{*}-\lambda}{\sqrt{S_{*}^2(Q)}} \leq \Phi^{-1}(1-\alpha/2)\mid\mathcal{Z}  \bigg) \\ 
    &=\lim\limits_{I \rightarrow \infty}\text{pr}\bigg(-\Phi^{-1}(1-\alpha/2) \leq \frac{\widehat{\lambda}_{*}-\lambda}{\sqrt{\text{var}(\widehat{\lambda}_{*}\mid\mathcal{Z})}}\lim_{I\rightarrow\infty}\sqrt{\frac{\text{var}(\widehat{\lambda}_{*}\mid\mathcal{Z})}{S_{*}^2(Q)}} \leq \Phi^{-1}(1-\alpha/2)\mid\mathcal{Z}  \bigg) \\ 
    &\geq \lim\limits_{I \rightarrow \infty}\text{pr}\Bigg(-\Phi^{-1}(1-\alpha/2) \leq \frac{\widehat{\lambda}_{*}-\lambda}{\sqrt{\text{var}(\widehat{\lambda}_{*}\mid\mathcal{Z})}} \leq \Phi^{-1}(1-\alpha/2)\mid\mathcal{Z}  \Bigg) \\ 
    &=\Phi(\Phi^{-1}(1-\alpha/2))-\Phi(\Phi^{-1}(\alpha/2)) \\ 
    &=1-\alpha,
\end{align*}
\vspace{-1.5cm}

where the equality in the third line comes from Lemma~\ref{lem: asymptotic normality}, Lemma~\ref{lem: valid variance}, and Slutsky's theorem, the inequality in the fourth line comes from Lemma~\ref{lem: valid variance}, and the equality in the fifth line is by Lemma~\ref{lem: asymptotic normality}.
\end{proof}

\subsection*{B.4 Proof of Proposition 2}
To prove Proposition 2, we consider the following widely used regularity conditions.

\begin{condition}[Consistent Propensity Score Estimator]\label{condition: estimated_propensity}
    The estimated propensity scores $\widehat{e}_{ij}$ obtained from a correctly specified model (parametric or nonparametric) are strongly consistent estimators of the true propensity scores $e_{ij}$, that is, $\widehat{e}_{ij}\xrightarrow{a.s.}e_{ij}$ for all $i,j$ as $N \rightarrow\infty$.
\end{condition}

\begin{condition}[Regularity of Estimated Propensity Scores]\label{condition: infimum} There exists a constant $\delta'>0$ such that, for sufficiently large $I$, we have $\widehat{e}_{ij}\in [\delta', 1-\delta^{\prime}]$ for all $i,j$. 
\end{condition}

\begin{condition}[Bounded Outcomes]\label{condition: bounded outcomes} There exists a constant $M<\infty$ such that $|Y_{ij}|\leq M$ for all $i=1,\dots,I,j=1,\dots,n_i$.
\end{condition}
Under Condition \ref{condition: bounded size of matched sets} and Assumption 2 in the main text, there exists some constant $q_{1}>0$ such that $p_{ij}\in [q_{1},1-q_{1}]$. Similarly, under Conditions \ref{condition: bounded size of matched sets}, \ref{condition: infimum} and Assumption 2 in the main text, there exists some constant $q_{2}>0$ such that $\widehat{p}_{ij}\in [q_{2},1-q_{2}]$.
Also, recall the following detailed form of Kolmogorov's strong law of large numbers.
\begin{lemma}\label{lemma: kolmogorov}
(Kolmogorov's Strong Law of Large Numbers): Suppose $X_1, X_2,\dots, X_{n}$ is an infinite sequence of independent but not identically distributed random variables, and $\text{var}(X_{k})<\infty$ for all $k=1,\dots, n$. In addition, suppose that $\lim_{n\rightarrow \infty}\sum_{k=1}^{n}k^{-2}\text{var}(X_{k})<\infty$. Then, we have $n^{-1}(X_1+X_2+\dots +X_{n})-n^{-1}E(X_1+X_2+\dots + X_{n})\xrightarrow{a.s.} 0$ as $n \to \infty$.
\end{lemma}

Next, we prove the following lemma, which states that the oracle estimator $\widehat{\lambda}_{*}$ is a consistent estimator for the sample average treatment effect $\lambda$.
\begin{lemma}\label{lemma: oracle}
Assuming independence of treatment assignments across matched sets, Assumptions 1 and 2 in the main text and Condition~\ref{condition: bounded outcomes}. As $I\rightarrow\infty$, the oracle estimator $\widehat{\lambda}_{*}$ is strongly consistent for $\lambda$, that is, $\widehat{\lambda}_{*}\xrightarrow{a.s.}\lambda$. 
\end{lemma}

\begin{proof}
    Let $U_i=\sum_{j=1}^{n_i}T_{ij}$, where $T_{ij}=\frac{Z_{ij}Y_{ij}}{p_{ij}}-\frac{(1-Z_{ij})Y_{ij}}{1-p_{ij}}-\{Y_{ij}(1)-Y_{ij}(0)\}$. Then, for any $i,j$, we have 
    \begin{equation*}
    \begin{split}
        &\Big|\frac{Z_{ij}Y_{ij}}{p_{ij}}-\frac{(1-Z_{ij})Y_{ij}}{1-p_{ij}}\Big|\leq\frac{|Y_{ij}|}{p_{ij}}+\frac{|Y_{ij}|}{1-p_{ij}}\leq \frac{2M}{q_1}, \text{which implies that } |T_{ij}|\leq\frac{2M}{q_1}+2M:=L.
    \end{split}
    \end{equation*}
    Therefore, we have $\text{var}(T_{ij})\leq E(T_{ij}^2)\leq L^2 < \infty$. By the Cauchy-Schwarz inequality, we have $U_i^2=\Big(\sum_{j=1}^{n_i}T_{ij}\Big)^2\leq n_i\sum_{j=1}^{n_i}T_{ij}^2$. Therefore, we have $\text{var}(U_i)\leq E(U_i^2)\leq n_i\sum_{j=1}^{n_i}E(T_{ij}^2)\leq n_i^2L^2\leq C_{1}^{2}L^{2}$, which implies that $\lim_{i\rightarrow\infty}\sum_{i=1}^{I}i^{-2}\text{var}(U_i)<\infty$.
    
    For each matched set $i$, we have $E(U_i|\mathcal{Z})=0$. Since $U_i$ are independent random variables, by Lemma~\ref{lemma: kolmogorov}, we have $\frac{1}{I}\sum_{i=1}^{I}U_i\xrightarrow{a.s.}0.$
    Note that $1\leq N/I \leq C_1$, so $I/N$ is a constant strictly between 0 and 1. Therefore,
    \vspace{-0.3cm}
    \begin{equation*}
        \widehat{\lambda}_*-\lambda=\frac{1}{N}\sum_{i=1}^{I}U_i=\frac{I}{N}\cdot\frac{1}{I}\sum_{i=1}^{I}U_i\xrightarrow{a.s.}0.
    \end{equation*}
    \end{proof}
\vspace{-1.3cm}
Now, we are ready to prove Proposition 2.
\begin{proof}
Note that
\vspace{-0.3cm}
\begin{equation*}
    \begin{split}
        &\left|\frac{1}{p_{ij}}-\frac{1}{\widehat{p}_{ij}}\right|=\frac{|\widehat{p}_{ij}-p_{ij}|}{p_{ij}\widehat{p}_{ij}}\leq\frac{|\widehat{p}_{ij}-p_{ij}|}{q_1q_2}, \\ 
        &\left|\frac{1}{1-p_{ij}}-\frac{1}{1-\widehat{p}_{ij}}\right|=\frac{|\widehat{p}_{ij}-p_{ij}|}{(1-p_{ij})(1-\widehat{p}_{ij})}\leq\frac{|\widehat{p}_{ij}-p_{ij}|}{q_1q_2}. \\
    \end{split}
\end{equation*}
Since $Z_{ij} \in \{0,1\}$ for any $i,j$, we have
\vspace{-0.3cm}
\begin{align*}
    &\Big|\Big(\frac{Z_{ij}Y_{ij}}{p_{ij}}-\frac{(1-Z_{ij})Y_{ij}}{1-p_{ij}}\Big)-\Big(\frac{Z_{ij}Y_{ij}}{\widehat{p}_{ij}}-\frac{(1-Z_{ij})Y_{ij}}{1-\widehat{p}_{ij}}\Big)\Big| \\ 
    =&\Big|\Big(\frac{Z_{ij}Y_{ij}}{p_{ij}}-\frac{Z_{ij}Y_{ij}}{\widehat{p}_{ij}}\Big)-\Big(\frac{(1-Z_{ij})Y_{ij}}{1-p_{ij}}-\frac{(1-Z_{ij})Y_{ij}}{1-\widehat{p}_{ij}}\Big)\Big| \\ 
    =&\Big|\frac{Z_{ij}Y_{ij}(\widehat{p}_{ij}-{p}_{ij})}{p_{ij}\widehat{p}_{ij}}-\frac{(1-Z_{ij})Y_{ij}({p}_{ij}-\widehat{p}_{ij})}{(1-p_{ij})(1-\widehat{p}_{ij})}\Big|\\
    \leq&|Y_{ij}|\frac{|\widehat{p}_{ij}-p_{ij}|}{q_1q_2}\\
     \leq& M\frac{|\widehat{p}_{ij}-p_{ij}|}{q_1q_2}.
\end{align*}
\vspace{-0.3cm}
Let $\widehat{\mathbf{p}}=(\widehat{p}_{11},\dots,\widehat{p}_{In_I})$ and $\mathbf{p}=(p_{11},\dots,p_{In_I})$ denote the estimated and true post-matching treatment assignment probability vector, we have
\vspace{-0.3cm}
\begin{equation*}
    |\widehat{\lambda}_{\diamond}-\widehat{\lambda}_*|\leq\frac{M}{Nq_1q_2}\|\widehat{\mathbf{p}}-\mathbf{p}\|_1.
\end{equation*}
\vspace{-0.8cm}
Recall that each $p_{ij}$ is a continuous function of $\mathbf{e}_i=(e_{i1},\dots,e_{in_{i}})$. By the continuous mapping theorem and Condition~\ref{condition: estimated_propensity}, we have
\vspace{-0.3cm}
\begin{equation*}
    \widehat{p}_{ij}=g_{ij}(\widehat{\mathbf{e}}_i) \xrightarrow{a.s.} g_{ij}(\mathbf{e}_i)=p_{ij}.
\end{equation*}
Thus, each $|\widehat{p}_{ij}-p_{ij}|\xrightarrow{a.s.}0$, and $|\widehat{\lambda}_{\diamond}-\widehat{\lambda}_*|\leq \frac{M}{Nq_1q_2}\|\widehat{\mathbf{p}}-\mathbf{p}\|_1\xrightarrow{a.s.}0$. Therefore, by Lemma~\ref{lemma: oracle}, we have $\widehat{\lambda}_{\diamond}\xrightarrow{a.s.} \lambda$. 
\end{proof}

\subsection*{B.5: Assessing the Finite-Sample Impact of Propensity Score Estimation Error}

In Appendix B.4, we showed that the difference between $\widehat{\lambda}_{\diamond}$ and $\widehat{\lambda}_{*}$ is bounded by the discrepancy between the estimated and true post-matching treatment assignment probabilities, which is denoted as 
\vspace{-0.3cm}
\begin{equation*}
    |\widehat{\lambda}_{\diamond}-\widehat{\lambda}_*|\leq\frac{M}{Nq_1q_2}\|\widehat{\mathbf{p}}-\mathbf{p}\|_1.
\end{equation*}
Next, we derive a bound on the difference between the oracle variance estimator $S^2_*(Q)$ and its plug-in analogue $S_{\diamond}^2(Q)$, where $S_{\diamond}^2(Q)$ is obtained by replacing the true propensity scores $p_{ij}$ with their estimates $\widehat{p}_{ij}$. Using the notations from the previous sections, let $\mathbf{u}=\mathbf{y}W$, $\mathbf{\widehat{u}}=\mathbf{\widehat{y}}W$, $\mathbf{u}_{\Delta}=\mathbf{\widehat{u}}-\mathbf{u}$, $P=\mathcal{I}-H_Q$ and $w_{max}=\text{max}_i w_i\leq IC_1/N\leq C_1/2$. Define $c_Q=\text{min}_{1\leq i \leq I}(1-h_{Qii})>0$, and, for matched set $i$, let $\widehat{\mathbf{p}}_i=(\widehat{p}_{i1},\dots,\widehat{p}_{in_i})$ and $\mathbf{p}_i=(p_{i1},\dots,p_{in_i})$ denote the estimated and true post-matching treatment assignment probability vectors respectively. Then we have:
\vspace{-0.3cm}
\begin{align*}
    &|\widehat{\lambda}_{*,i}|\leq \frac{M}{q_1}, \quad |\widehat{\lambda}_{\diamond,i}-\widehat{\lambda}_{*,i}|\leq \frac{M}{n_iq_1q_2}\|\widehat{\mathbf{p}}_i-\mathbf{p}_i\|_1; \\ 
    &|y_i|\leq \frac{M}{q_1\sqrt{c_Q}}, \quad |\widehat{y}_i-y_i|\leq \frac{M}{n_iq_1q_2\sqrt{c_Q}}\|\widehat{\mathbf{p}}_i-\mathbf{p}_i\|_1; \\ 
    &\|\mathbf{y}\|_2\leq \frac{\sqrt{I}M}{q_1\sqrt{c_Q}}, \quad \|\mathbf{\widehat{y}}-\mathbf{y}\|_2\leq \frac{M}{2q_1q_2\sqrt{c_Q}}\|\widehat{\mathbf{p}}-\mathbf{p}\|_1; \\ 
    &\|\mathbf{u}\|_2\leq \frac{C_1\sqrt{I}M}{2q_1\sqrt{c_Q}}, \quad \|\mathbf{u}_{\Delta}\|=\|\mathbf{\widehat{u}}-\mathbf{u}\|_2\leq \frac{IM}{Nq_1q_2\sqrt{c_Q}}\|\widehat{\mathbf{p}}-\mathbf{p}\|_1. \\
\end{align*}
Because $P$ is symmetric and idempotent, its operator norm satisfies $\|P\|_{op}=1$. Therefore,
\vspace{-1.3cm}
\begin{align*}
   |\mathbf{u}P\mathbf{u}_{\Delta}^T|&=|\mathbf{u}(P\mathbf{u}_{\Delta}^T)|=|\langle \mathbf{u}^T,P\mathbf{u}_{\Delta}^T\rangle| \\ 
    &\leq\|\mathbf{u}\|_2\cdot\|P\mathbf{u}_{\Delta}^T\|_2 \quad  \\     &\leq\|P\|_{op}\|\mathbf{u}\|_2\cdot\|\mathbf{u}_{\Delta}\|_2=\|\mathbf{u}\|_2\cdot\|\mathbf{u}_{\Delta}\|_2. \\ 
\end{align*}
\vspace{-2.3cm}
\\ Similarly, we can show that $|\mathbf{u}_{\Delta}P\mathbf{u}_{\Delta}^T|\leq\|\mathbf{u}_{\Delta}\|_2^2$. Then, we obtain
\vspace{-0.5cm}
\begin{align*}
    |S_{\diamond}^2(Q)-S_{*}^2(Q)|&=|I^{-2}(\mathbf{\widehat{u}}P\mathbf{\widehat{u}}^T-\mathbf{u}P\mathbf{u}^T)| \\
    &=|I^{-2}((\mathbf{u}+\mathbf{u}_{\Delta})P(\mathbf{u}+\mathbf{u}_{\Delta})^T-\mathbf{u}P\mathbf{u}^T)| \\ 
    &=|I^{-2}(2\mathbf{u}P\mathbf{u}_{\Delta}^T+\mathbf{u}_{\Delta}P\mathbf{u}_{\Delta}^T)| \\ 
    &\leq I^{-2}(2|\mathbf{u}P\mathbf{u}_{\Delta}^T|+|\mathbf{u}_{\Delta}P\mathbf{u}_{\Delta}^T|) \\
    &\leq I^{-2}(2\|\mathbf{u}\|_2\cdot\|\mathbf{u}_{\Delta}\|_2+\|\mathbf{u}_{\Delta}\|_2^2). \\
\end{align*}
\vspace{-2.2cm}
\\ Let $\Delta_N=\frac{1}{N}\|\widehat{\mathbf{p}}-\mathbf{p}\|_1$. Plugging $\Delta_N$ in the bounds derived above yields
\vspace{-0.3cm}
\begin{equation*}
    |S_{\diamond}^2(Q)-S_{*}^2(Q)|\leq\frac{C_1M^2}{q_1^2q_2c_Q}I^{-1/2}\Delta_N+\frac{M^2}{q_1^2q_2^2c_Q}\Delta_N^2.
\end{equation*}
\vspace{-0.8cm}
\\ Next, using the inequality $|\sqrt{a}-\sqrt{b}|\leq\sqrt{|a-b|}$, we obtain
\begin{equation*}
    |S_{\diamond}(Q)-S_{*}(Q)|\leq\sqrt{\frac{C_1M^2}{q_1^2q_2c_Q}I^{-1/2}\Delta_N+\frac{M^2}{q_1^2q_2^2c_Q}\Delta_N^2}.
\end{equation*}
\vspace{-0.8cm}
\\ Recall that $CI^{\lambda}_{*}=[\widehat{\lambda}_* - \Phi^{-1}(1-\alpha/2) \times S_{*}(Q), \widehat{\lambda}_{*} + \Phi^{-1}(1-\alpha/2) \times S_{*}(Q)]$ denotes the confidence interval reported by the oracle IPPW estimator and $CI^{\lambda}_{\diamond}=[\widehat{\lambda}_{\diamond} - \Phi^{-1}(1-\alpha/2) \times S_{\diamond}(Q), \widehat{\lambda}_{\diamond} + \Phi^{-1}(1-\alpha/2) \times S_{\diamond}(Q)]$ denotes the confidence interval based on the plug-in IPPW estimator. Then, the absolute difference between the corresponding lower endpoints of $CI^{\lambda}_{*}$ and $CI^{\diamond}_{*}$ is bounded by:
\vspace{-0.5cm}
\begin{equation*}
\left|(\widehat{\lambda}_{\diamond}-\widehat{\lambda}_{*})-\Phi^{-1}({1-\alpha/2})\cdot(S_{\diamond}(Q)-S_{*}(Q))\right|\leq\frac{M}{q_1q_2}\Delta_N+\Phi^{-1}(1-\alpha/2)\sqrt{\frac{C_1M^2}{q_1^2q_2c_Q}I^{-1/2}\Delta_N+\frac{M^2}{q_1^2q_2^2c_Q}\Delta_N^2}.
\end{equation*}
Similarly, the absolute difference between the corresponding upper endpoints of $CI^{\lambda}_{*}$ and $CI^{\diamond}_{*}$ is bounded by:
\vspace{-0.5cm}
\begin{equation*}
\left|(\widehat{\lambda}_{\diamond}-\widehat{\lambda}_{*})+\Phi^{-1}({1-\alpha/2})\cdot (S_{\diamond}(Q)-S_{*}(Q))\right|\leq\frac{M}{q_1q_2}\Delta_N+\Phi^{-1}(1-\alpha/2)\sqrt{\frac{C_1M^2}{q_1^2q_2c_Q}I^{-1/2}\Delta_N+\frac{M^2}{q_1^2q_2^2c_Q}\Delta_N^2}.
\end{equation*}

\subsection*{B.6: Proofs of Theorems 2 and 3}

To prove Theorems 2 and 3, we need some regularity conditions. We first define the limiting estimand $\gamma^*$:
\begin{equation}\label{eqn: limit}
    \boldsymbol{\gamma}^*=\lim_{I \rightarrow \infty}\frac{1}{I} \sum_{i=1}^{I}E\Big[\boldsymbol\psi^{full}(\mathbf{O}_i,\theta,\nu,\nu')|\mathbf{Y}(1),\mathbf{Y}(0),\mathbf{X}\Big]
\end{equation}
We consider the following common regularity conditions in finite-population M-estimation theory \citep{xu2021mestimator, han2024introduction}:
\begin{condition}
    The limiting estimand~(\ref{eqn: limit}) exists and is uniquely minimized at $\boldsymbol{\gamma}^*$.
\end{condition}
\begin{condition}
    We assume that $\gamma =(\theta,\nu,\nu') \in \Gamma$ for some compact set $\Gamma$.
\end{condition}
\begin{condition}
    If $\boldsymbol\psi_i^{full}(\mathbf{O}_i,\gamma)=\nabla_\gamma\boldsymbol m_i^{full}(\mathbf{O}_i,\gamma)$, then $\boldsymbol m_i^{full}(u,\gamma)$ is continuous in $\gamma$ for all $u$ in the support of $\mathbf{O}_i$ for all $i$.
\end{condition}
\begin{condition}
    $\sup\limits_{i}E\big[\sup\limits_{\gamma \in \Gamma}|\boldsymbol m_i^{full}(\mathbf{O}_i,\gamma)|^2\big]<\infty$.
\end{condition}
\begin{condition}
   Let $\mathcal{O}_i$ denote the support of $\mathbf{O}_i$. There is $h(w)\downarrow0$ as $w\downarrow0$ and $c_1(\cdot):\mathcal{O}_i\rightarrow \mathbbm{R} $ such that $\sup\limits_i E\Big[c_{1,i}(\mathbf{O}_i)\Big]<\infty$, and for all $\tilde{\gamma}, \gamma \in \Gamma$, $|\boldsymbol m_i^{full}(\mathbf{O}_i,\tilde{\gamma})-\boldsymbol m_i^{full}(\mathbf{O}_i,\gamma)|\leq c_{1,i}(\mathbf{O}_i)h(\|\tilde{\gamma}-\gamma\|)$.
\end{condition}
\begin{condition}
    $\boldsymbol{\gamma}^* \in int(\Gamma)$.
\end{condition}
\begin{condition}
    For all $i$, we have $\boldsymbol m_i^{full}(u,\gamma)$ is twice differentiable on $int(\Gamma)$ for all $u$ in the support of $\mathbf{O}_i$.
\end{condition}
\begin{condition}
    $\sup\limits_{i}E\big[\sup\limits_{\gamma \in \Gamma}\|\boldsymbol \psi_i^{full}(\mathbf{O}_i,\gamma)\|^4\big]<\infty$.
\end{condition}
\begin{condition}
    $\frac{1}{I}\sum_{i=1}^{I}\boldsymbol \psi_i^{full}(\mathbf{O}_i,\widehat{\gamma})=o_p(N^{-1/2})$.
\end{condition}
\begin{condition}
    $A(\theta_0,\nu_0,\nu_0')$ is nonsingular.
\end{condition}
\begin{condition}
    $\sup\limits_{i}E\big[\sup\limits_{\gamma \in \Gamma}\|\nabla_{\gamma}\boldsymbol \psi_i^{full}(\mathbf{O}_i,\gamma)\|^2\big]<\infty$.
\end{condition}
\begin{condition}
    There is $h(w)\downarrow0$ as $w\downarrow0$, and $c_2(\cdot):\mathcal{O}_i\rightarrow \mathbbm{R}$ such that $\sup\limits_i E\Big[c_{2,i}(\mathbf{O}_i)\Big]<\infty$, and for all $\tilde{\gamma}, \gamma \in \Gamma$, $\|\nabla_{\gamma}\boldsymbol \psi_i^{full}(\mathbf{O}_i,\tilde{\gamma})-\nabla_{\gamma}\boldsymbol \psi_i^{full}(\mathbf{O}_i,\gamma)\|\leq c_{2,i}(\mathbf{O}_i)h(\|\tilde{\gamma}-\gamma\|)$.
\end{condition}
\begin{condition}
    There is $h(w)\downarrow0$ as $w\downarrow0$, and $c_3(\cdot):\mathcal{O}_i\rightarrow \mathbbm{R} $ such that $\sup\limits_i E\Big[c_{3,i}(\mathbf{O}_i)^2\Big]<\infty$, and for all $\tilde{\gamma}, \gamma \in \Gamma$, $\|\boldsymbol \psi_i^{full}(\mathbf{O}_i,\tilde{\gamma})-\boldsymbol \psi_i^{full}(\mathbf{O}_i,\gamma)\|\leq c_{3,i}(\mathbf{O}_i)h(\|\tilde{\gamma}-\gamma\|)$.
\end{condition}


Next, we present a valid variance estimator for $\widehat{\lambda}_{\diamond}$ under the finite-population M-estimation framework \citep{xu2021mestimator, han2024introduction}, which jointly accounts for the uncertainties stemming from estimating the propensity score model parameters and those stemming from estimating the treatment-specific finite-population means. Throughout this section, we use the notations introduced in Section 3.2 of the main text.

Assume that the propensity scores are estimated using an M-estimation approach. Let $\theta_0 \in \mathbbm{R}^p$ denote the $p$-dimensional true parameter vector of the propensity score model. For each matched set $i$, we can define $\boldsymbol{\psi}(\mathbf{O}_i,\theta)=(\psi_1(\mathbf{O}_i,\theta),\cdots,\psi_p(\mathbf{O}_i,\theta))^T$ as the vector of estimating functions associated with the M-estimator. The estimator $\widehat{\theta}$ is then defined implicitly as the solution to the system of estimating equations $\sum_{i=1}^{I}\boldsymbol{\psi}(\mathbf{O}_i,\theta)=\mathbf{0}$. To incorporate estimation of the sample average treatment effect within this framework, we augment the system with two additional components. Specifically, let $\nu_0=\frac{1}{N}\sum_{i=1}^{I}\sum_{j=1}^{n_i}Y_{ij}(1)$ and $\nu'_0=\frac{1}{N}\sum_{i=1}^{I}\sum_{j=1}^{n_i}Y_{ij}(0)$ denote the finite-population means of the potential outcomes under treatment and control, respectively. These two quantities determine the sample average treatment, defined as $\lambda=\nu_0-\nu'_0$. To estimate $\nu_0$ and $\nu_0'$, we decompose the IPPW estimator we introduced in Section 3.1 in the main text into treated and control components and rewrite them in a form indexed by the covariates $\mathbf{x}_{ij}$ and the propensity score model parameters $\theta$. Then, the post-matching finite-population M-estimation framework is formulated as follows:

\begin{align*}
\sum_{i=1}^{I}
\left(
\begin{array}{c}
  \boldsymbol{\psi}(\mathbf{O}_i,\theta) \\[0.8em]
  \nu-\frac{I}{N} \sum_{j=1}^{n_i}\frac{Z_{ij}Y_{ij}}{p_{ij}} \\[0.8em]
  \nu'-\frac{I}{N} \sum_{j=1}^{n_i}\frac{(1-Z_{ij})Y_{ij}}{1-p_{ij}}
\end{array}
\right)&=\sum_{i=1}^{I'}
\left(
\begin{array}{c}
  \boldsymbol{\psi}(\mathbf{O}_i,\theta) \\[0.8em]
  \nu-\frac{I}{N} \sum_{j=1}^{n_i}\frac{Z_{ij}Y_{ij}\sum_{j'=1}^{n_i}\text{odds}\{g(\mathbf{x}_{ij'};\theta)\}}{\text{odds}\{g(\mathbf{x}_{ij};\theta)\}} \\[0.8em]
  \nu'-\frac{I}{N} \sum_{j=1}^{n_i}\frac{(1-Z_{ij})Y_{ij}\sum_{j'=1}^{n_i}\text{odds}\{g(\mathbf{x}_{ij'};\theta)\}}{\sum_{j'=1}^{n_i}\text{odds}\{g(\mathbf{x}_{ij'};\theta)\}-\text{odds}\{g(\mathbf{x}_{ij};\theta)\}}
\end{array}
\right)\\
&\quad \quad \quad +\sum_{i=I'+1}^{I}
\left(
\begin{array}{c}
  \boldsymbol{\psi}(\mathbf{O}_i,\theta) \\[0.8em]
  \nu-\frac{I}{N} \sum_{j=1}^{n_i}\frac{Z_{ij}Y_{ij}\sum_{j'=1}^{n_i}\text{odds}\{1-g(\mathbf{x}_{ij'};\theta)\}}{\sum_{j'=1}^{n_i}\text{odds}\{1-g(\mathbf{x}_{ij'};\theta)\}-\text{odds}\{1-g(\mathbf{x}_{ij};\theta)\}} \\[0.8em]
  \nu'-\frac{I}{N} \sum_{j=1}^{n_i}\frac{(1-Z_{ij})Y_{ij}\sum_{j'=1}^{n_i}\text{odds}\{1-g(\mathbf{x}_{ij'};\theta)\}}{\text{odds}\{1-g(\mathbf{x}_{ij};\theta)\}}
\end{array}
\right) \\
&=\mathbf{0}.
\end{align*}
The parameter estimates $\widehat{\theta},\widehat{\nu},\widehat{\nu}'$ are obtained by solving the above system of equations. For each matched set $i$, let $\boldsymbol\psi^{full}(\mathbf{O}_i,\theta,\nu,\nu')$ represent the corresponding estimating equations. Next, we aim to estimate the asymptotic variance of $\widehat{\nu}-\widehat{\nu}'$, denoted by $\widehat{\lambda}_{\diamond}$, while properly accounting for the uncertainty in estimation of the propensity scores. As introduced in Section 3.2 of the main text, we apply the finite-population sandwich variance estimation approach for finite-population M-estimators, following the formulation in \citet{xu2021mestimator}. The corresponding variance matrix is given by:
\begin{equation*}
V(\theta_0,\nu_0.\nu'_0)=A(\theta_0,\nu_0.\nu'_0)^{-1}B(\theta_0,\nu_0.\nu'_0)[A(\theta_0,\nu_0.\nu'_0)^{-1}]^T\in \mathbbm{R}^{(p+2)\times(p+2)},
\end{equation*}
where $A(\cdot)$ and $B(\cdot)$ are defined as:
\begin{align*}
    &A(\theta_0,\nu_0.\nu'_0)=\lim_{I \rightarrow \infty} \frac{1}{I} \sum_{i=1}^{I}E\Big[-\nabla_{\theta,\nu,\nu',}\boldsymbol\psi^{full}(\mathbf{O}_i,\theta_0,\nu_0,\nu'_0)\Big]\in \mathbbm{R}^{(p+2)\times(p+2)}, \\    
    &B(\theta_0,\nu_0.\nu'_0)=\lim_{I \rightarrow \infty} \frac{1}{I} \sum_{i=1}^{I}E\Big[\boldsymbol\psi^{full}(\mathbf{O}_i,\theta_0,\nu_0,\nu'_0)\boldsymbol\psi^{full}(\mathbf{O}_i,\theta_0,\nu_0,\nu'_0)^T\Big]\in \mathbbm{R}^{(p+2)\times(p+2)}.
\end{align*}
Next, we derive explicit expressions for $A(\cdot)$ and $B(\cdot)$:
\begin{align*}
    &A(\theta_0,\nu_0.\nu'_0)= \lim_{I \rightarrow \infty} 
    \left(
    \begin{matrix}
        A_{11} & 0 & 0 \\
        A_{21} & -1 & 0 \\
        A_{31} & 0 & -1
    \end{matrix}
    \right)_{(p+2)\times(p+2)}, \\
    &B(\theta_0,\nu_0.\nu'_0)=\lim_{I \rightarrow \infty}
    \left(
    \begin{matrix}
        B_{\theta\theta} & B_{\theta\nu} & B_{\theta\nu'} \\ 
        B_{\theta\nu}^T & B_{\nu\nu} & B_{\nu\nu'} \\ 
        B_{\theta\nu'}^T & B_{\nu\nu'}^T & B_{\nu'\nu'}
    \end{matrix}
    \right)_{(p+2)\times(p+2)},
\end{align*}
where
\begin{equation*}
\begin{split}
    &A_{11}=-\frac{1}{I}\sum_{i=1}^{I}E[\nabla_\theta\boldsymbol\psi(\mathbf{O_i},\theta_0)], \\
    &A_{21}=\frac{1}{N}\sum_{i=1}^{I'} \sum_{j=1}^{n_i}E\bigg[\nabla_\theta\frac{Z_{ij}Y_{ij}\sum_{j'=1}^{n_i}\text{odds}\{g(\mathbf{x}_{ij'};\theta_0)\}}{\text{odds}\{g(\mathbf{x}_{ij};\theta_0)\}}\bigg] \\ 
    &\hspace{1cm}+\frac{1}{N}\sum_{i=I'+1}^{I} \sum_{j=1}^{n_i}E\bigg[\nabla_\theta\frac{Z_{ij}Y_{ij}\sum_{j'=1}^{n_i}\text{odds}\{1-g(\mathbf{x}_{ij'};\theta_0)\}}{\sum_{j'=1}^{n_i}\text{odds}\{1-g(\mathbf{x}_{ij'};\theta_0)\}-\text{odds}\{1-g(\mathbf{x}_{ij};\theta_0)\}}\bigg], \\
    &A_{31}=\frac{1}{N}\sum_{i=1}^{I'}
     \sum_{j=1}^{n_i}E\bigg[\nabla_\theta\frac{(1-Z_{ij})Y_{ij}\sum_{j'=1}^{n_i}\text{odds}\{g(\mathbf{x}_{ij'};\theta_0)\}}{\sum_{j'=1}^{n_i}\text{odds}\{g(\mathbf{x}_{ij'};\theta_0)\}-\text{odds}\{g(\mathbf{x}_{ij};\theta_0)\}}\bigg] \\  
     &\hspace{1cm}+\frac{1}{N}\sum_{i=I'+1}^{I}\sum_{j=1}^{n_i}E\bigg[\nabla_\theta\frac{(1-Z_{ij})Y_{ij}\sum_{j'=1}^{n_i}\text{odds}\{1-g(\mathbf{X_{ij'}};\theta_0)\}}{\text{odds}\{1-g(\mathbf{X_{ij}};\theta_0)\}}\bigg], \\
    &B_{\theta\theta}=\frac{1}{I}\sum_{i=1}^{I}E\Big\{\boldsymbol{\psi}(\mathbf{O}_i,\theta_0)\boldsymbol{\psi}(\mathbf{O}_{i},\theta_0)^T\Big\}, \\
    &B_{\theta\nu}=\frac{1}{I}\sum_{i=1}^{I'}E\Bigg\{\boldsymbol{\psi}(\mathbf{O}_i,\theta_0)\Bigg[\nu_0-\frac{I}{N} \sum_{j=1}^{n_i}\frac{Z_{ij}Y_{ij}\sum_{j'=1}^{n_i}\text{odds}\{g(\mathbf{x}_{ij'};\theta_0)\}}{\text{odds}\{g(\mathbf{x}_{ij};\theta_0)\}}\Bigg]\Bigg\} \\
    &\hspace{1cm}+\frac{1}{I}\sum_{i=I'+1}^{I}E\Bigg\{\boldsymbol{\psi}(\mathbf{O}_i,\theta_0)\Bigg[\nu_0-\frac{I}{N} \sum_{j=1}^{n_i}\frac{Z_{ij}Y_{ij}\sum_{j'=1}^{n_i}\text{odds}\{1-g(\mathbf{x}_{ij'};\theta_0)\}}{\sum_{j'=1}^{n_i}\text{odds}\{1-g(\mathbf{x}_{ij'};\theta_0)\}-\text{odds}\{1-g(\mathbf{x}_{ij};\theta_0)\}}\Bigg]\Bigg\}, \\
    &B_{\theta\nu'}=\frac{1}{I}\sum_{i=1}^{I'}E\Bigg\{\boldsymbol{\psi}(\mathbf{O}_i,\theta_0)\Bigg[\nu'_0-\frac{I}{N} \sum_{j=1}^{n_i}\frac{(1-Z_{ij})Y_{ij}\sum_{j'=1}^{n_i}\text{odds}\{g(\mathbf{x}_{ij'};\theta_0)\}}{\sum_{j'=1}^{n_i}\text{odds}\{g(\mathbf{x}_{ij'};\theta_0)\}-\text{odds}\{g(\mathbf{x}_{ij};\theta_0)\}}\Bigg]\Bigg\} \\
    &\hspace{1cm}+\frac{1}{I}\sum_{i=I'+1}^{I}E\Bigg\{\boldsymbol{\psi}(\mathbf{O}_i,\theta_0)\Bigg[\nu'_0-\frac{I}{N} \sum_{j=1}^{n_i}\frac{(1-Z_{ij})Y_{ij}\sum_{j'=1}^{n_i}\text{odds}\{1-g(\mathbf{x}_{ij'};\theta_0)\}}{\text{odds}\{1-g(\mathbf{x}_{ij};\theta_0)\}}\Bigg]\Bigg\}, \\
        \end{split}
\end{equation*}
\begin{equation*}
\begin{split}
    &B_{\nu\nu}=\frac{1}{I}\sum_{i=1}^{I'}E\Bigg[\nu_0-\frac{I}{N} \sum_{j=1}^{n_i}\frac{Z_{ij}Y_{ij}\sum_{j'=1}^{n_i}\text{odds}\{g(\mathbf{x}_{ij'};\theta_0)\}}{\text{odds}\{g(\mathbf{x}_{ij};\theta_0)\}}\Bigg]^2 \\
    &\hspace{1cm}+\frac{1}{I}\sum_{i=I'+1}^{I}E\Bigg[\nu_0-\frac{I}{N} \sum_{j=1}^{n_i}\frac{Z_{ij}Y_{ij}\sum_{j'=1}^{n_i}\text{odds}\{1-g(\mathbf{x}_{ij'};\theta_0)\}}{\sum_{j'=1}^{n_i}\text{odds}\{1-g(\mathbf{x}_{ij'};\theta_0)\}-\text{odds}\{1-g(\mathbf{x}_{ij};\theta_0)\}}\Bigg]^2, \\
    &B_{\nu\nu'}=\frac{1}{I}\sum_{i=1}^{I'}E\Bigg\{\Bigg[\nu_0-\frac{I}{N} \sum_{j=1}^{n_i}\frac{Z_{ij}Y_{ij}\sum_{j'=1}^{n_i}\text{odds}\{g(\mathbf{x}_{ij'};\theta_0)\}}{\text{odds}\{g(\mathbf{x}_{ij};\theta_0)\}}\Bigg] \\ 
    &\hspace{2.6cm}\times \Bigg[\nu'_0-\frac{I}{N} \sum_{j=1}^{n_i}\frac{(1-Z_{ij})Y_{ij}\sum_{j'=1}^{n_i}\text{odds}\{g(\mathbf{x}_{ij'};\theta_0)\}}{\sum_{j'=1}^{n_i}\text{odds}\{g(\mathbf{x}_{ij'};\theta_0)\}-\text{odds}\{g(\mathbf{x}_{ij};\theta_0)\}}\Bigg]\Bigg\}\\
    &\hspace{1cm}+\frac{1}{I}\sum_{i=I'+1}^{I}E\Bigg\{\Bigg[\nu_0-\frac{I}{N} \sum_{j=1}^{n_i}\frac{Z_{ij}Y_{ij}\sum_{j'=1}^{n_i}\text{odds}\{1-g(\mathbf{x}_{ij'};\theta_0)\}}{\sum_{j'=1}^{n_i}\text{odds}\{1-g(\mathbf{x}_{ij'};\theta_0)\}-\text{odds}\{1-g(\mathbf{x}_{ij};\theta_0)\}}\Bigg]\\
    &\hspace{2.6cm}\times \Bigg[\nu'_0-\frac{I}{N} \sum_{j=1}^{n_i}\frac{(1-Z_{ij})Y_{ij}\sum_{j'=1}^{n_i}\text{odds}\{1-g(\mathbf{x}_{ij'};\theta_0)\}}{\text{odds}\{1-g(\mathbf{x}_{ij};\theta_0)\}}\Bigg]\Bigg\}, \\
    &B_{\nu'\nu'}=\frac{1}{I}\sum_{i=1}^{I'}E\Bigg[\nu'_0-\frac{I}{N} \sum_{j=1}^{n_i}\frac{(1-Z_{ij})Y_{ij}\sum_{j'=1}^{n_i}\text{odds}\{g(\mathbf{x}_{ij'}\;\theta_0)\}}{\sum_{j'=1}^{n_i}\text{odds}\{g(\mathbf{x}_{ij'};\theta_0)\}-\text{odds}\{g(\mathbf{x}_{ij};\theta_0)\}}\Bigg]^2 \\
    &\hspace{1cm}+\frac{1}{I}\sum_{i=I'+1}^{I}E\Bigg[\nu'_0-\frac{I}{N} \sum_{j=1}^{n_i}\frac{(1-Z_{ij})Y_{ij}\sum_{j'=1}^{n_i}\text{odds}\{1-g(\mathbf{x}_{ij'};\theta_0)\}}{\text{odds}\{1-g(\mathbf{x}_{ij};\theta_0)\}}\Bigg]^2. \\
    \end{split}
\end{equation*}
We now calculate $A(\theta_0,\nu_0.\nu'_0)^{-1}$. We proceed by representing $A(\theta_0,\nu_0.\nu'_0)$ as a block matrix:
\begin{equation*}
    A(\theta_0,\nu_0.\nu'_0)=\lim_{I \rightarrow \infty} 
    \left(
    \begin{matrix}
        M & \mathbf{0}  \\
        N & D \\
    \end{matrix}
    \right)_{(p+2)\times(p+2)},
\end{equation*}
\vspace{-0.5cm}
where
\begin{equation*}
\begin{split}
    &M=A_{11}, \\
    &N=
    \left(
        \begin{matrix}
        A_{21} \\ 
        A_{31}
        \end{matrix}
    \right), \\ 
    &D=
    \left(
        \begin{matrix}
        -1 & 0 \\ 
        0 & -1
        \end{matrix}
    \right). \\ 
\end{split}
\end{equation*}
Then, by the Schur complement,
\begin{equation*}
\begin{split}
    A(\theta_0,\nu_0.\nu'_0)^{-1}=&\lim_{I \rightarrow \infty} 
    \left(
    \begin{matrix}
        M^{-1} & \mathbf{0}  \\
        -D^{-1}NM^{-1} & D^{-1} \\
    \end{matrix}
    \right)_{(p+2)\times(p+2)} \\
    =&\lim_{I \rightarrow \infty} 
    \left(
    \begin{matrix}
        A_{11}^{-1} & 0 & 0  \\
        A_{21}A_{11}^{-1} & -1 & 0 \\
        A_{31}A_{11}^{-1} & 0 & -1 \\
    \end{matrix}
    \right)_{(p+2)\times(p+2)}.
\end{split}
\end{equation*}
Then, we use the entry $((p+1),(p+1))$, $((p+2),(p+2))$, $((p+1),(p+2))$ of the matrix $V(\theta_0,\nu_0.\nu'_0)$ to estimate the variance of $\widehat{\lambda}_{\diamond}$:
\begin{equation*}
    \widehat{\text{var}}(\widehat{\lambda}_{\diamond})=\frac{1}{I}V(\theta_0,\nu_0.\nu'_0)_{(p+1),(p+1)}+\frac{1}{I}V(\theta_0,\nu_0.\nu'_0)_{(p+2),(p+2)}-\frac{2}{I}V(\theta_0,\nu_0.\nu'_0)_{(p+1),(p+2)},
\end{equation*}
where
\begin{equation*}
\begin{split}
    &V(\theta_0,\nu_0.\nu'_0)_{(p+1),(p+1)}=\Big[\lim_{I \rightarrow \infty}A_{21}A_{11}^{-1},-1,0\Big]B(\theta_0,\nu_0.\nu'_0)\Big[\lim_{I \rightarrow \infty}A_{21}A_{11}^{-1},-1,0\Big]^T \\
    &\hspace{3.35cm}=\lim_{I \rightarrow \infty}\Big(B_{\nu\nu}-2A_{21}A_{11}^{-1}B_{\theta\nu}+A_{21}A_{11}^{-1}B_{\theta\theta}(A_{21}A_{11}^{-1})^T\Big), \\
    &V(\theta_0,\nu_0.\nu'_0)_{(p+2),(p+2)}=\Big[\lim_{I \rightarrow \infty}A_{31}A_{11}^{-1},0,-1\Big]B(\theta_0,\nu_0.\nu'_0)\Big[\lim_{I \rightarrow \infty}A_{31}A_{11}^{-1},0,-1\Big]^T \\
    &\hspace{3.35cm}=\lim_{I \rightarrow \infty}\Big(B_{\nu'\nu'}-2A_{31}A_{11}^{-1}B_{\theta\nu'}+A_{31}A_{11}^{-1}B_{\theta\theta}(A_{31}A_{11}^{-1})^T\Big), \\
    &V(\theta_0,\nu_0.\nu'_0)_{(p+1),(p+2)}=\Big[\lim_{I \rightarrow \infty}A_{21}A_{11}^{-1},-1,0\Big]B(\theta_0,\nu_0.\nu'_0)\Big[\lim_{I \rightarrow \infty}A_{31}A_{11}^{-1},0,-1\Big]^T \\
    &\hspace{3.35cm}=\lim_{I \rightarrow \infty}\Big(B_{\nu\nu'}-A_{21}A_{11}^{-1}B_{\theta\nu'}-A_{31}A_{11}^{-1}B_{\theta\nu}+A_{21}A_{11}^{-1}B_{\theta\theta}(A_{31}A_{11}^{-1})^T\Big). \\
\end{split}
\end{equation*}
Note that the oracle form of $\widehat{\text{var}}(\widehat{\lambda}_{\diamond})$ is not directly available in practice, as it depends on the unknown true parameters, the moments of the derivatives of the $\boldsymbol{\psi}$ functions, and the corresponding limiting quantities as $I \rightarrow \infty$. To obtain a sample version, we replace these unknown quantities with their sample estimates. Specifically, let the resulting sample-based variance matrix be denoted by $\widehat{V}(\widehat{\theta},\widehat{\nu},\widehat{\nu}')$. Within $\widehat{V}(\widehat{\theta},\widehat{\nu},\widehat{\nu}')$, we have
\begin{align*}
    &\widehat{A}(\widehat{\theta},\widehat{\nu},\widehat{\nu}')=-\frac{1}{I} \sum_{i=1}^{I}\nabla_{\theta,\nu,\nu',}\boldsymbol\psi^{full}(\widehat{\theta},\widehat{\nu},\widehat{\nu}'), \\    
    &\widehat{B}(\widehat{\theta},\widehat{\nu},\widehat{\nu}')=\frac{1}{I} \sum_{i=1}^{I}\boldsymbol\psi^{full}(\widehat{\theta},\widehat{\nu},\widehat{\nu}')\boldsymbol\psi^{full}(\widehat{\theta},\widehat{\nu},\widehat{\nu}')^T.
\end{align*}
Therefore, we obtain the following variance estimator of $\widehat{\lambda}_{\diamond}$, denoted by $S^2_{\mathcal{M}}$. 
\begin{equation*}
    S^2_{\mathcal{M}}=\frac{1}{I}\widehat{V}(\widehat{\theta},\widehat{\nu},\widehat{\nu}')_{(p+1),(p+1)}+\frac{1}{I}\widehat{V}(\widehat{\theta},\widehat{\nu},\widehat{\nu}')_{(p+2),(p+2)}-\frac{2}{I}\widehat{V}(\widehat{\theta},\widehat{\nu},\widehat{\nu}')_{(p+1),(p+2)},
\end{equation*}
where
\begin{equation*}
\begin{split}
    &\widehat{V}(\widehat{\theta},\widehat{\nu},\widehat{\nu}')_{(p+1),(p+1)}=\Big[\widehat{A}_{21}\widehat{A}_{11}^{-1},-1,0\Big]\widehat{B}(\widehat{\theta},\widehat{\nu},\widehat{\nu}')\Big[\widehat{A}_{21}\widehat{A}_{11}^{-1},-1,0\Big]^T, \\
    &\widehat{V}(\widehat{\theta},\widehat{\nu},\widehat{\nu}')_{(p+2),(p+2)}=\Big[\widehat{A}_{31}\widehat{A}_{11}^{-1},0,-1\Big]\widehat{B}(\widehat{\theta},\widehat{\nu},\widehat{\nu}')\Big[\widehat{A}_{31}\widehat{A}_{11}^{-1},0,-1\Big]^T, \\
    &\widehat{V}(\widehat{\theta},\widehat{\nu},\widehat{\nu}')_{(p+1),(p+2)}=\Big[\widehat{A}_{21}\widehat{A}_{11}^{-1},-1,0\Big]\widehat{B}(\widehat{\theta},\widehat{\nu},\widehat{\nu}')\Big[\widehat{A}_{31}\widehat{A}_{11}^{-1},0,-1\Big]^T. \\
\end{split}
\end{equation*}
Here, $\widehat{A}_{11}$, $\widehat{A}_{21}$, $\widehat{A}_{31}$ and $\widehat{B}(\widehat{\theta},\widehat{\nu},\widehat{\nu}')$ denote the sample analogues of $A_{11}$, $A_{21}$, $A_{31}$ and $B(\theta_0,\nu_0,\nu'_0)$ defined above. Under independence of treatment assignments across matched sets, along with Assumptions 1 and 2 in the main text and Conditions S8--S20, Theorems 2 and 3 in the main text can be established via applying Theorem 2.2 of \citet{xu2021mestimator} to the above derivations.

\subsection*{B.7: Proof of Theorem 4}

\begin{proof}
Let $a=(0, \dots, 0, -1, 1)\in \mathbbm{R}^{p+2}$ be a vector that has value $1$ in the coordinate corresponding to $\nu$, value $-1$ in the coordinate corresponding to $\nu'$, and zeros elsewhere. Then, we have $\widehat{\lambda}_{\diamond}-\lambda=a(\widehat{\boldsymbol{\gamma}}-\boldsymbol{\gamma}_{0})^{T}$. By Theorem 2 in the main text and the continuous mapping theorem, $\sqrt{I}(\widehat{\lambda}_{\diamond}-\lambda)=a\sqrt{I}(\widehat{\boldsymbol{\gamma}}-\boldsymbol{\gamma}_{0})^{T}\xrightarrow{d}N(0,aV_{fp}a^{T}) \ \text{as} \ I\rightarrow\infty$. Consequently, we have 
\begin{equation*}
        \frac{\widehat{\lambda}_{\diamond}-\lambda}{\sqrt{\text{var}(\widehat{\lambda}_{\diamond}|\mathcal{Z})}} =\frac{\widehat{\lambda}_{\diamond}-\lambda}{\sqrt{(aV_{fp}a^{T})/I}}\xrightarrow{d}N(0,1) \ \text{as} \ I\rightarrow\infty.
\end{equation*}

Next, by Theorem 3 in the main text,  as $I\rightarrow\infty$, we have
\begin{equation*}
    \frac{\text{var}(\widehat{\lambda}_{\diamond}|\mathcal{Z})}{S^2_{\mathcal{M}}} \xrightarrow{p} \frac{\text{var}(\widehat{\lambda}_{\diamond}|\mathcal{Z})}{\widehat{\text{var}}(\widehat{\lambda}_{\diamond}|\mathcal{Z})} \in (0,1).
\end{equation*}
Finally, by Slutsky's theorem, we have
\vspace{-0.3cm}
\begin{align*}
    &\quad \ \lim_{I \rightarrow \infty}\text{pr}\big(\lambda \in CI^{\lambda}_{\mathcal{M}}|\mathcal{Z}\big) \\
    &=\lim\limits_{I \rightarrow \infty}\text{pr}\Big(\widehat{\lambda}_{\diamond}-\Phi^{-1}(1-\alpha/2)\sqrt{S_{\mathcal{M}}^2} \leq \lambda \leq \widehat{\lambda}_{\diamond}+\Phi^{-1}(1-\alpha/2)\sqrt{S_{\mathcal{M}}^2}\mid\mathcal{Z} \Big) \\  
    &=\lim\limits_{I \rightarrow \infty}\text{pr}\bigg(-\Phi^{-1}(1-\alpha/2) \leq \frac{\widehat{\lambda}_{\diamond}-\lambda}{\sqrt{S^2_{\mathcal{M}}}} \leq \Phi^{-1}(1-\alpha/2)\mid\mathcal{Z}  \bigg) \\ 
    &=\lim\limits_{I \rightarrow \infty}\text{pr}\bigg(-\Phi^{-1}(1-\alpha/2) \leq \frac{\widehat{\lambda}_{\diamond}-\lambda}{\sqrt{\text{var}(\widehat{\lambda}_{\diamond}\mid\mathcal{Z})}}\sqrt{\frac{\text{var}(\widehat{\lambda}_{\diamond}\mid\mathcal{Z})}{S^2_{\mathcal{M}}}} \leq \Phi^{-1}(1-\alpha/2)\mid\mathcal{Z}  \bigg) \\ 
    &\geq \lim\limits_{I \rightarrow \infty}\text{pr}\Bigg(-\Phi^{-1}(1-\alpha/2) \leq \frac{\widehat{\lambda}_{\diamond}-\lambda}{\sqrt{\text{var}(\widehat{\lambda}_{\diamond}\mid\mathcal{Z})}} \leq \Phi^{-1}(1-\alpha/2)\mid\mathcal{Z}  \Bigg) \\ 
    &=\Phi(\Phi^{-1}(1-\alpha/2))-\Phi(\Phi^{-1}(\alpha/2)) \\ 
    &=1-\alpha. 
\end{align*}
\end{proof}

\subsection*{B.8: Statement and Proof of Proposition S3}
\begin{proposition}\label{prop: unbiasdness_AIPPW} Under Assumptions 1 and 2 in the main text, we have $E(\widehat{\lambda}_{\dagger}\mid \mathcal{Z})=\lambda$.
\end{proposition}
\begin{proof}
For each matched set $i$, we have 
\begin{align*}
    &E(\widehat{\lambda}_{\dagger, i}\mid \mathcal{Z})\\
    =&E\Bigg\{\frac{1}{n_i}\sum_{j=1}^{n_i}\bigg\{\frac{Z_{ij}}{p_{ij}}
    \Big(Y_{ij}-\widehat{g}_1(\mathbf{x}_{ij})\Big)- 
    \frac{(1-Z_{ij})}{1-p_{ij}}
    \Big(Y_{ij}-\widehat{g}_0(\mathbf{x}_{ij})\Big)+\widehat{g}_1(\mathbf{x}_{ij})-\widehat{g}_0(\mathbf{x}_{ij})\bigg\}\mid \mathcal{Z}\Bigg\}\\
    =&\frac{1}{n_{i}}\Bigg\{\sum_{j=1}^{n_{i}}\frac{E(Z_{ij}\mid\mathcal{Z})}{p_{ij}}\Big(Y_{ij}(1)-\widehat{g}_1(\mathbf{x}_{ij})\Big)-\sum_{j=1}^{n_{i}}\frac{(1-E(Z_{ij}\mid \mathcal{Z})) }{1-p_{ij}}\Big(Y_{ij}(0)-\widehat{g}_0(\mathbf{x}_{ij})\Big)+\widehat{g}_1(\mathbf{x}_{ij})-\widehat{g}_0(\mathbf{x}_{ij})\Bigg\}\\
    =&\frac{1}{n_{i}}\sum_{j=1}^{n_{i}}\{Y_{ij}(1)-\widehat{g}_1(\mathbf{x}_{ij})-Y_{ij}(0)+\widehat{g}_0(\mathbf{x}_{ij})+\widehat{g}_1(\mathbf{x}_{ij})-\widehat{g}_0(\mathbf{x}_{ij})\}\\
    =&\frac{1}{n_{i}}\sum_{j=1}^{n_{i}}\{Y_{ij}(1)-Y_{ij}(0)\}.
\end{align*}
Therefore, we have 
\begin{equation*}       
E(\widehat{\lambda}_{\dagger}\mid\mathcal{Z})=E\Big(\sum_{i=1}^{I}\frac{n_{i}}{N} \widehat{\lambda}_{\dagger, i}\mid \mathcal{Z}\Big)=\sum_{i=1}^{I}\frac{n_{i}}{N}E(\widehat{\lambda}_{\dagger, i}\mid \mathcal{Z})=\frac{1}{N}\sum_{i=1}^{I}\sum_{j=1}^{n_{i}}\{Y_{ij}(1)-Y_{ij}(0)\}.
\end{equation*}
\end{proof}

\vspace{-0.6cm}
\subsection*{B.9: Statement and Proof of Proposition S4}
To prove Proposition S4, we need the following regularity condition:
\begin{condition}[Bounded Limiting Outcome Models]\label{condition: bounded outcomes for outcome models} As $I\rightarrow \infty$, we have $\widehat{g}_{1}(\mathbf{x}_{ij})\xrightarrow{a.s.} \widetilde{g}_1(\mathbf{x}_{ij})$ and $\widehat{g}_{0}(\mathbf{x}_{ij})\xrightarrow{a.s.} \widetilde{g}_0(\mathbf{x}_{ij})$, where $\widetilde{g}_1(\mathbf{x}_{ij})$ and $\widetilde{g}_0(\mathbf{x}_{ij})$ are bounded functions.
\end{condition}
\begin{proposition}\label{prop: oracle AIPPW}
Under independence of treatment assignments across matched sets, Assumptions 1 and 2 in the main text, Conditions~\ref{condition: bounded outcomes} and~\ref{condition: bounded outcomes for outcome models}, if $\widehat{g}_{1}(\mathbf{x}_{ij})\xrightarrow{a.s.} \widetilde{g}_1(\mathbf{x}_{ij})$ and $\widehat{g}_{0}(\mathbf{x}_{ij})\xrightarrow{a.s.} \widetilde{g}_0(\mathbf{x}_{ij})$, then the oracle AIPPW estimator $\widehat{\lambda}_{\dagger}$ is strongly consistent for $\lambda$ as $I\to\infty$, that is, $\widehat{\lambda}_{\dagger}\xrightarrow{a.s.}\lambda$. 
\end{proposition}

\begin{proof}
    By Conditions~\ref{condition: bounded outcomes} and~\ref{condition: bounded outcomes for outcome models}, there exists a constant $M^{*}<\infty$ such that $|Y_{ij}-\widehat{g}_1(\mathbf{x}_{ij})|\leq M^{*}$ and $|Y_{ij}-\widehat{g}_0(\mathbf{x}_{ij})|\leq M^{*}$ almost surely for all $I$ and for all $i=1,\dots,I,j=1,\dots,n_i$. Let $K_i=\sum_{j=1}^{n_i}G_{ij}$, where $G_{ij}=\frac{Z_{ij}}{p_{ij}}
    \Big(Y_{ij}-\widehat{g}_1(\mathbf{x}_{ij})\Big)-\frac{(1-Z_{ij})}{1-p_{ij}}\Big(Y_{ij}-\widehat{g}_0(\mathbf{x}_{ij})\Big)+\widehat{g}_1(\mathbf{x}_{ij})-\widehat{g}_0(\mathbf{x}_{ij})-\{Y_{ij}(1)-Y_{ij}(0)\}$. Then, for any $i,j$, we have 
    \vspace{-0.3cm}
    \begin{equation*}
    \begin{split}
        &\Bigg|\frac{Z_{ij}}{p_{ij}}\Big(Y_{ij}-\widehat{g}_1(\mathbf{x}_{ij})\Big)-\frac{(1-Z_{ij})}{1-p_{ij}}\Big(Y_{ij}-\widehat{g}_0(\mathbf{x}_{ij})\Big)\Bigg|\leq\frac{\Big|\Big(Y_{ij}-\widehat{g}_1(\mathbf{x}_{ij})\Big)\Big|}{p_{ij}}+\frac{\Big|\Big(Y_{ij}-\widehat{g}_0(\mathbf{x}_{ij})\Big)\Big|}{1-p_{ij}}\leq \frac{2M^*}{q_1}
    \end{split}
    \end{equation*}
    and
    \vspace{-0.3cm}
    \begin{equation*}
        \big|\widehat{g}_1(\mathbf{x}_{ij})-\widehat{g}_0(\mathbf{x}_{ij})-\{Y_{ij}(1)-Y_{ij}(0)\}\big|\leq2M^*, 
    \end{equation*}
    which implies that
    \begin{equation*}
        |G_{ij}|\leq\frac{2M^*}{q_1}+2M^*=L^*.
    \end{equation*}
    Therefore, we have $\text{var}(G_{ij})\leq E(G_{ij}^2)\leq L{^*}^2 < \infty$. By the Cauchy-Schwarz inequality, we have $H_i^2=\Big(\sum_{j=1}^{n_i}G_{ij}\Big)^2\leq n_i\sum_{j=1}^{n_i}G_{ij}^2$. Therefore, we have $\text{var}(H_i)\leq E(H_i^2)\leq n_i\sum_{j=1}^{n_i}E(G_{ij}^2)\leq n_i^2L{^*}^2\leq C_{1}^{2}L{^*}^{2}$, which implies that $\lim_{i\rightarrow\infty}\sum_{i=1}^{I}i^{-2}\text{var}(H_i)<\infty$.
    
    For each matched set $i$, as shown in Appendix B.8, we have $E(K_i|\mathcal{Z})=0$. Since $K_i$ are independent random variables, by Lemma~\ref{lemma: kolmogorov}, we have $\frac{1}{I}\sum_{i=1}^{I}K_i\xrightarrow{a.s.}0.$
    Note that $1\leq N/I \leq C_1$, so $I/N$ is a constant strictly between 0 and 1. Therefore, we have
    \vspace{-0.3cm}
    \begin{equation*}
        \widehat{\lambda}_{\dagger}-\lambda=\frac{1}{N}\sum_{i=1}^{I}K_i=\frac{I}{N}\cdot\frac{1}{I}\sum_{i=1}^{I}K_i\xrightarrow{a.s.}0.
    \end{equation*}
    \vspace{-0.3cm}
\end{proof}

\subsection*{B.10: Proof of Proposition 3}

\begin{proof}
Define
\begin{equation*}
    b_{ij}=\bigg|\frac{Z_{ij}(\widehat{p}_{ij}-p_{ij})}{p_{ij}\widehat{p}_{ij}}+\frac{(1-Z_{ij})(\widehat{p}_{ij}-p_{ij})}{p_{ij}\widehat{p}_{ij}}\bigg|,\quad c_{ij}=\bigg|Z_{ij}\Big(Y_{ij}-\widehat{g}_1(\mathbf{x}_{ij})\Big)+(1-Z_{ij})\Big(Y_{ij}-\widehat{g}_0(\mathbf{x}_{ij})\Big)\bigg|.
\end{equation*}
    Under Conditions~\ref{condition: infimum} and \ref{condition: bounded outcomes} stated in Section~B.4, for any $i,j$, we have
\begin{equation*}
        \begin{split}        
            &\bigg|\Big\{\frac{Z_{ij}}{p_{ij}}\Big(Y_{ij}-\widehat{g}_1(\mathbf{x}_{ij})\Big)-\frac{(1-Z_{ij})}{1-p_{ij}}\Big(Y_{ij}-\widehat{g}_0(\mathbf{x}_{ij})\Big)\Big\}-\Big\{\frac{Z_{ij}}{\widehat{p}_{ij}}\Big(Y_{ij}-\widehat{g}_1(\mathbf{x}_{ij})\Big)-\frac{(1-Z_{ij})}{1-\widehat{p}_{ij}}\Big(Y_{ij}-\widehat{g}_0(\mathbf{x}_{ij})\Big)\Big\}\bigg| \\  
            =&\bigg|Z_{ij}\bigg(\frac{1}{p_{ij}}-\frac{1}{\widehat{p}_{ij}}\bigg)\Big(Y_{ij}-\widehat{g}_1(\mathbf{x}_{ij})\Big)-\big(1-Z_{ij}\big)\bigg(\frac{1}{p_{ij}}-\frac{1}{\widehat{p}_{ij}}\bigg)\Big(Y_{ij}-\widehat{g}_0(\mathbf{x}_{ij})\Big)\bigg| \\
            =&\bigg|\frac{Z_{ij}(\widehat{p}_{ij}-p_{ij})}{p_{ij}\widehat{p}_{ij}}\Big(Y_{ij}-\widehat{g}_1(\mathbf{x}_{ij})\Big)-\frac{(1-Z_{ij})(\widehat{p}_{ij}-p_{ij})}{p_{ij}\widehat{p}_{ij}}\Big(Y_{ij}-\widehat{g}_0(\mathbf{x}_{ij})\Big)\bigg|=b_{ij}c_{ij}.
        \end{split}
    \end{equation*}
Then we have 
\begin{equation*}
\begin{split}
    |\widehat{\lambda}_{\ddagger}-\widehat{\lambda}_{\dagger}|&\leq\frac{1}{N}\sum_{i=1}^{I}\sum_{j=1}^{n_i}b_{ij}c_{ij} \\
    &\leq\frac{1}{N}\bigg(\sum_{i=1}^{I}\sum_{j=1}^{n_i}b_{ij}^2\bigg)^{1/2}\bigg(\sum_{i=1}^{I}\sum_{j=1}^{n_i}c_{ij}^2\bigg)^{1/2}
    \quad \text{(Cauchy-Schwarz inequality)} \\
    &=\frac{1}{N}\Bigg[\sum_{i=1}^{I}\sum_{j=1}^{n_i}\bigg\{\frac{Z_{ij}(\widehat{p}_{ij}-p_{ij})^2}{p_{ij}^2\widehat{p}_{ij}^2}+\frac{(1-Z_{ij})(\widehat{p}_{ij}-p_{ij})^2}{p_{ij}^2\widehat{p}_{ij}^2}\bigg\}\Bigg]^{1/2} \\
    &\quad \quad\times\Bigg[\sum_{i=1}^{I}\sum_{j=1}^{n_i}\bigg\{Z_{ij}\Big(Y_{ij}-\widehat{g}_1(\mathbf{x}_{ij})\Big)^2+(1-Z_{ij})\Big(Y_{ij}-\widehat{g}_0(\mathbf{x}_{ij})\Big)^2\bigg\}\Bigg]^{1/2}. \\
    &=\frac{1}{Nq_{1}q_{2}}\|\widehat{\mathbf{p}}-\mathbf{p}\|_2 \times\Bigg[\sum_{i=1}^{I}\sum_{j=1}^{n_i}\bigg\{Z_{ij}\Big(Y_{ij}-\widehat{g}_1(\mathbf{x}_{ij})\Big)^2+(1-Z_{ij})\Big(Y_{ij}-\widehat{g}_0(\mathbf{x}_{ij})\Big)^2\bigg\}\Bigg]^{1/2}.
\end{split}
\end{equation*}
If either the propensity score model is correctly specified (i.e. $\widehat{p}_{ij}\xrightarrow{a.s.}p_{ij}$) or the outcome models are correctly specified (i.e., $\widehat{g}_1(\mathbf{x}_{ij})\xrightarrow{a.s.} Y_{ij}(1)$ and $\widehat{g}_0(\mathbf{x}_{ij})\xrightarrow{a.s.} Y_{ij}(0)$), then $\widehat{\lambda}_{\ddagger}-\widehat{\lambda}_{\dagger}\xrightarrow{a.s.}0$. Because $\widehat{\lambda}_{\ddagger}-\widehat{\lambda}_{\dagger}\xrightarrow{a.s.}0$ and $\widehat{\lambda}_{\dagger}-\lambda\xrightarrow{a.s.}0$, Proposition 3 follows from the transitivity of almost surely convergence.
\end{proof}

\subsection*{B.11: Proof of Proposition S1}

Let $\overline{Y}=\sum_{i=1}^{I}\sum_{j=1}^{n_{i}}Y_{ij}(1)-Y_{ij}(0)$ and $\overline{D}=\sum_{i=1}^{I}\sum_{j=1}^{n_{i}}D_{ij}(1)-D_{ij}(0)$. Next, we define $T=\sum_{i=1}^{I}T_{i}$, where $T_{i}=\sum_{j=1}^{n_{i}}\frac{1}{p_{ij}(1-p_{ij})}Y_{ij}(Z_{ij}-p_{ij})$. Then, we denote $S=\sum_{i=1}^{I}S_{i}$, where $S_{i}=\sum_{j=1}^{n_{i}}\frac{1}{p_{ij}(1-p_{ij})}D_{ij}(Z_{ij}-p_{ij})$. That is, the bias-corrected Wald estimator $\widehat{\theta}_{*}=T/S$. 

To prove Proposition S1, we need the following conditions. 

\begin{condition}\label{condition: convergence of estimand IPPW} 
(Convergence of Finite-Population Means) As $I\to\infty$, the $I^{-1}\overline{D}$ converges to some positive finite value, and the $I^{-1}\overline{Y}$ converges to some finite value.
\end{condition}

\begin{condition}\label{condition: bounded IPPW}
  (No Extreme Second Moments or Variances) For any fixed $I$ and for all $i=1,\dots,I$, we have $E(T_{i}^{2}\mid \mathcal{Z})<\infty$ and $E(S_{i}^{2}\mid \mathcal{Z})<\infty$. In addition, we have $\lim_{I\rightarrow\infty}\sum_{i=1}^{I}i^{-2}\text{var}(T_{i}\mid \mathcal{Z})<\infty$ and $\lim_{I\rightarrow\infty} \sum_{i=1}^{I}i^{-2}\text{var}(S_{i}\mid \mathcal{Z})<\infty$. 
\end{condition}

Now, we are ready to prove Proposition S1.
\begin{proof}
    Note that

\vspace{-1.2cm}
\begin{align*}
    E(T\mid\mathcal{Z})&=E\bigg\{\sum_{i=1}^{I}\sum_{j=1}^{n_{i}}\frac{1}{p_{ij}(1-p_{ij})}Y_{ij}(Z_{ij}-p_{ij})\mid\mathcal{Z}\bigg\} \\ 
    &=\sum_{i=1}^{I}\sum_{j=1}^{n_{i}}\frac{1}{p_{ij}(1-p_{ij})}E\big\{Y_{ij}(Z_{ij}-p_{ij})\mid\mathcal{Z}\big\} \\ 
    &=\sum_{i=1}^{I}\sum_{j=1}^{n_{i}}\frac{1}{p_{ij}(1-p_{ij})}\big\{Y_{ij}(1)(1-p_{ij})\times\text{pr}(Z_{ij}=1\mid\mathcal{Z})-Y_{ij}(0)p_{ij}\times\text{pr}(Z_{ij}=0\mid\mathcal{Z})\big\} \\ 
    &=\sum_{i=1}^{I}\sum_{j=1}^{n_{i}}Y_{ij}(1)-Y_{ij}(0).
\end{align*}
\vspace{-1.2cm}

\noindent Hence, we have $E(T\mid\mathcal{Z})=E\big(\sum_{i=1}^{I}T_{i}\mid\mathcal{Z}\big)=\overline{Y}$.

Since $T_{i}$ are independent but not identically distributed random variables, by Condition~\ref{condition: bounded IPPW} and Lemma~\ref{lemma: kolmogorov}, we have $I^{-1}\sum_{i=1}^{I}T_{i}-I^{-1}E\big(\sum_{i=1}^{I}T_{i}\mid \mathcal{Z}\big)\xrightarrow{a.s.}0$ as $I \to \infty$. Since $E\big(\sum_{i=1}^{I}T_{i}\mid\mathcal{Z}\big)=\overline{Y}$, we have $I^{-1}T-I^{-1}\overline{Y}\xrightarrow{a.s.}0$ as $I \to \infty$. 
Similarly, we can obtain that $I^{-1}S-I^{-1}\overline{D}\xrightarrow{a.s.}0$ as $I\to\infty$. Note that

\vspace{-0.4cm}
\begin{equation*}
    \frac{T}{S}-\frac{\overline{Y}}{\overline{D}}=\frac{T\overline{D}-S\overline{Y}}{S\overline{D}} =\frac{(I^{-1}\overline{D})I^{-1}(T-\overline{Y})+(I^{-1}\overline{Y})I^{-1}(\overline{D}-S)}{(I^{-1}S)(I^{-1}\overline{D})}.
\end{equation*}
\vspace{-0.8cm}

\noindent Under Condition~\ref{condition: convergence of estimand IPPW}, since $I^{-1}T-I^{-1}\overline{Y}\xrightarrow{a.s.}0$ and $I^{-1}S-I^{-1}\overline{D}\xrightarrow{a.s.}0$, we have $I^{-1}T \xrightarrow{a.s.}\lim_{I\rightarrow \infty} I^{-1}\overline{Y}<\infty$ and $I^{-1}S \xrightarrow{a.s.}\lim_{I\rightarrow \infty} I^{-1}\overline{D}\in (0,\infty)$. Therefore, as $I\rightarrow\infty$, we have 

\vspace{-0.8cm}
\begin{equation*}
  \widehat{\theta}_{*}-\theta = \frac{T}{S}-\frac{\overline{Y}}{\overline{D}}=\frac{(I^{-1}\overline{D})I^{-1}(T-\overline{Y})+(I^{-1}\overline{Y})I^{-1}(\overline{D}-S)}{(I^{-1}S)(I^{-1}\overline{D})}\xrightarrow{a.s.}0.
\end{equation*}
\vspace{-0.8cm}

\noindent That is, we have proven Proposition S1.
\end{proof}

\subsection*{B.12: Statement and Proof of Proposition S4}
\begin{proposition}\label{prop: variance2}
    Assuming the IV ignorability assumption and independence of treatment assignments across matched sets, we have $E\{V_{*}^2(\theta_{0})\mid \mathcal{Z}\}\geq \text{var}\{A_{*}(\theta_{0})\mid \mathcal{Z}\}$ under $H_{0}:\theta=\theta_{0}$.
\end{proposition}

\begin{proof}
    Let $\kappa_{i,\theta}=E\{A_{*,i}(\theta)\mid\mathcal{Z}\}$, $\kappa_{\theta}=E\{A_{*}(\theta)\mid\mathcal{Z}\}$, and $\omega_{i,\theta}=\text{var}\{A_{*,i}(\theta)\mid\mathcal{Z}\}$. Then, we have

\vspace{-1.5cm}
\begin{align*}
    E\{V_{*}^2(\theta)\mid\mathcal{Z}\} &=\frac{1}{I(I-1)}\sum_{i=1}^{I}E\big[\{A_{*,i}(\theta)-A_{*}(\theta)\}^2\mid\mathcal{Z}\big] \\ 
    &=\frac{1}{I(I-1)}\sum_{i=1}^{I}\Big[E\{A_{*,i}^2(\theta)\mid\mathcal{Z}\}+E\{A_{*}^2(\theta)\mid\mathcal{Z}\}-2E\{A_{*,i}(\theta)A_{*}(\theta)\mid\mathcal{Z}\}\Big] \\ 
    &=\frac{1}{I(I-1)}\sum_{i=1}^{I}\Bigg\{(\kappa_{i,\theta}^2+\omega_{i,\theta})+\Bigg(\kappa_{\theta}^2+\frac{1}{I^2}\sum_{j=1}^{I}\omega_{j,\theta}\Bigg)-\frac{2}{I}\Bigg(\kappa_{i,\theta}^2+\omega_{i,\theta}+\sum_{j\neq i}\kappa_{i,\theta}\kappa_{j,\theta}\Bigg)\Bigg\} \\ 
    &=\frac{1}{I(I-1)}\sum_{i=1}^{I}\Bigg(\omega_{i,\theta}-\frac{2}{I}\omega_{i,\theta}+\frac{1}{I^2}\sum_{j=1}^{I}\omega_{j,\theta}\Bigg)+\frac{1}{I(I-1)}\sum_{i=1}^{I}\Bigg(\kappa_{i,\theta}^2+\kappa_{\theta}^2-\frac{2}{I}\sum_{j=1}^{I}\kappa_{i,\theta}\kappa_{j,\theta}\Bigg) \\ 
    &=\bigg(\frac{I^2-2I+I}{I(I-1)}\bigg)\frac{1}{I^2}\sum_{i=1}^{I}\omega_{i,\theta}+\frac{1}{I(I-1)}\sum_{i=1}^{I}(\kappa_{i,\theta}-\kappa_{\theta})^2 \\ 
    &=\frac{1}{I^2}\sum_{i=1}^{I}\omega_{i,\theta}+\frac{1}{I(I-1)}\sum_{i=1}^{I}(\kappa_{i,\theta}-\kappa_{\theta})^2.
\end{align*}
\vspace{-1cm}

\noindent Since matched sets are independent and $A_{*}(\theta)=I^{-1}\sum_{i=1}^{I}A_{*,i}(\theta)$, we have $\text{var}\{A_{*}(\theta)\mid\mathcal{Z}\}=I^{-2}\sum_{i=1}^{I}\text{var}\{A_{*,i}(\theta)\mid\mathcal{Z}\}=I^{-2}\sum_{i=1}^{I}\omega_{i,\theta}$. Therefore, we have 

\vspace{-0.7cm}
\begin{equation*}
E\{V_{*}^2(\theta)\mid\mathcal{Z}\}-\text{var}\{A_{*}(\theta)\mid\mathcal{Z}\}=\frac{1}{I(I-1)}\sum_{i=1}^{I}(\kappa_{i,\theta}-\kappa_{\theta})^2\geq 0.    
\end{equation*}
\vspace{-1cm}

\end{proof}

\vspace{-0.5cm}
\subsection*{B.13: Proof of Theorem S1}

To prove Theorem S1, we need the following regularity conditions.
\begin{condition}[Non-Degenerate Test Statistics]\label{condition: non-degenerate}
    For any $I$, we have $\text{var}\{A_{*}(\theta)\mid \mathcal{Z}\}>0$. Also, we have $\lim\inf_{I\rightarrow\infty}\text{var}\{A_{*}(\theta)\mid \mathcal{Z}\}>0$.
\end{condition}

\begin{condition}[No Extreme Third Moments]\label{condition: lim}
    We have

    \vspace{-0.7cm}
    \begin{equation*}
        \limsup\limits_{I\to\infty}\frac{\sum_{i=1}^{I}E\{|A_{*,i}(\theta)-\kappa_{i,\theta}|^{3}\mid\mathcal{Z}\}}{[\sum_{i=1}^{I}\text{var}\{A_{*,i}(\theta)\mid\mathcal{Z}\}]^{3/2}}=0.
    \end{equation*}
    \vspace{-1.1cm}
    
\end{condition}
\begin{condition}[No Extreme Fourth Moments]\label{condition: slow}
 As $I\rightarrow \infty$, we have $I^{-1}\sum_{i=1}^{I}E\{A_{*,i}^{4}(\theta)\mid\mathcal{Z}\}=o(I)$.
\end{condition}
Conditions~\ref{condition: lim} and \ref{condition: slow} are commonly used moment conditions in matched IV studies (\citealp{baiocchi2010building, kang2016full}). Also, recall the following form of the Lyapunov central limit theorem:
\begin{lemma}\label{lem: Lyapunov}
(Lyapunov Central Limit Theorem): Suppose we have a sequence of independent random variables $X_i, 1\leq i \leq n$, each unit $X_{i}$ has finite expected value $\mu_i$ and finite variance $\sigma_{i}^2$. Let $S_{n}^2=\sum_{i=1}^{n}\sigma_{i}^2$. If the sequence of $X_{i}$ satisfies

\vspace{-0.7cm}
\begin{equation*}
    \lim_{n\to\infty}\frac{1}{S_{n}^{2+\delta}}\sum_{i=1}^{n}E\{|X_i-\mu_i|^{2+\delta}\}=0 \quad \text{for some $\delta>0$,}
\end{equation*}
\vspace{-1.2cm}

\noindent we have 

\vspace{-0.7cm}
\begin{equation*}
    \text{$\frac{1}{S_n}\sum_{i=1}^{n}(X_i-\mu_i)\xrightarrow{d}N(0,1)$ as $n\to\infty$.}
\end{equation*}
\vspace{-1.2cm}

\end{lemma}

Then, we are ready to prove Theorem 2, of which the idea is to generalize the arguments in \citet{baiocchi2010building} and \citet{kang2016full} from the exact matching to the potentially inexact matching case.
\begin{proof}
Let $h_{i,\theta}=E\{A_{*,i}^2(\theta)\mid\mathcal{Z}\}$, and $\omega_{\theta}=\text{var}\{A_{*}(\theta)\mid\mathcal{Z}\}$. We have $E\{I^{-1}\sum_{i=1}^{I}A_{*,i}^2(\theta)\mid\mathcal{Z}\}=I^{-1}\sum_{i=1}^{I}h_{i,\theta}$. Note that $\text{var}\{I^{-1}\sum_{i=1}^{I}A_{*,i}^2(\theta)\mid\mathcal{Z}\}\leq I^{-2}\sum_{i=1}^{I}E\{A_{*,i}^4(\theta)\mid\mathcal{Z}\}$. Under Condition~\ref{condition: slow}, we have $\text{var}\{I^{-1}\sum_{i=1}^{I}A_{*,i}^2(\theta)\mid\mathcal{Z}\}\rightarrow 0$ as $I\rightarrow \infty$. Invoking Chebyshev's inequality, we have

\vspace{-0.7cm}
\begin{equation*}
I^{-1}\sum_{i=1}^{I}A_{*,i}^2(\theta)-I^{-1}\sum_{i=1}^{I}h_{i,\theta} \xrightarrow{p} 0 \quad \text{as $I\to\infty$}.    
\end{equation*}
\vspace{-1cm}

\noindent Next, by Jensen's inequality and Condition~\ref{condition: slow}, we have $I^{-1}\sum_{i=1}^{I}E\{A_{*,i}^{2}(\theta)\mid\mathcal{Z}\}=o(I)$. Invoking Chebyshev's inequality, we have

\vspace{-0.7cm}
\begin{equation*}
    A_{*}(\theta)-\kappa_{\theta}\xrightarrow{p} 0 \quad \text{as $I\to\infty$}. 
\end{equation*}
Since $\kappa_{\theta}=0$ for all $I$, by the continuous mapping theorem, we have 
\vspace{-0.7cm}
\begin{equation*}
    A_{*}^2(\theta)\xrightarrow{p} 0 \quad \text{as $I\to\infty$}.
\end{equation*}
\vspace{-1.5cm}

 \noindent Combining all these convergence results, we can get that for any $\epsilon>0$ and any $\delta>0$, there exists $I^*$ such that for all $I\geq I^*$, we have

\vspace{-1.2cm}
\begin{align*}
 \text{pr}\Bigg\{I^{-1}\sum_{i=1}^{I}A_{*,i}^2(\theta)-I^{-1}\sum_{i=1}^{I}h_{i,\theta}\leq-\frac{\epsilon}{2}\mid \mathcal{Z}\Bigg\}\leq\frac{\delta}{2}
\end{align*}
\vspace{-1.2cm}

\noindent and

\vspace{-1.7cm}
\begin{align*}
 \text{pr}\Big\{A_{*}^2(\theta)\geq \frac{\epsilon}{2}\mid \mathcal{Z}\Big\}\leq\frac{\delta}{2}.
\end{align*}
\vspace{-1.2cm}

\noindent Also, note that for any $I$, we have

\vspace{-1.2cm}
\begin{equation*}
    \frac{1}{I-1}\sum_{i=1}^{I}h_{i,\theta}-I\omega_{\theta}=\frac{1}{I-1}\sum_{i=1}^{I}E\{A_{*,i}^2(\theta)\mid\mathcal{Z}\}-\frac{1}{I}\sum_{i=1}^{I}\text{var}\{A_{*,i}(\theta)\mid\mathcal{Z}\}\geq 0.
\end{equation*}
\vspace{-1cm}

\noindent Then, we have

\vspace{-1.3cm}
\begin{align}\label{eqn: conservative variance estimator}
    &\quad \ \text{pr}\big\{IV_{*}^2(\theta)-I\omega_{\theta}\leq-\epsilon\mid \mathcal{Z} \big\}\nonumber \\ &=\text{pr}\Bigg[\frac{I}{I-1}\Bigg\{I^{-1}\sum_{i=1}^{I}A_{*,i}^2(\theta)-A_{*}^2(\theta)\Bigg\}-I\omega_{\theta}\leq-\epsilon\mid \mathcal{Z} \Bigg]\nonumber \\ 
    &=\text{pr}\Bigg[\frac{I}{I-1}\Bigg\{I^{-1}\sum_{i=1}^{I}A_{*,i}^2(\theta)-I^{-1}\sum_{i=1}^{I}h_{i,\theta}+I^{-1}\sum_{i=1}^{I}h_{i,\theta}-A_{*}^2(\theta)\Bigg\}-I\omega_{\theta}\leq-\epsilon\mid \mathcal{Z}\Bigg]\nonumber \\ 
    &=\text{pr}\Bigg[\frac{I}{I-1}\Bigg\{I^{-1}\sum_{i=1}^{I}A_{*,i}^2(\theta)-I^{-1}\sum_{i=1}^{I}h_{i,\theta}-A_{*}^2(\theta)\Bigg\}+\frac{1}{I-1}\sum_{i=1}^{I}h_{i,\theta}-I\omega_{\theta}\leq-\epsilon \mid \mathcal{Z}\Bigg]\nonumber \\ 
    &\leq\text{pr}\Bigg[\frac{I}{I-1}\Bigg\{I^{-1}\sum_{i=1}^{I}A_{*,i}^2(\theta)-I^{-1}\sum_{i=1}^{I}h_{i,\theta}-A_{*}^2(\theta)\Bigg\}\leq-\epsilon \mid \mathcal{Z}\Bigg]  \nonumber\\
     &\leq\text{pr}\Bigg[\frac{I}{I-1}\Bigg\{I^{-1}\sum_{i=1}^{I}A_{*,i}^2(\theta)-I^{-1}\sum_{i=1}^{I}h_{i,\theta}\Bigg\}\leq-\frac{\epsilon}{2} \mid \mathcal{Z}\Bigg]  + \text{pr}\bigg\{-\frac{I}{I-1}A_{*}^2(\theta)\leq-\frac{\epsilon}{2} \mid \mathcal{Z}\Big\} \nonumber\\
    &\leq\frac{\delta}{2}+\frac{\delta}{2}=\delta.
\end{align}

Since $I^{-1}\sum_{i=1}^{I}\kappa_{i,\theta}=\kappa_{\theta}=0$, we can rewrite the test statistic as $A_{*}(\theta)=I^{-1}\sum_{i=1}^{I}A_{*,i}(\theta)=I^{-1}\sum_{i=1}^{I}[A_{*,i}(\theta)-\kappa_{i,\theta}]$. Therefore, under Condition~\ref{condition: non-degenerate}, we can write

\vspace{-1cm}
\begin{align}\label{eqn: two terms}
    \frac{A_{*}(\theta)}{V_{*}(\theta)}=\Bigg[\frac{I^{-1}\sum_{i=1}^{I}\{A_{*,i}(\theta)-\kappa_{i,\theta}\}}{\sqrt{I^{-2}\sum_{i=1}^{I}\omega_{i,\theta}}}\Bigg]\Bigg\{\frac{\sqrt{I^{-2}\sum_{i=1}^{I}\omega_{i,\theta}}}{\sqrt{V_{*}^2(\theta)}}\Bigg\}.
\end{align}
\vspace{-1cm}

\noindent Under Condition~\ref{condition: bounded IPPW} (which implies that each $A_{*,i}(\theta)-\kappa_{i,\theta}$ has finite first and second moments) and Condition~\ref{condition: lim}, Lyapunov's conditions are satisfied for the central limit theorem stated in Lemma~\ref{lem: Lyapunov}, so the first term on the right-hand side of (\ref{eqn: two terms}) converges in distribution to $N(0,1)$. From (\ref{eqn: conservative variance estimator}) (i.e., the second term in the right-hand side of (\ref{eqn: two terms}) will be less than or equal to 1 with arbitrarily high probability as $I\rightarrow \infty$), by Slutsky's theorem, we can obtain that for any $t>0$, we have

\vspace{-1.2cm}
\begin{align*}
    \limsup\limits_{I\to\infty} \text{pr}\bigg\{\frac{A_{*}(\theta)}{V_{*}(\theta)}\leq -t \mid\mathcal{Z}\bigg\}\leq\Phi(-t), \quad \limsup\limits_{I\to\infty} \text{pr}\bigg\{\frac{A_{*}(\theta)}{V_{*}(\theta)}\geq t \mid\mathcal{Z}\bigg\}\leq\Phi(-t).
\end{align*}
\vspace{-1.2cm}

\noindent Therefore, we have 

\vspace{-1.2cm}
\begin{align*}
   &\quad \ \liminf_{I\rightarrow \infty}\text{pr}\big(\theta\in CS^{\theta}_{*}\mid \mathcal{Z} \big)\\
   &= \liminf_{I\rightarrow \infty}\text{pr}\bigg\{\Big|\frac{A_{*}(\theta)}{V_{*}(\theta)}\Big|\leq \Phi^{-1}(1-\alpha/2) \mid\mathcal{Z}\bigg\}\\
    &=1-\limsup\limits_{I\to\infty} \text{pr}\bigg\{\frac{A_{*}(\theta)}{V_{*}(\theta)}\leq -\Phi^{-1}(1-\alpha/2) \mid\mathcal{Z}\bigg\}-\limsup\limits_{I\to\infty} \text{pr}\bigg\{\frac{A_{*}(\theta)}{V_{*}(\theta)}\geq \Phi^{-1}(1-\alpha/2) \mid\mathcal{Z}\bigg\}\\
    &\geq 1-\alpha.
\end{align*}
\vspace{-1.2cm}

\noindent So the desired result in Theorem 2 follows. 
\end{proof}

\section*{Appendix C:  Simulation Studies}\label{sec: Appendix C}

\subsection*{C.1: Comparative Performance of the Proposed IPPW Method and Other Methods}

We conduct a simulation study to assess the bias and coverage rate of the proposed IPPW method compared with the conventional pre-matching and post-matching randomization-based inference method for the sample average treatment effect (as reviewed in Section 2 in the main text), which ignores inexact matching and parsimoniously assumes the randomization assumption (\citealp{rosenbaum2002observational, fogarty2018mitigating}). We set the sample size $N=400$. The five covariates $\mathbf{x}_{n}=(x_{n1},\cdots,x_{n5})$ for each pre-matching unit $n$ are generated using the following process: $(x_{n1},x_{n2},x_{n3})\overset{\text{i.i.d.}}{\sim} \mathcal{N}((0,0,0),\mathbf{I}_{3 \times 3})$, $x_{n4}\overset{\text{i.i.d.}}{\sim}\text{Laplace}(0, \sqrt{2}/2)$, and $x_{n5}\overset{\text{i.i.d.}}{\sim}\text{Laplace}(0, \sqrt{2}/2)$. Let $g(\mathbf{x}_{n})=0.72x_{n1}+0.88x_{n2}+0.93x_{n3}+0.65x_{n4}+0.78x_{n5}-0.8$ and $f(\mathbf{x}_{n})=0.1x_{n1}^3+0.3x_{n2}+0.2\log(x_{n3}^2)+0.1x_{n4}+0.2x_{n5}+|x_{n1}x_{n2}|+(x_{n3}x_{n4})^2+0.5(x_{n2}x_{n4})^2-2.5$. Then, we consider the following three common models for generating the treatment indicator $Z_{n}$ for each unit $n$:
\begin{itemize}
    \item Model 1 (Linear Logistic Model): $\text{logit}\ \text{pr}(Z_{n}=1\mid \mathbf{x}_{n})=g(\mathbf{x}_{n})$.
    \item Model 2 (Nonlinear Logistic Model): $\text{logit}\ \text{pr}(Z_{n}=1\mid \mathbf{x}_{n})=f(\mathbf{x}_{n})+\epsilon_{n}^{z}$ with $\epsilon_{n}^{z}\overset{\text{i.i.d.}}{\sim} N(0,1)$.
    \item Model 3 (Nonlinear Selection Model): $Z_{n}=\mathbbm{1}\{f(\mathbf{x}_{n})>\epsilon_{n}^{z}\}$ with $\epsilon_{n}^{z}\overset{\text{i.i.d.}}{\sim} N(0,1)$.
\end{itemize}
For each unit $n$, we consider the following generating process for the potential outcome under control: $Y_{n}(0)=0.2x_{n1}^3+0.2|x_{n2}|+0.2x_{n3}^3+0.5|x_{n4}|+0.3x_{n5}+\epsilon_{n}^{y}$ with $\epsilon_{n}^{y}\overset{\text{i.i.d.}}{\sim} N(0,1)$, and the potential outcome under treatment: $Y_{n}(1)=Y_{n}(0)+(1+0.3x_{n1}+0.2x_{n3}^3)$ (i.e., heterogeneous treatment effects). After generating 400 pre-matching units in each scenario, we generate matched datasets using optimal full matching (\citealp{hansen2004full, rosenbaum2020design}) implemented via the widely used \texttt{R} package \texttt{optmatch} (\citealp{hansen2006optimal}). Each model uses two matching regimes: Model 1 applies optimal full matching with mild and strict propensity score calipers, whereas Models 2 and 3 use optimal full matching with and without a caliper, respectively. For each simulation setting, we generate 1000 matched datasets that meet the commonly used covariate balance criteria--the absolute standardized mean differences are less than 0.2 for all the covariates (\citealp{rosenbaum2020design, pimentel2024covariate}). In Table \ref{tab: simulation IPPW}, we report the mean estimation bias, the mean length of the 95\% confidence interval, and the coverage rate of the 95\% confidence interval for the sample average treatment effect based on the following methods: the conventional post-matching randomization-based inference method for the sample average treatment effect based on $\widehat{\lambda}$ (\citealp{rosenbaum2002observational, fogarty2018mitigating}), and randomization-based inferences using the IPPW estimator $\widehat{\lambda}_{\diamond}$ and the oracle IPPW estimator $\widehat{\lambda}_{*}$. Because the core idea of the proposed IPPW method is to re-weight the post-matching finite-population difference-in-means estimator based on discrepancies in $p_{ij}$ due to inexact matching, =we also report results based on the finite-population weighting estimator directly applied to the pre-matching (unmatched) datasets (\citealp{rosenbaum1987model, mukerjee2018using}), which represents the conventional randomization-based inference for observational data without matching. We consider both the version using estimated propensity scores and the version using oracle propensity scores. For a fair comparison, we use consistent approaches to estimate the propensity scores and conduct inference across all estimators in the same model. In Model 1, the estimated propensity scores $\widehat{e}_{ij}$ are obtained using the commonly applied logistic regression method. For the IPPW estimator $\widehat{\lambda}_{\diamond}$, inference is conducted via the finite-population M-estimation approach described in Section 3.2. In Models 2 and 3, 
the estimated propensity scores $\widehat{e}_{ij}$ are obtained by the XGBoost method (\citealp{chen2016xgboost}). For the IPPW estimator $\widehat{\lambda}_{\diamond}$ under these two models, we use the plug-in variance estimator $S^2_{\diamond}(Q)$ for the simulation study. Importantly, under the full matching framework, the post-matching sample size is still $N=400$ and matches the pre-matching sample size, and the estimand $\lambda$ (the sample average treatment effect) is the same for all the methods.

\setcounter{table}{0}
\begin{table}[ht]
\caption{The mean values of the estimation bias, the confidence interval (CI) length, and the coverage rate of $95\%$ confidence intervals for the sample average treatment effect based on the conventional post-matching randomization-based inference (which ignores inexact matching), the conventional randomization-based inference without matching (based on the estimated propensity scores $\widehat{e}_{ij}$ and the oracle propensity scores $e_{ij}$), and the proposed IPPW method (based on the estimated post-matching probabilities $\widehat{p}_{ij}$ and the oracle post-matching probabilities $p_{ij}$).}
\centering
\small
\scalebox{0.88}{\begin{tabular}{ccccccc}
\toprule
\multirow{2}{*}{Model 1}&\multicolumn{3}{c}{With Mild Caliper}&\multicolumn{3}{c}{With Strict Caliper} \\
\cmidrule(rl){2-4} \cmidrule(rl){5-7} 
 & {Bias} & {CI Length} & {Coverage Rate} & {Bias} & {CI Length} & {Coverage Rate} \\
\midrule
Conventional (Post-Matching)  & 0.301 & 0.907 & 0.741 & 0.220 & 1.034 & 0.930 \\
Conventional (Without Matching) & 0.262 & 0.924 & 0.790 & 0.283 & 0.935 & 0.768 \\
IPPW (The Proposed Method) & 0.094 & 0.810 & 0.964 & 0.096 & 0.915 & 0.974 \\
Oracle Conventional (Without Matching) & -0.012 & 1.306 & 0.889 & 0.032 & 1.313 & 0.862 \\
Oracle IPPW (The Proposed Method) & 0.089 & 0.850 & 0.964 & 0.094 & 0.965 & 0.979 \\
\bottomrule
\multirow{2}{*}{Model 2}&\multicolumn{3}{c}{Without Caliper}&\multicolumn{3}{c}{With Caliper} \\
\cmidrule(rl){2-4} \cmidrule(rl){5-7} 
 & {Bias} & {CI Length} & {Coverage Rate} & {Bias} & {CI Length} & {Coverage Rate} \\
\midrule
Conventional (Post-Matching)  & 0.378 & 0.868 & 0.591 & 0.311 & 0.932 & 0.767 \\
Conventional (Without Matching) & 0.442 & 1.004 & 0.600 & 0.454 & 1.007 & 0.583 \\
IPPW (The Proposed Method) & 0.301 & 0.879 & 0.743 & 0.250 & 0.940 & 0.871 \\
Oracle Conventional (Without Matching) & 0.099 & 1.128 & 0.922 & 0.112 & 1.128 & 0.916 \\
Oracle IPPW (The Proposed Method) & 0.119 & 0.868 & 0.951 & 0.151 & 0.948 & 0.950 \\
\bottomrule
\multirow{2}{*}{Model 3}&\multicolumn{3}{c}{Without Caliper}&\multicolumn{3}{c}{With Caliper} \\
\cmidrule(rl){2-4} \cmidrule(rl){5-7} 
 & {Bias} & {CI Length} & {Coverage Rate} & {Bias} & {CI Length} & {Coverage Rate} \\
\midrule
Conventional (Post-Matching)  & 0.492 & 1.032 & 0.506 & 0.431 & 1.138 & 0.686 \\
Conventional (Without Matching) & 0.547 & 1.136 & 0.495 & 0.572 & 1.142 & 0.469 \\
IPPW (The Proposed Method) & 0.325 & 0.993 & 0.786 & 0.300 & 1.103 & 0.854 \\
Oracle Conventional (Without Matching) & 0.303 & 1.313 & 0.813 & 0.314 & 1.361 & 0.800 \\
Oracle IPPW (The Proposed Method) & 0.220 & 1.127 & 0.920 & 0.260 & 1.390 & 0.926 \\
\bottomrule

\end{tabular}}
\label{tab: simulation IPPW}
\end{table}

Table \ref{tab: simulation IPPW} delivers three important messages. First, the proposed IPPW estimator consistently outperforms the conventional randomization-based inference without matching (i.e., that based on the finite-population weighting estimator), both when the true propensity scores are used and when the propensity scores are estimated. This confirms the usefulness of matching (as well as our improved inference methods based on matching) under the considered simulation settings, which agrees with the previous literature on advocating matching as an effective and nonparametric data preprocessing step \citep{hansen2004full, guo2023statistical, pimentel2024covariate}. Second, the proposed IPPW method can evidently reduce estimation bias and improve coverage rate compared with the conventional post-matching randomization-based inference method commonly used in previous studies (\citealp{rosenbaum2002observational, fogarty2018mitigating}). Third, as expected, in nonparametric propensity score settings, the performance of the IPPW method based on $\widehat{\lambda}_{\diamond}$ is suboptimal to that based on the oracle form $\widehat{\lambda}_{*}$. This suggests that there is still space for improving the proposed IPPW method by improving the estimation and uncertainty quantification of $p_{ij}$ in nonparametric propensity score settings, which is a meaningful future research direction. 

To summarize, our simulation study confirms the following straightforward rationale: by incorporating the post-matching covariate imbalance information, the proposed IPPW method is promising to effectively reduce estimation bias and improve the coverage rate of the confidence interval for the sample average treatment effect. In contrast, the conventional post-matching randomization-based inference is typically suboptimal because it discards the covariate imbalance information after matching and simply treats the post-matching treatment assignments as uniform assignments. 
\setcounter{remark}{0}
\begin{remark}
Although the oracle conventional estimator without matching is unbiased for the sample average treatment effect, the oracle propensity scores can be highly extreme (i.e., close to zero or one) in our simulations. These extreme values may induce substantial finite-sample bias and/or yield non-informative confidence intervals. To mitigate this issue, we implement a regularization strategy to reduce the influence of extreme propensity scores. As a result, the regularized estimator is no longer exactly unbiased, which helps explain why the oracle conventional methods without matching do not achieve the 0.95 coverage rate in our simulation studies (see Remark~\ref{rem: weighting} of Appendix E for details).
\end{remark}

\subsection*{C.2: Comparison of Variance Estimators for the IPPW Estimator with Estimated Propensity Scores}

Empirically, the true propensity scores are unknown, and therefore the true post-matching treatment assignment probabilities $p_{ij}$ are also unknown. Accordingly, we adopt a plug-in strategy and use the IPPW estimator $\widehat{\lambda}_{\diamond}$, as introduced in Section 3.2 of the main text. For inference on $\widehat{\lambda}_{\diamond}$, we propose the variance estimator $S_{\mathcal{M}}^2$ based on the finite-population M-estimation framework, which enjoys the formal validity guarantee stated in Theorem 4 of the main text. In practice, however, $S_{\mathcal{M}}^2$ may be unavailable or inconvenient to compute, for example, when the propensity scores are estimated using flexible nonparametric methods. In such settings, a simple alternative is to apply the plug-in variance estimator $S^2_{\diamond}(Q)$ from Section 3.2 in the main text. Although this plug-in variance estimator does not come with a general asymptotic validity guarantee, it performs well empirically. To assess its performance, we focus on Model 1 in the simulation study in Section C.1, where the propensity score model is correctly specified and parametric. We evaluate the bias of $\widehat{\lambda}_{\diamond}$ and compare two variance estimators for constructing $95\%$ confidence intervals: (i) the plug-in variance estimator $S^2_{\diamond}(Q)$, and (ii) the asymptotically valid estimator $S_{\mathcal{M}}^2$ derived from the finite-population M-estimation framework. As shown in Table \ref{tab: comparison}, the resulting coverage rates are very similar across the two approaches, suggesting that the plug-in variance estimator can provide a reasonable practical approximation in the considered simulation setting.

\begin{table}[ht]
\caption{The mean coverage rate of $95\%$ confidence intervals for the sample average treatment effect based on the IPPW estimator $\widehat{\lambda}_{\diamond}$ with two variance estimators: the empirical (plug-in) variance estimator and the M-estimation-based variance estimator. }
\centering
\small
\begin{tabular}{ccccc}
\toprule
&\multicolumn{1}{c}{Without Caliper}&\multicolumn{1}{c}{With Caliper} \\
\midrule
Plug-in Variance Estimator & 0.965 & 0.975 \\
M-Estimation-Based Variance Estimator & 0.964 & 0.974 \\
\bottomrule
\end{tabular}
\label{tab: comparison}
\end{table}

\subsection*{C.3: Comparative Performance of the IPPW and AIPPW Methods Based on Estimated Post-Matching Treatment Assignment Probabilities}

In Section 4 of the main text, we discuss incorporating outcome models into the IPPW estimator, leading to the augmented inverse post-matching probability weighting (AIPPW) estimator. Because the true propensity scores are not available in practice, they must be estimated from the sample. In this setting, adding outcome models can either reduce or increase estimation bias (i.e., mitigate or exacerbate the impact of not knowing the true propensity scores), depending on how well the outcome models predict the outcomes. To examine this idea, we conduct additional simulation studies. The simulation setup is the same as that in Section C.1. We compare two estimators: the proposed plug-in IPPW estimator $\widehat{\lambda}_{\diamond}$ and the plug-in AIPPW estimator $\widehat{\lambda}_{\ddagger}$. As in Section C.1, the estimated propensity scores $\widehat{e}_{ij}$ (which are used to calculate the estimated post-matching treatment assignment probabilities $\widehat{p}_{ij}$ using the formulas in Section 2 of the main text) are obtained using logistic regression in Model 1, whereas they are obtained using XGBoost in Models 2 and 3. In all three models, the outcome models used in the AIPPW estimator are fitted by linear regression using the sample: specifically, we fit the outcome model $g_1(\cdot)$ using the treated units and the outcome model $g_0(\cdot)$ using the control units. 

We report the mean estimation bias, the mean length of the 95\% confidence interval, and the coverage rate of the 95\% confidence interval for the sample average treatment effect in Table \ref{tab: simulation AIPPW}. The results in Table \ref{tab: simulation AIPPW} show three main patterns. First, AIPPW has a smaller or similar bias than IPPW in Models 1 and 2, but AIPPW has a larger bias than IPPW in Model 3. This is consistent with our discussion (in Section 4 of the main text) that adding outcome models can either mitigate or worsen estimation bias, depending on model performance. Second, across all three models, AIPPW has a smaller variance than the IPPW estimator, suggesting that the outcome models explain part of the variation in the potential outcomes, so the residuals are usually less dispersed than the original outcomes. Third, despite having smaller bias in some settings, the AIPPW estimator has lower coverage rates than the IPPW estimator in all three models. This indicates that reducing bias alone does not necessarily lead to better overall inferential performance, as confidence interval coverage also depends on variance estimation and the underlying sample-specific distribution of the estimator.

\begin{table}[ht]
\caption{The mean values of the estimation bias, the confidence interval (CI) length, and the coverage rate of $95\%$ confidence intervals for the sample average treatment effect using the proposed IPPW method and AIPPW methods (both based on the estimated post-matching probabilities $\widehat{p}_{ij}$).}
\centering
\small
\begin{tabular}{ccccccc}
\toprule
\multirow{2}{*}{Model 1}&\multicolumn{3}{c}{With Mild Caliper}&\multicolumn{3}{c}{With Strict Caliper} \\
\cmidrule(rl){2-4} \cmidrule(rl){5-7} 
 & {Bias} & {CI Length} & {Coverage Rate} & {Bias} & {CI Length} & {Coverage Rate} \\
\midrule
IPPW & 0.094 & 0.810 & 0.964 & 0.096 & 0.915 & 0.974 \\
AIPPW & -0.006 & 0.716 & 0.936 & 0.001 & 0.806 & 0.964 \\
\bottomrule
\multirow{2}{*}{Model 2}&\multicolumn{3}{c}{Without Caliper}&\multicolumn{3}{c}{With Caliper} \\
\cmidrule(rl){2-4} \cmidrule(rl){5-7} 
 & {Bias} & {CI Length} & {Coverage Rate} & {Bias} & {CI Length} & {Coverage Rate} \\
\midrule
IPPW & 0.301 & 0.879 & 0.743 & 0.250 & 0.940 & 0.871 \\
AIPPW & 0.262 & 0.745 & 0.723 & 0.252 & 0.789 & 0.769 \\
\bottomrule
\multirow{2}{*}{Model 3}&\multicolumn{3}{c}{Without Caliper}&\multicolumn{3}{c}{With Caliper} \\
\cmidrule(rl){2-4} \cmidrule(rl){5-7} 
 & {Bias} & {CI Length} & {Coverage Rate} & {Bias} & {CI Length} & {Coverage Rate} \\
\midrule
IPPW & 0.325 & 0.993 & 0.786 & 0.300 & 1.103 & 0.854 \\
AIPPW & 0.365 & 0.828 & 0.578 & 0.367 & 0.892 & 0.616 \\
\bottomrule

\end{tabular}
\label{tab: simulation AIPPW}
\end{table}

\subsection*{C.4: Simulation Studies Under the Instrumental Variable Setting}\label{sec: IV simulations}

We conduct a simulation study under the instrumental variable setting to assess the bias and coverage rate of the proposed bias-corrected Wald estimator compared with the classical post-matching Wald estimator. In the simulation study, we set the sample size $N=400$. All five covariates $\mathbf{x}_{n}=(x_{n1},\dots,x_{n5})$ for each unit $n$ ($n=1,\dots, N$) are generated from the same process as that considered in Section 4 in the main text: $(x_{n1},x_{n2},x_{n3})\overset{\text{i.i.d.}}{\sim} \mathcal{N}((0,0,0),\mathbf{I}_{3 \times 3})$, $x_{n4}\overset{\text{i.i.d.}}{\sim}\text{Laplace}(0, \sqrt{2}/2)$, and $x_{n5}\overset{\text{i.i.d.}}{\sim}\text{Laplace}(0, \sqrt{2}/2)$. We let $f_1(\mathbf{x}_{n})=0.1x_{n1}^3+0.3x_{n2}+0.2\log(x_{n3}^2)+0.1x_{n4}+0.2x_{n5}+|x_{n1}x_{n2}|+(x_{n3}x_{n4})^2+0.5(x_{n2}x_{n4})^2-2.5$, then the following two models are used to generate the instrumental variable $Z_{n}$ for each unit $n$:
\begin{itemize}
\item Model 1 (Nonlinear Logistic Model): $\text{logit}\ \text{pr}(Z_{n}=1\mid \mathbf{x}_{n})=f_1(\mathbf{x}_{n})+\epsilon^{z}_{n}$ with $\epsilon^{z}_{n}\overset{\text{i.i.d.}}{\sim} N(0,1)$.
    \item Model 2 (Nonlinear Selection Model): $Z_{n}=\mathbbm{1}\{f_1(\mathbf{x}_{n})>\epsilon^{z}_{n}\}$ with $\epsilon^{z}_{n}\overset{\text{i.i.d.}}{\sim} N(0,1)$.
\end{itemize}
Then, we set $f_2(\mathbf{x}_{n})=0.7x_{n1}+0.4\text{sin}(x_{n2})+0.4|x_{n3}|+0.6x_{n4}+0.1x_{n5}+0.3x_{n3}x_{n4}-1$. For each unit $n$, we consider the following generating process for the treatment variable: $D_n=\mathbbm{1}\{f_2(\mathbf{x}_{n})+u_{n}^d+(2+0.8x_{2n}^2)Z_{n}>\epsilon_{n}^{d}\}$ with $\epsilon_{n}^{d}\overset{\text{i.i.d.}}{\sim} N(0,1)$.
Next, we consider the following generating process for the outcome: $Y_{n}=f_3(\mathbf{x}_{n})+u^y_{n}+(1+0.1x_{n1}+0.3x_{n3}^2)D_n$, where $f_3(\mathbf{x}_{n})=0.4x_{n1}^2+0.1|x_{n2}|+0.1x_{n3}^2+0.2\text{cos}(x_{n4})+0.5\text{sin}(x_{n5})$. Here, the treatment-outcome unobserved covariates $u_{n}^d$ and $u_{n}^y$ of each unit $n$ follow the following joint distribution:
\[\begin{pmatrix}
u_{n}^d \\ 
u_{n}^y
\end{pmatrix}\overset{\text{i.i.d.}}\sim N\left(\begin{pmatrix}
    0 \\ 
    0
\end{pmatrix},\begin{pmatrix}
    1&0.8 \\ 
    0.8&1
\end{pmatrix}\right).
\]
After generating 400 unmatched units in each scenario, we use the widely used optimal full matching procedure with and without propensity score caliper to generate matched datasets (\citealp{hansen2006optimal, rosenbaum2020design}). For each simulation setting, we generate 1000 matched datasets where the absolute mean differences for all the observed covariates are less than the commonly used threshold 0.2 (\citealp{rosenbaum2020design, zhang2023social, pimentel2024covariate}). 

In Table \ref{tab: simulation}, we report the mean estimation bias, the mean $95\%$ confidence interval length, and the coverage rate of $95\%$ confidence intervals for the effect ratio using the classical post-matching Wald estimator $\widehat{\theta}$, the bias-corrected post-matching Wald estimator $\widehat{\theta}_{\diamond}$ based on the estimate $\widehat{p}_{ij}$, and the oracle bias-corrected post-matching Wald estimator $\widehat{\theta}_{*}$ based on the oracle $p_{ij}$. In the bias-corrected Wald estimator $\widehat{\theta}_{\diamond}$, each estimate $\widehat{p}_{ij}$ is obtained by plugging the propensity scores estimated by XGBoost (\citealp{chen2016xgboost}) into the corresponding formulas in the main text. Both $\widehat{\theta}_{\diamond}$ and its oracle form $\widehat{\theta}_{*}$ adopt a regularization step for handling extreme values of $\widehat{p}_{ij}$ (see Remark~\ref{rem: regularization} for details).

Table \ref{tab: simulation} suggests three key points. First, the bias-corrected Wald estimator can significantly reduce the estimation bias compared with the classical Wald estimator in the considered simulation settings. Second, the coverage rates of the confidence intervals reported by the bias-corrected Wald estimator are evidently higher than those reported by the classical Wald estimator. Third, as expected, the oracle bias-corrected Wald estimator $\widehat{\theta}_{*}$ outperforms the practical bias-corrected Wald estimator $\widehat{\theta}_{\diamond}$ in terms of estimation bias and coverage rate. This suggests that there is substantial potential to enhance the performance of the bias-corrected Wald estimator $\widehat{\theta}_{\diamond}$ by improving the estimation precision and uncertainty quantification of $\widehat{p}_{ij}$, pointing out a valuable direction for future research. All these findings are consistent with the findings from the simulation study in Section C.1.

\begin{table}[ht]
\caption{The mean estimation bias, the mean confidence interval (CI) length, and the coverage rate of $95\%$ confidence intervals for the effect ratio based on different types of post-matching Wald estimators: the classical type and the proposed bias-corrected type (based on the estimate $\widehat{p}_{ij}$ and the oracle $p_{ij}$). }
\centering
\small
\begin{tabular}{ccccccc}
\toprule
\multirow{2}{*}{Model 1}&\multicolumn{3}{c}{Without Caliper}&\multicolumn{3}{c}{With Caliper} \\
\cmidrule(rl){2-4} \cmidrule(rl){5-7} 
 & {Bias} & {CI Length} & {Coverage Rate} & {Bias} & {CI Length} & {Coverage Rate} \\
\midrule
Classical Wald & 0.380 & 1.046 & 0.689 & 0.337 & 1.112 & 0.755 \\
Bias-Corrected Wald & 0.343 & 1.094 & 0.763 & 0.298 & 1.161 & 0.805 \\
Bias-Corrected Wald (Oracle) & 0.258 & 1.198 & 0.865 & 0.249 & 1.243 & 0.856 \\
\bottomrule
\multirow{2}{*}{Model 2}&\multicolumn{3}{c}{Without Caliper}&\multicolumn{3}{c}{With Caliper} \\
\cmidrule(rl){2-4} \cmidrule(rl){5-7} 
 & {Bias} & {CI Length} & {Coverage Rate} & {Bias} & {CI Length} & {Coverage Rate} \\
\midrule
Classical Wald & 0.541 & 1.122 & 0.509 & 0.514 & 1.212 & 0.570 \\
Bias-Corrected Wald & 0.474 & 1.216 & 0.657 & 0.451 & 1.309 & 0.695 \\
Bias-Corrected Wald (Oracle) & 0.373 & 1.401 & 0.784 & 0.421 & 1.492 & 0.777 \\
\bottomrule

\end{tabular}
\label{tab: simulation}
\end{table}

\section*{Appendix D:  Data Application}\label{sec: Appendix B}

Nowadays, kidney disease has been found to have a high incidence rate all over the world, especially in low- and middle-income agricultural countries (\citealp{corsi2012kidney,kinyoki2021kidney}). To explore the effect of agricultural work on chronic kidney disease in men, we use the Zimbabwe 2015 Demographic and Health Surveys (DHS) data, in which there are 4688 individual records in total. Following \citet{yuzhou2022hemoglobin}, we use optimal full matching with propensity score caliper (\citealp{hansen2006optimal, rosenbaum2020design}) to form matched sets of agricultural workers (treated units) and non-agricultural workers (controls) based on the following eight covariates: age, body mass index, wealth index, educational level, marital status, religion, and cluster agricultural percentage. After matching, the 4688 individuals were optimally grouped into 987 matched sets. Table \ref{tab: balance table} reports the pre-matching and post-matching covariate balance for the Zimbabwe 2015 DHS data. From Table \ref{tab: balance table}, we can see that the absolute standardized mean differences in covariates between the treated and control units are all less than 0.08, which would be regarded as sufficient post-matching covariate balance based on the commonly used threshold 0.2 or 0.1 (\citealp{rosenbaum2020design, heng2021sharpening, zhang2023social, pimentel2024covariate}). 

\begin{table}[ht]
    \caption{Pre-matching and post-matching covariate balance of the Zimbabwe 2015 data, measured by the standardized difference in means (std.dif) between the treated and control units. }
    \centering
    \begin{tabular}[t]{ccc}
\toprule
 Covariate & Std.dif (Pre-Matching) & Std.dif (Post-Matching)\\
\midrule
 \multicolumn{1}{l}{Age} & \multicolumn{1}{l}{\,\,\:\:$0.137$} & \multicolumn{1}{l}{$-0.014$} \\
            \multicolumn{1}{l}{Body Mass Index} & \multicolumn{1}{l}{$ -0.047$} & \multicolumn{1}{l}{$-0.004$} \\
            \multicolumn{1}{l}{Wealth Index} & \multicolumn{1}{l}{$-0.095$} & \multicolumn{1}{l}{$-0.062$} \\
            \multicolumn{1}{l}{Educational Level} & \multicolumn{1}{l}{$-0.282$} & \multicolumn{1}{l}{$-0.074$} \\
            \multicolumn{1}{l}{Currently Married or Not} & \multicolumn{1}{l}{\,\,\:\:$0.132$} & \multicolumn{1}{l}{$-0.001$} \\
            \multicolumn{1}{l}{Christian or Not} & \multicolumn{1}{l}{$-0.112$} & \multicolumn{1}{l}{$-0.013$} \\
            \multicolumn{1}{l}{No Religion or Not} & \multicolumn{1}{l}{\,\,\:\:$0.058$} & \multicolumn{1}{l}{\,\,\:\:$0.013$} \\
            \multicolumn{1}{l}{Cluster Agricultural Percentage} & \multicolumn{1}{l}{\,\,\:\:$1.075$} & \multicolumn{1}{l}{\,\,\:\:$0.062$} \\         
\bottomrule
\end{tabular}
    \label{tab: balance table}
\end{table}

When evaluating kidney dysfunction, the hemoglobin level has been commonly used as an ancillary marker due to the lack of data on serum creatinine assessments in the population (\citealp{corsi2012kidney,kinyoki2021kidney}). In our data analysis, the outcome variable is the adjusted hemoglobin level considered in \citet{yuzhou2022hemoglobin}. Table \ref{tab: results} presents the estimates and 95\% confidence intervals of the sample average treatment effect (among the whole study population), reported by the following three methods: 
\begin{itemize}
    \item The pre-matching finite-population weighting estimator $\widehat{\lambda}_{W}$ (\citealp{rosenbaum1987model,mukerjee2018using}), which represents the conventional randomization-based inference method for the sample average treatment effect without matching. 
    
    \item The post-matching finite-population difference-in-means estimator $\widehat{\lambda}$ (\citealp{rosenbaum2002observational, fogarty2018mitigating}), which represents the conventional post-matching randomization-based inference method for the sample average treatment effect. 
    
    \item The proposed inverse post-matching probability weighting estimator $\widehat{\lambda}_{\diamond}$, as a bias-corrected post-matching randomization-based inference method for the sample average treatment effect. 
\end{itemize}
All three methods target the same finite-population estimand (namely, the sample average treatment effect among the whole population) and are implemented via the same procedure as that described in Section 4 in the main text. From the results reported in Table \ref{tab: results}, we can see that the pre-matching finite-population weighting method reports a near-zero estimate (i.e., null effect), with the 95\% confidence interval centering around zero. The routinely used post-matching finite-population difference-in-means estimator detects a moderate treatment effect ($=-0.056$), and the corresponding 95\% confidence interval centers on the negative effect side. In contrast, the proposed IPPW method detects a much larger treatment effect ($=-0.102$), with the 95\% confidence interval evidently moving towards the negative effect side. 

\begin{table}[ht]
    \caption{The results of the data analyses for the sample average treatment effect among the whole study population in the Zimbabwe 2015 DHS data. }
    \centering
    \begin{tabular}[t]{lrr}
\toprule
 Randomization-Based Inference Method & Estimate & 95\% Confidence Interval\\
\midrule
Conventional (Without Matching) & $0.003$ & $[-0.106,0.112]$ \\
            Conventional (Post-Matching) & $ -0.056$ & $[-0.195,0.083]$ \\
            IPPW (The Proposed Method) & $-0.102$ & $[-0.257,0.053]$ \\
            
\bottomrule
\end{tabular}
    \label{tab: results}
\end{table}

\section*{Appendix E: Additional Remarks}\label{sec: additional remarks}

\begin{remark}\label{rem: differences with model-based methods}
    Our work, which focuses on randomization-based (finite-population) causal inference under inexact matching, is significantly different from the relevant literature on bias correction for inexact matching in super-population causal inference (e.g., \citealp{rubin1973use, ho2007matching, abadie2011bias, guo2023statistical}). First, concerning the target estimands, the relevant literature in super-population causal inference focuses on either super-population average treatment effects (or those on the treated) or some super-population constant effect. In contrast, our work focuses on finite-population average treatment effects. For detailed discussions of the advantages and limitations of finite-population versus super-population causal estimands and inferences, see \citet{imbens2015causal}, \citet{athey2017econometrics}, \citet{zhao2018randomization}, \citet{li2023randomization}, and \citet{ding2024first}. Second, concerning the sources of randomness needed for causal inference, the relevant literature in super-population causal inference needs to assume that the outcomes are i.i.d. realizations from some super-population data-generating process. In contrast, our work relies only on the randomness of treatment assignments and does not require any modeling or distributional assumptions on the outcome variable. Third, concerning the strategies for correcting for bias associated with inexact matching, the relevant literature in super-population causal inference proposes to use post-matching outcome regression to adjust for bias due to inexact matching. However, this strategy does not directly apply to randomization-based (finite-population) causal inference because the potential outcomes are treated as fixed values in randomization-based inference, and the only source of randomness is from the treatment variables instead of the post-treatment outcomes. To address this gap, we propose a different strategy for bias correction for inexact matching, namely the inverse post-matching probability weighting method, which does not rely on any distributional assumptions on the outcome and can facilitate randomization-based causal inference.
\end{remark}

\begin{remark}\label{rem: novelty}
The IPPW estimator was originally proposed in \citet{zhu2023bias}, a preliminary and earlier version of the current manuscript, which focused on estimation instead of inference. In Section A.1 of \citet{pimentel2024covariate}, Pimentel \& Huang derived the variance estimator for the IPPW estimator under the \textit{constant treatment effect model}. Therefore, the inference methods proposed in \citet{pimentel2024covariate} focus on the constant treatment effect model and its extensions (i.e., Fisher's sharp null) and cannot be applied to study the sample average treatment effect (i.e., Neyman's weak null). This motivated us to derive new inference methods based on the IPPW estimator in the newest version of our manuscript. 
\end{remark}

\begin{remark}\label{rem: motivation}
   In the randomization-based inference literature, there are two major reasons behind the importance of developing new methods for handling inexact matching in the average treatment effect (i.e., Neyman's weak null) case. First, as mentioned in the main text, the existing randomization-based inference methods that account for inexact matching primarily focus on Fisher's sharp null, such as the constant treatment effect model and its parametric extensions (\citealp{rosenbaum1988permutation, pimentel2024covariate}). In practice, those parametric treatment effect models may not hold. In contrast, the average treatment effect is always well-defined and is immune to model misspecification. Second, when studying Fisher's sharp null, an alternative strategy to bias-corrected randomization-based inference is to discard matched sets with unsatisfactory covariate balance and only use a subset of matched sets with sufficient covariate balance (\citealp{rosenbaum2012optimal, visconti2018handling}). This is because the target causal estimands in Fisher's sharp null (e.g., those in a parametric treatment effect model) typically do not change with the selection/trimming of matched sets. However, for studying average treatment effects (i.e., Neyman's weak null), the aforementioned trimming strategy may not work as it will change the original causal estimand (e.g., the sample average treatment effect among the whole study population), and developing new inference methods that can adjust for inexact matching is perhaps the only sensible option.
\end{remark}

\begin{remark}
    In this paper, we focus on the sample average treatment effect and its extensions (e.g., the effect ratio) after applying full matching, which is universally interpretable across all inference methods under either exact or inexact matching. If we use other matching methods that trimmed samples (e.g., pair matching), the interpretation of the post-matching sample average treatment effect may depend on the covariate balance and inference methods. For example, after optimal pair matching, we have $\lambda=(2I)^{-1}\sum_{i=1}^{I}\sum_{j=1}^{2}Y_{ij}(1)-Y_{ij}(0)$. Without loss of generality, we assume $Z_{i1}=1$ for each $i$. Under the ignorability assumption, if matching was exact for each pair, we have 
    \vspace{-0.5cm}
    \begin{align*}
        E[Y_{i1}(1)-Y_{i1}(0)\mid Z_{i1}=1, \mathbf{x}_{i1}] &= E[Y_{i1}(1)-Y_{i1}(0)\mid \mathbf{x}_{i1}]\\
        &=E[Y_{i2}(1)-Y_{i2}(0)\mid \mathbf{x}_{i2}]\\
        &=E[Y_{i2}(1)-Y_{i2}(0)\mid Z_{i2}=0, \mathbf{x}_{i2}].
    \end{align*}
\end{remark}
\vspace{-0.5cm}
Therefore, we have
\vspace{-0.3cm}
\begin{equation*}
   E\lambda = E_{\mathbf{X}_{T}}[E(\lambda\mid \mathbf{X}_{T})]=E_{\mathbf{X}_{T}}[E(Y_{i1}(1)-Y_{i1}(0)\mid \mathbf{x}_{i1}, Z_{i1}=1)]=E(Y(1)-Y(0) \mid Z=1).
   \vspace{-0.3cm}
\end{equation*}
However, if $\mathbf{x}_{i1}\neq \mathbf{x}_{i2}$, it may happen that
\vspace{-0.3cm}
 \begin{align*}
        E[Y_{i1}(1)-Y_{i1}(0)\mid Z_{i1}=1, \mathbf{x}_{i1}] \neq E[Y_{i2}(1)-Y_{i2}(0)\mid Z_{i2}=0, \mathbf{x}_{i2}].
    \vspace{-1cm}
    \end{align*}
In this case, we have 
\begin{align*}
   E\lambda&=\frac{1}{2}E_{\mathbf{X}_{T}}[E(Y_{i1}(1)-Y_{i1}(0)\mid \mathbf{x}_{i1}, Z_{i1}=1)]+\frac{1}{2}E_{\mathbf{X}_{C, \mathcal{M}}}[E(Y_{i2}(1)-Y_{i2}(0)\mid \mathbf{x}_{i2}, Z_{i2}=0)]\\
   &=\frac{1}{2}E(Y(1)-Y(0) \mid Z=1)+\frac{1}{2}E(Y(1)-Y(0) \mid Z=0, \mathcal{M}).
   \vspace{-0.6cm}
\end{align*}
\begin{remark}\label{rem: weighting}
   In the finite-population survey sampling literature and the finite-population causal inference literature, the classic finite-population weighting estimator (\citealp{rosenbaum1987model, mukerjee2018using}) takes the following form:
\begin{equation*}
    \widehat{\lambda}_{W}=\Bigg\{\sum_{n=1}^{N}\frac{Z_{n}}{\widehat{e}_n}\Bigg\}^{-1}\sum_{n=1}^{N}\frac{Y_{n}Z_{n}}{\widehat{e}_n}-\Bigg\{\sum_{n=1}^{N}\frac{1-Z_{n}}{1-\widehat{e}_n}\Bigg\}^{-1}\sum_{n=1}^{N}\frac{Y_n (1-Z_n)}{1-\widehat{e}_n},
\end{equation*} 
where each $Z_{n}$, $Y_{n}$, and $\widehat{e}_{n}$ denotes the observed treatment indicator, the observed outcome, and the estimated propensity score of unit $n$ among the pre-matching population, respectively ($n=1,\dots, N$). An essential difference between the finite-population weighting estimator $\widehat{\lambda}_{W}$ and our proposed IPPW estimator $\widehat{\lambda}_{\diamond}$ is that the $\widehat{\lambda}_{W}$ simply uses the estimated propensity scores $\widehat{e}_{n}$ for weighting, while the $\widehat{\lambda}_{\diamond}$ uses matching as a nonparametric data pre-processing step and then adopts the estimated \textit{post-matching} treatment assignment probabilities $\widehat{p}_{ij}$ for weighting. Therefore, the IPPW estimator can also be viewed as an extension of the classic finite-population weighting estimator from the finite-population sampling literature (e.g., \citealp{rosenbaum1987model, mukerjee2018using}) to matched observational studies. 

In Models 2 and 3 of the simulation study in Section C.1 and the data analysis in Section D, the estimated propensity scores $\widehat{e}_{n}$ used in $\widehat{\lambda}_{W}$ is obtained by the commonly used XGBoost method (\citealp{chen2016xgboost}), which are the same as those used for calculating the $\widehat{p}_{ij}$ involved in the IPPW estimator $\widehat{\lambda}_{\diamond}$. For constructing the confidence intervals based on $\widehat{\lambda}_{W}$, we use a commonly used sandwich estimator to calculate the variance of $\widehat{\lambda}_{W}$ (\citealp{jared2004variance}). In addition, to avoid extreme weights caused by $\widehat{e}_{n}$ close to 0 or 1, in both the simulation study and data application, we consider a regularized propensity score $\widehat{e}_{n}^{\text{reg}}=\widehat{e}_{n}\times \mathbbm{1}\{\widehat{e}_{n}\in (0.1, 0.9)\}+ 0.1\times \mathbbm{1}\{\widehat{e}_{n}\leq 0.1\} + 0.9\times \mathbbm{1}\{\widehat{e}_{n}\geq 0.9\}$. This regularization step follows the commonly adopted trimming or truncation strategy in the weighting literature (\citealp{crump2009dealing, ma2020robust}), with the widely used regularization threshold of 0.1 (\citealp{crump2009dealing, sturmer2021propensity}). 

\end{remark}

\begin{remark}\label{rem: regularization}
If some post-matching probabilities $\widehat{p}_{ij}$ involved in $\widehat{\lambda}_{\diamond}$ are very close to 0 or 1, the weights $1/\widehat{p}_{ij}$ or $1/(1-\widehat{p}_{ij})$ will be excessively large, which can render non-informative confidence intervals and large finite-sample bias. Similar to the trimming or truncation strategy commonly adopted in the weighting literature \citep{crump2009dealing, ma2020robust}, we propose to handle extreme values of $\widehat{p}_{ij}$ by considering some regularized post-matching probability $\widehat{p}_{ij}^{\text{reg}}=\widehat{p}_{ij}\times \mathbbm{1}\{\min_{j}\widehat{p}_{ij} > \gamma \text{ and } \max_{j}\widehat{p}_{ij} < 1-\gamma\}+(m_{i}/n_{i})\times \mathbbm{1}\{\min_{j}\widehat{p}_{ij} < \gamma \text{ or } \max_{j}\widehat{p}_{ij} > 1-\gamma\}$, where $\gamma$ is some prespecified small number. For ensuring fair comparisons in both the simulation studies and data application, similar to the regularization step described in Remark~\ref{rem: weighting}, we set the regularization threshold $\gamma = 0.1$, which is a commonly used value in the relevant trimming or truncation literature (\citealp{crump2009dealing, sturmer2021propensity}).
\end{remark}

\begin{remark}\label{rem: classical wald}
The exact form of the classical post-matching Wald estimator $\widehat{\theta}$ (\citealp{baiocchi2010building, kang2016full}) for the effect ratio $\theta$ can be expressed as   
\begin{equation*}
    \widehat{\theta}=\frac{\sum_{i=1}^{I}\frac{n_i^2}{m_i(n_i-m_i)}\sum_{j=1}^{n_{i}}(Z_{ij}-\overline{Z}_{i})(Y_{ij}-\overline{Y}_{i})}{\sum_{i=1}^{I}\frac{n_i^2}{m_i(n_i-m_i)}\sum_{j=1}^{n_{i}}(Z_{ij}-\overline{Z}_{i})(D_{ij}-\overline{D}_{i})},
\end{equation*} 
in which we define $\overline{Z}_{i}=n_{i}^{-1}\sum_{j=1}^{n_{i}}Z_{ij}(1)-Z_{ij}(0)$, $\overline{Y}_{i}=n_{i}^{-1}\sum_{j=1}^{n_{i}}Y_{ij}(1)-Y_{ij}(0)$, and $\overline{D}_{i}=n_{i}^{-1}\sum_{j=1}^{n_{i}}D_{ij}(1)-D_{ij}(0)$. In the simulation study in Section C.4, we use the classical variance estimator for $\widehat{\theta}$ proposed in \citet{baiocchi2010building} and \citet{kang2016full} to construct 95\% confidence sets for the effect ratio $\theta$.

\end{remark}

\begin{remark}\label{rem: bias-corrected wald}

In practice, the bias-corrected Wald estimator $\widehat{\theta}_{\diamond}$ can be obtained by plugging $\widehat{p}_{ij}$ in the oracle bias-corrected Wald estimator $\widehat{\theta}_{*}$:
\begin{equation*}
    \widehat{\theta}_{\diamond}=\frac{\sum_{i=1}^{I}\sum_{j=1}^{n_{i}}\frac{1}{\widehat{p}_{ij}(1-{\widehat{p}_{ij}})}Y_{ij}(Z_{ij}-\widehat{p}_{ij})}{\sum_{i=1}^{I}\sum_{j=1}^{n_{i}}\frac{1}{\widehat{p}_{ij}(1-\widehat{p}_{ij})}D_{ij}(Z_{ij}-\widehat{p}_{ij})}.
\end{equation*}
Similarly, by plugging $\widehat{p}_{ij}$ in both the oracle test statistic $A_{*}(\theta_{0})$ and its variance estimator $V_{*}^2(\theta_{0})$, we can obtain the corresponding test statistic $A_{\diamond}(\theta_{0})$ and corresponding variance estimator $V_{\diamond}^2(\theta_{0})$. Correspondingly, the proposed $100(1-\alpha)\%$ confidence set can be expressed as $CS^{\theta}_{\diamond}=\{\theta_{0}: |A_{\diamond}(\theta_{0})/\sqrt{V^{2}_{\diamond}(\theta_{0})}|\leq\Phi^{-1}(1-\alpha/2)\}$, where the prespecified level $\alpha\in (0,1/2)$.
\end{remark}

\section*{Appendix F: Concluding Remark}\label{sec: Appendix G}
Our theoretical results, simulation studies, and real data analysis convey a consistent message: even when a matched dataset appears balanced in covariates by standard practices, residual post-matching imbalance can still introduce significant bias in randomization-based estimation and inference for average treatment effects. To address this, we propose a bias-corrected randomization-based inference method--the inverse post-matching probability weighting (IPPW) method--suitable for inexactly matched observational studies. We further extend this method to examine the effect ratio, including the complier average treatment effect as a special case, in matched instrumental variable studies. A follow-up study (\citealp{frazier2024bias}) explores how to adjust the proposed method to handle continuous treatments in the pair-matching case. Unlike conventional randomization-based inference methods that overlook inexact matching, our proposed methods account for post-matching covariate imbalance, thereby reducing bias in randomization-based estimation and inference. Moreover, unlike the existing randomization-based inference methods for handling inexact matching, our proposed methods do not require assuming any treatment effect models and, therefore, can be used to study average treatment effects (i.e., Neyman's weak null).

\end{document}